\documentclass[11pt]{article}

\pdfoutput=1

\title{On the hardness of finding normal surfaces}
\date{\today}

\author{
 Benjamin A. Burton\\
 The University of Queensland\\
 \email{bab}{maths.uq.edu.au}
 \and
 Alexander He\\
 The University of Queensland\\
 \email{a.he}{uqconnect.edu.au}
}

\usepackage{amsmath, amssymb, amsthm, calc, multicol, caption, subcaption}
\usepackage{mathtools, tikz, anyfontsize}

\usepackage[shortlabels]{enumitem}
\setlist[itemize]{leftmargin=*, noitemsep}
\setlist[enumerate]{leftmargin=*, noitemsep}
\setlist[description]{leftmargin=*, labelwidth=*, noitemsep}

\usepackage[bottom]{footmisc}

\usepackage[a4paper, margin=2cm]{geometry}

\usepackage{hyperref}

\usetikzlibrary{arrows.meta}
\usetikzlibrary{decorations.markings}

\theoremstyle{plain}
\newtheorem{theorem}{Theorem}
\newtheorem{lemma}[theorem]{Lemma}
\newtheorem{proposition}[theorem]{Proposition}
\newtheorem{construction}[theorem]{Construction}
\newtheorem{observation}[theorem]{Observation}

\theoremstyle{definition}
\newtheorem{problem}[theorem]{Problem}

\newcommand{\Integer}{\mathbb{Z}}

\newcommand{\Real}{\mathbb{R}}

\newcommand{\rvec}[1]{\mathbf{#1}}
\newcommand{\mbold}[1]{\mbox{\boldmath$#1$}}

\newcommand{\sfrac}[3][1.5pt]{\tfrac{\hspace{#1}#2\hspace{#1}}{\hspace{#1}#3\hspace{#1}}}

\newcommand{\scap}[1]{{\normalfont\textsc{#1}}}

\DeclarePairedDelimiter\paren{(}{)}
\DeclarePairedDelimiter\pparen{\big(}{\big)}

\DeclarePairedDelimiter\curly{\{}{\}}

\tikzset{
	midarrow/.style n args={2}{
		postaction={
			decorate,
			decoration={
				markings,
				mark=at position #1 with {\arrow{#2}}
			}
		}
	}
}

\tikzset{
	outlined/.style n args={1}{
		preaction={draw, line width=#1, white}
	}
}

\newcommand{\whitearrow}{Latex[length=2mm,width=3mm,fill=white]}
\newcommand{\blackarrow}{Latex[length=2mm,width=3mm]}
\newcommand{\stealtharrow}{Stealth[length=3mm,width=3mm]}
\newcommand{\dirarrows}{\stealtharrow\stealtharrow\stealtharrow}

\makeatletter
\renewenvironment{proof}[1][\proofname] {\par\pushQED{\qed}\normalfont\topsep6\p@\@plus6\p@\relax\trivlist\item[\hskip\labelsep\itshape\bfseries#1\@addpunct{.}]\ignorespaces}{\popQED\endtrivlist\@endpefalse}
\makeatother

\newcommand{\email}[2]{\href{mailto:#1@#2}{\textsf{#1\hspace{1pt}$@$\hspace{1pt}#2}}}

\begin{document}

\maketitle

\begin{abstract}
For many fundamental problems in computational topology,
such as unknot recognition and $3$-sphere recognition,
the existence of a polynomial-time solution remains unknown.
A major algorithmic tool behind some of the best known algorithms for these problems is normal surface theory.
However, we currently have a poor understanding of
the computational complexity of problems in normal surface theory:
many such problems are still not known to have polynomial-time algorithms,
yet proofs of $\mathrm{NP}$-hardness also remain scarce.
We give three results that provide some insight on this front.

A number of modern normal surface theoretic algorithms depend critically on
the operation of finding a non-trivial normal sphere or disc in a $3$-dimensional triangulation.
We formulate an abstract problem that captures the algebraic and combinatorial aspects of this operation,
and show that this abstract problem is $\mathrm{NP}$-complete.
Assuming $\mathrm{P}\neq\mathrm{NP}$,
this result suggests that any polynomial-time procedure for finding a non-trivial normal sphere or disc
will need to exploit some geometric or topological intuition.

Another key operation, which applies to a much wider range of topological problems,
involves finding a vertex normal surface of a certain type.
We study two closely-related problems that can be solved using this operation.
For one of these problems, we give a simple alternative solution that runs in polynomial time;
for the other, we prove $\mathrm{NP}$-completeness.
\end{abstract}

\section{Introduction}\label{sec:intro}
Despite the central importance of topological problems like unknot recognition and $3$-sphere recognition,
it remains unknown whether these problems have polynomial-time solutions.
Since unknot recognition has been shown to be in both $\mathrm{NP}$
(by Hass, Lagarias and Pippenger \cite{HassLagariasPippenger99})
and $\mathrm{co\text{-}NP}$ (by Lackenby \cite{Lackenby16}, based on ideas first presented by Agol \cite{Agol02}),
there is reason to suspect that there might indeed exist a polynomial-time algorithm for unknot recognition.
The situation for $3$-sphere recognition is similar, though not quite as clear-cut:
Ivanov \cite{Ivanov08} and Schleimer \cite{Schleimer11} showed that $3$-sphere recognition is in $\mathrm{NP}$,
while Zentner \cite{Zentner18} showed that, assuming the Generalised Riemann Hypothesis,
$3$-sphere recognition is in $\mathrm{co\text{-}NP}$.

Although we do not yet have polynomial-time algorithms for topological problems like these, there has been significant progress in the development of algorithms that are both simple enough to have been implemented, and efficient enough to be useful in practice. Many of these practical algorithms are based on normal surface theory. For instance, normal surface theory forms the foundation for several important algorithms provided in the topological software package \texttt{Regina} \cite{Burton13Regina,Regina}, including algorithms for unknot recognition and $3$-sphere recognition.

The success of normal surface theory in producing practical algorithms is due in large part to a number of significant developments from the last two decades. Indeed, up until the early 2000s, the role of normal surface theory was essentially confined to proving decidability of topological problems. This is exemplified by Haken's famous 1961 paper \cite{Haken61}, which detailed the first ever algorithm for unknot recognition, making unknot recognition the first of several problems to be proven decidable using normal surface theory. In the four decades following Haken's pioneering work, although normal surface theory successfully produced algorithms for a number of topological problems, such algorithms were generally too intricate to be implemented, and too inefficient to be useful in practice.

Given that normal surface theory had been applied to so many problems,
finding a way to make these algorithms feasible to both implement and execute was an enormous achievement.
The breakthrough came in 2003, when Jaco and Rubinstein introduced their theory of \textbf{0-efficiency} \cite{JacoRubinstein03},
which drew on earlier unpublished work of Casson.
As part of this theory, Jaco and Rubinstein developed the \textbf{crushing procedure},
which provides a computationally efficient way to simplify a $3$-dimensional triangulation $\mathcal{T}$
by ``crushing away'' a non-trivial normal sphere or disc in $\mathcal{T}$.
This crushing procedure lies at the heart of many of the practical normal-surface-theoretic algorithms
that have been developed and implemented in the last two decades \cite{Burton13Regina,Burton14,BurtonOzlen12}.

For the crushing procedure to be useful, we need to have some algorithm to \emph{find} a non-trivial normal sphere or disc (or else show that no such surface exists). The traditional approach involves a technique known as ``vertex normal surface enumeration'', and relies on the following fact: if a triangulation contains any non-trivial normal spheres or discs, then at least one such surface must appear as a special type of normal surface known as a \textbf{vertex normal surface}. Since a triangulation always contains only finitely many vertex normal surfaces, it is enough to enumerate all vertex normal surfaces, and then check whether any of these surfaces is a non-trivial sphere or disc. Unfortunately, in the worst case, the number of vertex normal surfaces is exponential in the number of tetrahedra, which makes this traditional approach computationally expensive.

As yet, no polynomial-time algorithm for finding a non-trivial normal sphere or disc is known. However, there was some exciting progress in 2012, when the authors of \cite{BurtonOzlen12} introduced a ``tree traversal procedure'' for finding a non-trivial normal sphere or disc, which they used to give a new algorithm for unknot recognition. This unknot recognition algorithm is notable for the following two reasons.
\begin{itemize}
\item First, although this algorithm theoretically requires exponential time, it is the first unknot recognition algorithm to exhibit \emph{polynomial-time} behaviour under exhaustive experimentation. In particular, when the tree traversal procedure is replaced with techniques based on vertex normal surface enumeration, the experimental running times revert to an exponential-time profile. This suggests that the tree traversal procedure is, in some sense, ``polynomial-time in practice''.
\item Second, the tree traversal procedure is the only component of this unknot recognition algorithm that requires exponential time; indeed, every other component of the algorithm is provably polynomial-time. For the optimistic reader, this has the following consequence: if we could design a polynomial-time procedure for finding a non-trivial normal sphere or disc, then we would immediately obtain a polynomial-time algorithm for unknot recognition.
\end{itemize}

With this second point in mind, we have a very strong motivation to investigate whether there exists a polynomial-time algorithm for the problem of finding a non-trivial normal sphere or disc. Given that such an algorithm has so far been elusive, we have equally strong motivation to investigate the possibility that this problem is actually $\mathrm{NP}$-hard. Notably, a proof of $\mathrm{NP}$-hardness would have the following practical implication: if we hope to design a polynomial-time algorithm for unknot recognition, then we will probably need to use different techniques than those used in \cite{BurtonOzlen12}. Such a proof would also be theoretically significant, since $\mathrm{NP}$-hardness results remain rare in normal surface theory.

Although we do not prove any results that directly relate to
the problem of finding a non-trivial normal sphere or disc,
we do provide some insight by studying three related problems.
For each of these related problems, we give a ``conclusive'' complexity result:
we either exhibit a polynomial-time algorithm, or prove $\mathrm{NP}$-completeness.
In a nutshell, our three main results can be seen as ``reference points'' to guide
our intuition about the complexity of other problems in normal surface theory.

Our first result is directly motivated by the techniques used in \cite{BurtonOzlen12}.
In particular, the tree traversal procedure from \cite{BurtonOzlen12}
essentially treats the task of finding a non-trivial normal sphere or disc as
an optimisation problem over a certain set of linear and combinatorial constraints.
In section \ref{sec:absNormConOpt}, we introduce a particular abstract formulation of these constraints,
and show that the corresponding optimisation problem is $\mathrm{NP}$-complete.
This has a significant implication for anyone hoping to design a polynomial-time algorithm for
finding a non-trivial normal sphere or disc:
assuming $\mathrm{P}\neq\mathrm{NP}$, our result suggests that
such an algorithm would need to exploit some geometric or topological ideas,
not just the linear algebra that is currently used in the tree traversal procedure.

Our other two results deal with concrete problems from normal surface theory (as opposed to the first, more abstract result). Specifically, in section \ref{sec:splitting}, we study two problems involving a type of normal surface called a \textbf{spanning central surface}. Just like the problem of finding a non-trivial normal sphere or disc, our spanning central surface problems can both be solved by enumerating vertex normal surfaces. Motivated by this similarity, we prove the following pair of contrasting results: one of our spanning central surface problems has a simple polynomial-time solution, while the other is $\mathrm{NP}$-complete. The upshot is that we have identified two closely-related problems that sit near the threshold between ``easy'' and ``hard'', one on each side.

\section{Preliminaries}\label{sec:prelims}

In computational $3$-manifold topology, we need a discrete way to represent $3$-manifolds.
Since simplicial complexes can often be very large,
we use a more flexible structure known as a \textbf{(generalised)\footnote{
For brevity, we usually drop the word ``generalised''.} triangulation};
a triangulation $\mathcal{T}$ consists of $n$ tetrahedra
that have been ``glued'' together by affinely identifying some or all of their $4n$ triangular faces in pairs.
Such triangulations are ``generalised'' because we allow two faces of the same tetrahedron to be identified.
Moreover, we allow multiple edges of a single tetrahedron to become identified
(as a result of the face identifications), and likewise with vertices.

We tacitly assume that our triangulations are connected. We also insist that our triangulations never contain any \textbf{invalid edges}; an edge $e$ is invalid if, as a result of the face identifications, $e$ has been identified with itself in reverse.

If the underlying topological space of a triangulation $\mathcal{T}$ is
actually a topological $3$-manifold, then we call $\mathcal{T}$ a \textbf{3-manifold triangulation}.
One way to determine whether $\mathcal{T}$ is a $3$-manifold triangulation is
to examine its \textbf{vertex links};
given a vertex $v$ of $\mathcal{T}$,
the link of $v$ is defined to be the boundary of a small regular neighbourhood of $v$.
A triangulation is a $3$-manifold triangulation if and only if
every vertex link is either a sphere (if the vertex is internal) or a disc (if the vertex is on the boundary).

Having introduced triangulations, we devote the remainder of this section to
introducing the aspects of normal surface theory that are most relevant for our purposes;
a more comprehensive overview can be found in \cite{HassLagariasPippenger99}.
A \textbf{normal surface} in a triangulation $\mathcal{T}$ is a surface which:
\begin{itemize}
\item is properly embedded in $\mathcal{T}$;
\item intersects each simplex in $\mathcal{T}$ transversely; and
\item intersects each tetrahedron $\Delta$ of $\mathcal{T}$ in a finite (and possibly empty) collection of discs,
where each disc is a curvilinear triangle or quadrilateral whose vertices lie on different edges of $\Delta$.
\end{itemize}
The curvilinear triangles and quadrilaterals are collectively known as \textbf{elementary discs}.
In the literature, elementary discs are also often called ``normal discs'';
however, we reserve the words ``normal disc'' to mean an entire normal surface that forms a disc.

Under a \textbf{normal isotopy}, which is defined to be an ambient isotopy that preserves every simplex in a given triangulation, the number of times a normal surface intersects each simplex can never change. Within each tetrahedron, the elementary discs get divided into seven equivalence classes under normal isotopy. Each equivalence class is called an \textbf{elementary disc type} \cite{HassLagariasPippenger99}. As illustrated in Figure \ref{fig:elemDiscTypes}, there are four \textbf{triangle types} and three \textbf{quadrilateral types} in each tetrahedron.

\begin{figure}[htbp]
\centering
	\begin{subfigure}[t]{0.24\textwidth}
	\centering
		\begin{tikzpicture}[scale=0.5]
		
		\draw[thick, dashed] (-3,2) -- (3,2);
		
		\fill[pink] (-0.75,0.5) -- (0,1.375) -- (0.75,0.5);
		\draw[dashed] (-0.75,0.5) -- (0.75,0.5);
		\draw (-0.75,0.5) -- (0,1.375) -- (0.75,0.5);
		
		\fill[pink] (-2.25,1.5) -- (-1.5,2) -- (-2.25,2.875);
		\draw[dashed] (-2.25,1.5) -- (-1.5,2) -- (-2.25,2.875);
		\draw (-2.25,1.5) -- (-2.25,2.875);
		
		\fill[pink] (2.25,1.5) -- (1.5,2) -- (2.25,2.875);
		\draw[dashed] (2.25,1.5) -- (1.5,2) -- (2.25,2.875);
		\draw (2.25,1.5) -- (2.25,2.875);
		
		\fill[pink] (-0.75,4.625) -- (0,4.125) -- (0.75,4.625);
		\draw[dashed] (-0.75,4.625) -- (0.75,4.625);
		\draw (-0.75,4.625) -- (0,4.125) -- (0.75,4.625);
		
		\draw[thick] (0,0) -- (-3,2) -- (0,5.5) -- (3,2) -- cycle;
		\draw[thick] (0,0) -- (0,5.5);
		
		\end{tikzpicture}
	\end{subfigure}
	\begin{subfigure}[t]{0.24\textwidth}
	\centering
		\begin{tikzpicture}[scale=0.5]
		
		\draw[thick, dashed] (-3,2) -- (3,2);
		
		\fill[pink] (1.5,1) -- (0,3.25) -- (-1.5,3.75) -- (-0.5,2);
		\draw[dashed] (-1.5,3.75) -- (-0.5,2) -- (1.5,1);
		\draw (1.5,1) -- (0,3.25) -- (-1.5,3.75);
		
		\draw[thick] (0,0) -- (-3,2);
		\draw[thick] (-3,2) -- (0,5.5);
		\draw[thick] (0,5.5) -- (3,2);
		\draw[thick] (3,2) -- (0,0);
		\draw[thick] (0,0) -- (0,5.5);
		
		\end{tikzpicture}
	\end{subfigure}
	\begin{subfigure}[t]{0.24\textwidth}
	\centering
		\begin{tikzpicture}[scale=0.5]
		
		\draw[thick, dashed] (-3,2) -- (3,2);
		
		\fill[pink] (-1.5,1) -- (0,3.25) -- (1.5,3.75) -- (0.5,2);
		\draw[dashed] (1.5,3.75) -- (0.5,2) -- (-1.5,1);
		\draw (-1.5,1) -- (0,3.25) -- (1.5,3.75);
		
		\draw[thick] (0,0) -- (-3,2);
		\draw[thick] (-3,2) -- (0,5.5);
		\draw[thick] (0,5.5) -- (3,2);
		\draw[thick] (3,2) -- (0,0);
		\draw[thick] (0,0) -- (0,5.5);
		
		\end{tikzpicture}
	\end{subfigure}
	\begin{subfigure}[t]{0.24\textwidth}
	\centering
		\begin{tikzpicture}[scale=0.5]
		
		\draw[thick, dashed] (-3,2) -- (3,2);
		
		\fill[pink] (1.5,1) -- (1.5,3.75) -- (-1.5,3.75) -- (-1.5,1);
		\draw[dashed] (-1.5,1) -- (1.5,1);
		\draw[dashed] (-1.5,3.75) -- (1.5,3.75);
		\draw (1.5,1) -- (1.5,3.75);
		\draw (-1.5,1) -- (-1.5,3.75);
		
		\draw[thick] (0,0) -- (-3,2) -- (0,5.5) -- (3,2) -- cycle;
		\draw[thick] (0,0) -- (0,5.5);
		
		\end{tikzpicture}
	\end{subfigure}
\caption{The seven elementary disc types.}
\label{fig:elemDiscTypes}
\end{figure}
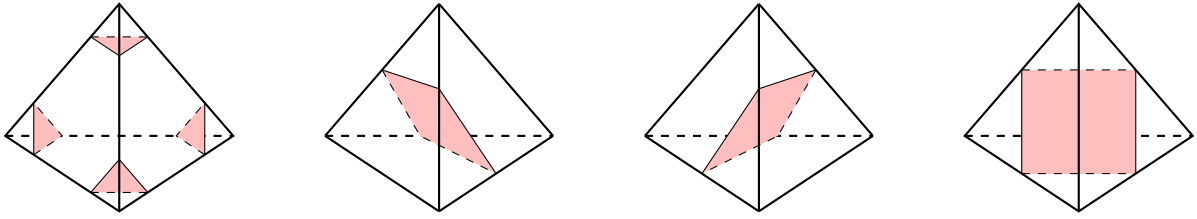

Let $\mathcal{T}$ be a triangulation with $n$ tetrahedra.
Any normal surface $S$ in $\mathcal{T}$ can be represented uniquely, up to normal isotopy,
by a vector $\rvec{v}(S)\in\Integer^{7n}$ that counts the number of elementary discs of each type in each tetrahedron.
The vector $\rvec{v}(S)$ is called the \textbf{standard vector representation} of $S$,
and its $7n$ integer coordinates are called \textbf{normal coordinates}.
The $4n$ coordinates that count triangles are called \textbf{triangle coordinates},
and the $3n$ coordinates that count quadrilaterals are called \textbf{quadrilateral coordinates}.

It is not immediately clear which points $\rvec{x}\in\Integer^{7n}$ actually represent normal surfaces. The most obvious necessary condition is the \textbf{non-negativity condition}, which requires that the coordinates of $\rvec{x}$ are all non-negative. It turns out that we only need two more necessary conditions, known as the matching equations and the quadrilateral constraints, to obtain a list of conditions which is also sufficient \cite{Haken61}. The following is a brief overview of these two conditions.
\begin{itemize}
\item For a collection of elementary discs to glue together to form a normal surface, they need to ``match up'' across pairs of identified triangular faces. An elementary disc always meets a triangular face $f$ of a tetrahedron in one of three possible types of \textbf{normal arcs}. Each normal arc type is ``parallel'' to one of the three edges of $f$, in the sense that the arc joins the \emph{other} two edges of $f$. Two normal arcs are of the same type if and only if they are related by a normal isotopy.

With this in mind, consider a triangular face $f_1$ of a tetrahedron $\Delta_1$, and let $e_1$ be one of the edges of $f_1$. Suppose $f_1$ is identified with another face $f_2$ of some tetrahedron $\Delta_2$. As a result of this face identification, $e_1$ becomes identified with an edge $e_2$ of $f_2$. The matching condition requires the number of normal arcs parallel to $e_1$ in $f_1$ to be equal to the number of normal arcs parallel to $e_2$ in $f_2$. To describe this constraint as a linear equation, we let:
	\begin{itemize}
	\item $t_1$ be the number of triangles in $\Delta_1$ giving rise to a normal arc parallel to $e_1$ in $f_1$;
	\item $q_1$ be the number of quadrilaterals in $\Delta_1$ giving rise to a normal arc parallel to $e_1$ in $f_1$;
	\item $t_2$ be the number of triangles in $\Delta_2$ giving rise to a normal arc parallel to $e_2$ in $f_2$; and
	\item $q_2$ be the number of quadrilaterals in $\Delta_2$ giving rise to a normal arc parallel to $e_2$ in $f_2$.
	\end{itemize}
Then, we require points $\rvec{x}\in\Integer^{7n}$ to satisfy the \textbf{matching equation}
\[
t_1+q_1=t_2+q_2.
\]
Figure \ref{fig:matchEqn} gives an example with $t_1=1$, $q_1=2$, $t_2=3$ and $q_2=0$, which clearly satisfies this matching equation.

Since there are $4n$ tetrahedron faces in total, up to $2n$ pairs of faces can be identified. Each such identification yields three matching equations, giving a total of at most $6n$ equations, with equality if and only if all $4n$ faces have been paired up.

\item The \textbf{quadrilateral constraints} require that each tetrahedron has at most one non-zero quadrilateral coordinate. This condition is necessary because in each tetrahedron, any two quadrilaterals of different types must always intersect. Such intersections need to be avoided because normal surfaces are (by definition) embedded surfaces.
\end{itemize}

\begin{figure}[htbp]
\centering
	\begin{tikzpicture}[scale=0.8]
	
	\draw[thick, dashed] (-3,2) -- (3,2);
	
	\fill[pink] (1.5,1) -- (0,3.25) -- (-1.5,3.75) -- (-0.5,2);
	\draw[dashed] (-1.5,3.75) -- (-0.5,2) -- (1.5,1);
	\draw (0,3.25) -- (-1.5,3.75);
	\draw[ultra thick, red] (1.5,1) -- (0,3.25);
	
	\fill[pink] (1.125,0.75) -- (0,2.32) -- (-2,3.1667) -- (-1.3,2);
	\draw[dashed] (-2,3.1667) -- (-1.3,2) -- (1.125,0.75);
	\draw (0,2.32) -- (-2,3.1667);
	\draw[ultra thick, red] (1.125,0.75) -- (0,2.32);
	
	\fill[pink] (-0.75,0.5) -- (0,1.5) -- (0.75,0.5);
	\draw[dashed] (-0.75,0.5) -- (0.75,0.5);
	\draw (-0.75,0.5) -- (0,1.5);
	\draw[ultra thick, red] (0,1.5) -- (0.75,0.5);
	
	\draw[thick] (0,0) -- (-3,2);
	\draw[thick] (-3,2) -- (0,5.5);
	\draw[ultra thick, line cap=round] (0,5.5) -- (3,2);
	\draw[ultra thick, line cap=round] (3,2) -- (0,0);
	\draw[ultra thick, line cap=round] (0,0) -- (0,5.5);
	\node[above right, inner sep=0pt] at (1.5,3.75) {$e_1$};

	\begin{scope}[shift={(8,0)}]
	\draw[thick, dashed] (3,2) -- (-3,2);
	
	\fill[pink] (-1.5,1) -- (0,3.25) -- (1.5,1);
	\draw[dashed] (-1.5,1) -- (1.5,1);
	\draw (1.5,1) -- (0,3.25);
	\draw[ultra thick, blue] (0,3.25) -- (-1.5,1);
	
	\fill[pink] (-1.125,0.75) -- (0,2.32) -- (1.125,0.75);
	\draw[dashed] (-1.125,0.75) -- (1.125,0.75);
	\draw (1.125,0.75) -- (0,2.32);
	\draw[ultra thick, blue] (0,2.32) -- (-1.125,0.75);
	
	\fill[pink] (-0.75,0.5) -- (0,1.5) -- (0.75,0.5);
	\draw[dashed] (-0.75,0.5) -- (0.75,0.5);
	\draw (0.75,0.5) -- (0,1.5);
	\draw[ultra thick, blue] (0,1.5) -- (-0.75,0.5);
	
	\draw[ultra thick, line cap=round] (0,0) -- (-3,2);
	\draw[ultra thick, line cap=round] (-3,2) -- (0,5.5);
	\draw[thick] (0,5.5) -- (3,2);
	\draw[thick] (3,2) -- (0,0);
	\draw[ultra thick, line cap=round] (0,0) -- (0,5.5);
	\node[above left, inner sep=0pt] at (-1.5,3.75) {$e_2$};
	\end{scope}

	\node at (1.25,2.75) {$\mbold{f_1}$};
	\node at (6.75,2.75) {$\mbold{f_2}$};
	
	\end{tikzpicture}
\caption{The matching equations assert that the number of red arcs should equal the number of blue arcs.}
\label{fig:matchEqn}
\end{figure}
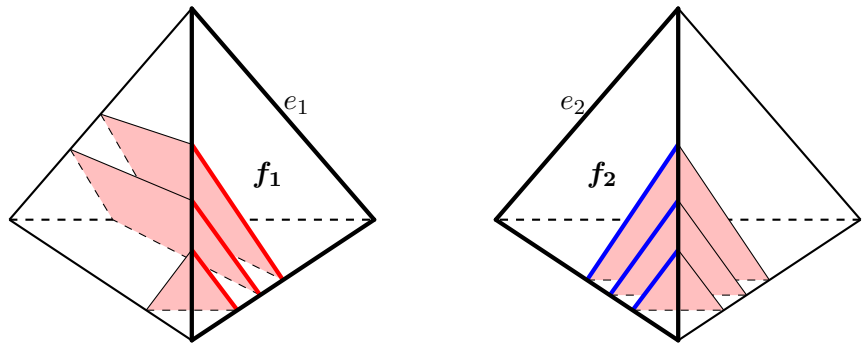

Any vector $\rvec{x}\in\Real^{7n}$ (including non-integer vectors) is called \textbf{admissible} if it simultaneously satisfies the non-negativity condition, the matching equations and the quadrilateral constraints. Haken showed that the admissible points in $\Integer^{7n}$ are precisely the points that represent normal surfaces.

\begin{theorem}[Haken]\label{thm:admissibility}
Let $\mathcal{T}$ be an $n$-tetrahedron triangulation. A vector $\rvec{x}\in\Integer^{7n}$ represents a normal surface in $\mathcal{T}$ (uniquely up to normal isotopy) if and only if $\rvec{x}$ is admissible \cite{Haken61,HassLagariasPippenger99}.
\end{theorem}

Notice that the matching equations, together with the non-negativity condition, give a collection of homogeneous linear equations and inequalities over $\Real^{7n}$. The set of solutions to these linear constraints forms a polyhedral cone $\mathcal{C}\subset\Real^{7n}$ known as the \textbf{standard solution cone}. We define a \textbf{vertex normal surface} to be a normal surface $S$ such that:
\begin{itemize}
\item $\rvec{v}(S)$ lies on an extreme ray of the cone $\mathcal{C}$; and
\item there is no $q\in(0,1)$ such that $q\rvec{v}(S)$ is an integral point in $\Real^{7n}$.
\end{itemize}
Vertex normal surfaces were first introduced by Jaco and Oertel in 1984 (but they used the name ``fundamental edge surface'') \cite{JacoOertel84}. It turns out that a normal surface $S$ is a \textit{vertex}  normal surface if and only if the integer multiples of $\rvec{v}(S)$ are the only integral points $\rvec{x},\rvec{y}\in\mathcal{C}$ that can satisfy an equation of the form $k\rvec{v}(S)=\rvec{x}+\rvec{y}$, for some positive integer $k$ \cite{JacoTollefson95}.

For any vertex $v$ of a triangulation $\mathcal{T}$, the link of $v$ can always be represented as a normal surface.
Such a surface is built from the elementary triangles that ``surround'' $v$ in the triangulation
(see Figure \ref{fig:vertLink}).
Since vertex-linking normal surfaces always exist, we call a normal surface \textbf{non-trivial} if
it does not consist entirely of vertex links.

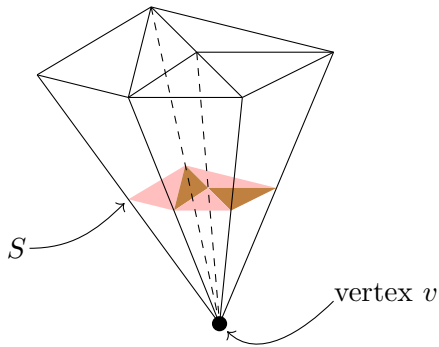
\begin{figure}[htbp]
\centering
	\begin{tikzpicture}[scale=1]
	\fill[pink] (-0.45,2.1) -- (-0.15,1.8) -- (0.75,1.8) -- cycle;
	\fill[pink] (-0.6,1.5) -- (-1.2,1.65) -- (-0.45,2.1) -- cycle;
	\fill[brown] (-0.45,2.1) -- (-0.6,1.5) -- (-0.15,1.8) -- cycle;
	\fill[brown] (0.15,1.5) -- (-0.15,1.8) -- (0.75,1.8) -- cycle;
	\fill[pink] (0.15,1.5) -- (-0.6,1.5) -- (-0.15,1.8) -- cycle;
	
	\draw[dashed] (-0.3,3.6) -- (0,0);
	\draw[dashed] (-0.9,4.2) -- (0,0);
	
	\draw (0,0) -- (0.3,3);
	\draw (0,0) -- (-1.2,3);
	\draw (0,0) -- (-2.4,3.3);
	\draw (0,0) -- (1.5,3.6);
	
	\draw (-0.9,4.2) -- (-2.4,3.3);
	\draw (-0.9,4.2) -- (-1.2,3);
	\draw (-0.9,4.2) -- (-0.3,3.6);
	\draw (-0.9,4.2) -- (1.5,3.6);
	\draw (-0.3,3.6) -- (-1.2,3);
	\draw (-0.3,3.6) -- (0.3,3);
	\draw (-0.3,3.6) -- (1.5,3.6);
	\draw (0.3,3) -- (-1.2,3);
	\draw (0.3,3) -- (1.5,3.6);
	\draw (-2.4,3.3) -- (-1.2,3);
	
	\fill[black] (0,0) circle[radius=0.1];
	\draw[->] (1.5,0.3) node[above right, inner sep=0pt]{vertex $v$} to[out=225,in=315] (0.1,-0.1);
	\draw[->] (-2.5,1) node[left, inner sep=1pt]{$S$} to[out=0,in=225] (-1.25,1.6);
	\end{tikzpicture}
\caption{Building a vertex-linking normal surface $S$ from triangles.}
\label{fig:vertLink}
\end{figure}

If $\mathcal{T}$ is a one-vertex triangulation, then there is of course only one vertex-linking normal surface;
this surface must have every triangle coordinate equal to $1$, and every quadrilateral coordinate equal to $0$.
Any \emph{connected} non-trivial normal surface in $\mathcal{T}$ must therefore have
at least one of its triangle coordinates set to $0$.
To see why, observe that if we have some non-trivial normal surface $S$ in $\mathcal{T}$ such that
every triangle coordinate is non-zero,
then one of the connected components of $S$ must be vertex-linking.
But since $S$ is non-trivial, $S$ must also include at least one quadrilateral,
and this quadrilateral must belong to a second component of $S$, so $S$ must be disconnected \cite{BurtonOzlen12}.

\section{Abstract normal constraint optimisation}\label{sec:absNormConOpt}

In section \ref{sec:intro},
we mentioned that the authors of \cite{BurtonOzlen12} introduced the first unknot recognition algorithm
to exhibit polynomial-time behaviour under exhaustive experimentation.
In particular, we noted that almost every component of this algorithm runs in polynomial time,
with the only exception being
the exponential-time tree traversal procedure for finding a non-trivial normal sphere or disc.
As we discuss in more detail shortly,
this tree traversal procedure essentially consists of two subroutines,
and it is really the first subroutine that is the exponential-time bottleneck.
Roughly speaking,
the bottleneck involves solving a certain combinatorial optimisation problem over normal coordinates.
Our main contribution in this section is to show that
an abstract version of this optimisation problem is $\mathrm{NP}$-complete.
This result suggests that if we want to design a polynomial-time algorithm for
finding a non-trivial normal sphere or disc,
then such an algorithm would probably need to exploit some geometric or topological ideas,
not just the linear algebra that is currently used in \cite{BurtonOzlen12}.

We begin by giving a precise formulation of the aforementioned optimisation problem over normal coordinates, as this informs how we formulate our abstract problem. Since this optimisation problem would be somewhat unmotivated if we simply presented it without further context, we take a more circuitous route. Specifically, we give a rough outline of the unknot recognition algorithm from \cite{BurtonOzlen12}, with a particular focus on describing how this algorithm boils down to solving the optimisation problem in question.
\begin{itemize}
\item The input for unknot recognition is a knot $K$ embedded in the topological $3$-sphere $S^3$.
To apply normal surface theory, we convert the knot $K$ into a $3$-manifold known as the \textbf{knot exterior};
in essence, the knot exterior $\overline{K}$ is the manifold that is left over after ``drilling out'' $K$ from $S^3$.
The first step of the algorithm in \cite{BurtonOzlen12} is essentially to convert the knot $K$ into
a one-vertex triangulation $\mathcal{T}$ of $\overline{K}$.
\item It turns out that the unknottedness of $K$ can be certified by the existence of a particular normal disc in $\mathcal{T}$. Using the crushing procedure that we briefly mentioned in section \ref{sec:intro}, the problem of determining whether $\mathcal{T}$ contains such a disc can essentially be boiled down to the problem of determining whether $\mathcal{T}$ contains any non-trivial normal sphere or disc. (This is a dramatic oversimplification; see \cite{BurtonOzlen12} for a comprehensive explanation of how this works.)
\item As we have already mentioned, the authors of \cite{BurtonOzlen12} use the tree traversal procedure to find a non-trivial normal sphere or disc in $\mathcal{T}$ (or else show that no such surface exists). This procedure basically consists of two subroutines.
	\begin{enumerate}[(1)]
	\item For each choice of triangle coordinate $t$ in $\mathcal{T}$, the first subroutine searches for a (possibly disconnected) normal surface $S$ in $\mathcal{T}$ that satisfies the following two conditions.
		\begin{enumerate}[(i)]
		\item The surface $S$ has positive Euler characteristic.
		\item The triangle coordinate $t$ is set to $0$ in $S$.
		\end{enumerate}
	If such a surface $S$ exists, then at least one of its connected components must have positive Euler characteristic. 
	\item If subroutine (1) finds an appropriate surface $S$ (for some choice of triangle coordinate $t$), the second subroutine deletes ``unwanted'' elementary discs from $S$, until the remaining elementary discs form a connected normal surface $S'$ with positive Euler characteristic. Since $\overline{K}\subset S^3$, and since the projective plane cannot be embedded in $S^3$, the surface $S'$ must either be a sphere or a disc. Moreover, since $\mathcal{T}$ is a one-vertex triangulation, and since the triangle coordinate $t$ is set to $0$ in $S'$, we see that $S'$ must be a \emph{non-trivial} normal sphere or disc.
	\end{enumerate}
\end{itemize}

As mentioned earlier,
subroutine (1) of the tree traversal procedure is the sole source of
the exponential running time of the unknot recognition algorithm from \cite{BurtonOzlen12}.
This subroutine essentially treats the Euler characteristic as
a \emph{linear} objective function that needs to be maximised.
This approach relies on the following fact:
given an $n$-tetrahedron triangulation $\mathcal{T}$,
there exists a homogeneous linear function $\chi:\Integer^{7n}\to\Integer$
such that if $S$ is a normal surface in $\mathcal{T}$,
then $\chi\pparen{\rvec{v}(S)}$ equals the Euler characteristic of $S$.
See \cite{BurtonOzlen12} for one possible choice of $\chi$.

Recall from section \ref{sec:prelims} that the normal surfaces in a triangulation correspond precisely to the normal coordinate vectors that are admissible. Thus, subroutine (1) boils down to repeatedly solving the following combinatorial optimisation problem, once for each choice of triangle coordinate $t$ (for a total of at most $4n$ repetitions, where $n$ is the number of tetrahedra in the input triangulation).

\begin{problem}\label{prbm:oneVertFixedTri}
\hspace{2em}
\begin{description}[nosep]
\item[\texttt{INSTANCE:}] An $n$-tetrahedron one-vertex triangulation $\mathcal{T}$, and a fixed triangle coordinate $t$ from one of the $n$ tetrahedra.
\item[\texttt{QUESTION:}] Does there exist a vector $\rvec{x}\in\Integer^{7n}$ such that $\chi(\rvec{x})>0$, subject to the constraints that $\rvec{x}$ is admissible and $t$ is set to $0$?
\end{description}
\end{problem}

We emphasise again that solving Problem \ref{prbm:oneVertFixedTri} is the critical exponential-time bottleneck for the unknot recognition algorithm from \cite{BurtonOzlen12}. If we could find a polynomial-time algorithm for this problem, then we would immediately get a polynomial-time algorithm for unknot recognition. So, with the goal of gaining some insight into the computational complexity of Problem \ref{prbm:oneVertFixedTri}, we devote the remainder of this section to:
\begin{itemize}
\item formulating an abstract problem that captures the algebraic and combinatorial aspects of Problem \ref{prbm:oneVertFixedTri}; and
\item proving that our abstract problem is $\mathrm{NP}$-hard.
\end{itemize}
As we hinted earlier, this result has a significant practical implication: assuming $\mathrm{P}\neq\mathrm{NP}$, a polynomial-time algorithm for Problem \ref{prbm:oneVertFixedTri} cannot rely solely on exploiting the algebraic and combinatorial aspects of the problem; instead, such an algorithm would probably need to exploit some geometric or topological ideas that are \emph{not} captured by our abstract problem.

From now until the end of this section, we will primarily be discussing ``abstract'' analogues of notions from normal surface theory, such as ``abstract normal coordinates'' and ``abstract matching equations''. For brevity, we will usually omit the word ``abstract''. So, to avoid confusion, whenever we need to discuss the ``concrete'' counterparts of these new abstract notions, we will explicitly use descriptions such as ``concrete normal coordinates'' and ``concrete matching equations''.

To formulate our algebraic abstraction, we start by defining an \textbf{abstract tetrahedron} to be a tuple
\[
T=(x_1,x_2,x_3,x_4,x_5,x_6,x_7)
\]
of seven non-negative integer variables called \textbf{(abstract) normal coordinates}. For $i\in\curly{1,2,3}$, $x_i$ is called an \textbf{(abstract) quadrilateral coordinate}. For $i\in\curly{4,5,6,7}$, $x_i$ is called an \textbf{(abstract) triangle coordinate}. The quadrilateral coordinates $x_1,x_2,x_3$ must satisfy the \textbf{(abstract) quadrilateral constraint}, which says that no two of these quadrilateral coordinates may be simultaneously non-zero.

Suppose we have $n$ abstract tetrahedra $T_1,\ldots,T_n$, where for each $k=1,\ldots,n$ we write
\[
T_k=\paren{x_{k,1},x_{k,2},\ldots,x_{k,7}}.
\]
All of these coordinates together determine a vector
\[
\rvec{x} = (x_{1,1},\ldots,x_{1,7};\;x_{2,1},\ldots,x_{2,7};\;\ldots\ldots;\;x_{n,1},\ldots,x_{n,7}) \in\Integer^{7n}.
\]
Of course, not every point $\rvec{x}\in\Integer^{7n}$ determines a ``valid'' assignment of values to our normal coordinates. To get a valid assignment, we require every coordinate of $\rvec{x}$ to be non-negative, and we require the quadrilateral constraint to be satisfied in each $T_k$, $k=1,\ldots,n$.

An \textbf{(abstract) matching equation} is an equation of the form $q+t=q'+t'$, where $q$ and $q'$ are both quadrilateral coordinates, and $t$ and $t'$ are both triangle coordinates. We will call a collection $M$ of matching equations \textbf{compatible} if $M$ satisfies the following two conditions:
\begin{itemize}
\item every quadrilateral coordinate appears at most four times in the equations in $M$; and
\item every triangle coordinate appears at most three times in the equations in $M$.
\end{itemize}
Note that we consider an equation of the form $q+t=q+t'$ to contain \emph{two} occurrences of the quadrilateral coordinate $q$. Similarly, an equation of the form $q+t=q'+t$ contains \emph{two} occurrences of the triangle coordinate $t$.

This notion of compatibility is motivated by two simple geometric observations. First, since a quadrilateral gives rise to four normal arcs, a concrete quadrilateral coordinate can only be involved in at most four concrete matching equations. Similarly, since a triangle gives rise to three normal arcs, a concrete triangle coordinate can only be involved in at most three concrete matching equations.

Recall from section \ref{sec:prelims} that a triangulation yields at most $6n$ concrete matching equations, with equality precisely when all $4n$ tetrahedron faces have been paired up. We get an analogue of this property as an easy consequence of the compatibility conditions.

\begin{lemma}\label{lem:compatible6n}
Let $M$ be a compatible collection of abstract matching equations. Then $M$ contains at most $6n$ equations. Moreover, $M$ contains exactly $6n$ equations if and only if every triangle coordinate appears in exactly three equations and every quadrilateral coordinate appears in exactly four equations.
\end{lemma}

\begin{proof}
Each of the $4n$ triangle coordinates can appear at most three times, for a total of up to $12n$ appearances of triangle coordinates. Since every matching equation includes two appearances of triangle coordinates, $M$ can contain at most $6n$ equations. In particular, if every triangle coordinate appears exactly three times, then there are exactly $12n$ appearances of triangle coordinates, and hence there must be exactly $6n$ equations.

On the other hand, suppose $M$ contains exactly $6n$ equations. Since each equation contains two appearances of triangle coordinates, there are a total of $12n$ appearances of triangle coordinates. But each of the $4n$ triangle coordinates can only appear at most three times, so they must all appear exactly three times. Similarly, since each equation contains two appearances of quadrilateral coordinates, there are a total of $12n$ appearances of quadrilateral coordinates. So, since each of the $3n$ quadrilateral coordinates can only appear at most four times, they must all appear exactly four times.
\end{proof}

With this in mind, fix a compatible collection $M$ of matching equations. We say that a point $\rvec{x}\in\Integer^{7n}$ is \textbf{$M$-admissible} if:
\begin{itemize}
\item the coordinates of $\rvec{x}$ are all non-negative;
\item $\rvec{x}$ satisfies the quadrilateral constraints; and
\item $\rvec{x}$ satisfies every equation in $M$.
\end{itemize}
The notion of $M$-admissibility mirrors the notion of admissibility discussed in section \ref{sec:prelims}.

Recalling that the Euler characteristic of a normal surface $S$ can be expressed as a homogeneous linear function $\chi$ of the concrete normal coordinates of $S$, where the coefficients in $\chi$ are all integers, we introduce the following abstraction of Problem \ref{prbm:oneVertFixedTri}.

\begin{problem}[\scap{Abstract normal constraint optimisation}]\label{prbm:absNormConOpt}
\hspace{2em}
\begin{description}[nosep]
\item[\texttt{INSTANCE:}]$n$ abstract tetrahedra, a compatible collection $M$ of matching equations, a homogeneous linear function $\chi:\Integer^{7n}\to\Integer$ (which can only have integer coefficients), and a fixed triangle coordinate $t$ from one of the $n$ abstract tetrahedra.
\item[\texttt{QUESTION:}]Does there exist a vector $\rvec{x}\in\Integer^{7n}$ such that $\chi(\rvec{x})>0$, subject to the constraints that $\rvec{x}$ is $M$-admissible and $t$ is set to $0$?
\end{description}
\end{problem}

This problem turns out to be $\mathrm{NP}$-complete
(see Theorem \ref{thm:abstractNPhard} below).
Since this problem has been deliberately formulated to
mimic the key algebraic and combinatorial aspects of Problem \ref{prbm:oneVertFixedTri},
this result therefore sheds some light on the complexity of Problem \ref{prbm:oneVertFixedTri}.
Having said this, it is important to note that Problem \ref{prbm:absNormConOpt}
does not capture all of the geometric and topological restrictions that
are inherent to Problem \ref{prbm:oneVertFixedTri}.
Indeed, the $\mathrm{NP}$-hardness part of our proof relies on
being able to construct instances of Problem \ref{prbm:absNormConOpt} that
violate these geometric and topological restrictions.
Thus, our proof does not readily extend to a proof that Problem \ref{prbm:oneVertFixedTri} is $\mathrm{NP}$-hard.

In detail,
there are two main distinctions between Problems \ref{prbm:oneVertFixedTri} and \ref{prbm:absNormConOpt},
both of which we discuss below.
Given that a polynomial-time algorithm for Problem \ref{prbm:oneVertFixedTri} would
represent a major breakthrough in computational topology,
we also discuss how these distinctions should inform future efforts to design such an algorithm.
\begin{itemize}
\item First, given a triangulation $\mathcal{T}$, recall that the concrete matching equations come directly from the face identifications in $\mathcal{T}$. This means that a collection of concrete matching equations must always have a somewhat rigid combinatorial structure. Apart from our notion of compatibility, most of this structure is \emph{not} captured by our abstract matching equations.

For instance, given a concrete matching equation $q+t=q'+t'$, the coordinates $q$ and $t$ must both come from the same tetrahedron (and likewise for $q'$ and $t'$). This property is not required in our abstract matching equations. Indeed, our proof of $\mathrm{NP}$-hardness of Problem \ref{prbm:absNormConOpt} relies on having access to equations that violate this property; in particular, see equation (d) in the proof of Lemma \ref{lem:repVars}, as well as equations (d)--(g) in the proof of Lemma \ref{lem:quadUnif}. Concrete matching equations also always come in groups of three, since a face identification always yields three equations; this is another property that is not reflected in our abstract matching equations.

The key takeaway is that the hardness of Problem \ref{prbm:absNormConOpt} could well be due, in large part, to the fact that our abstract matching equations fail to capture most of the inherent geometric properties of concrete matching equations. Thus, exploiting such geometric properties could be crucial for anyone hoping to design a polynomial-time algorithm for Problem \ref{prbm:oneVertFixedTri}.

\item Second, recall that in Problem \ref{prbm:oneVertFixedTri}, the function $\chi$ carries important topological information, because for any normal surface $S$, $\chi\pparen{\rvec{v}(S)}$ must equal the Euler characteristic of $S$. In contrast, Problem \ref{prbm:absNormConOpt} only requires $\chi$ to be \emph{some} homogeneous linear function with integer coefficients, and there is no need for this function to have any kind of topological significance.

This distinction could well be crucial to the hardness of Problem \ref{prbm:absNormConOpt}, as our proof of $\mathrm{NP}$-hardness exploits the ability to choose an essentially arbitrary function $\chi$. This also suggests that if we want to design a polynomial-time algorithm for Problem \ref{prbm:oneVertFixedTri}, then it may not be enough to simply treat $\chi$ as a homogeneous linear objective function, as is done in \cite{BurtonOzlen12}; instead, we may need to exploit some more specific properties of the Euler characteristic.
\end{itemize}

With all this in mind,
we now discuss our strategy for proving that Problem \ref{prbm:absNormConOpt} is
$\mathrm{NP}$-hard (this is the more intricate half of our $\mathrm{NP}$-completeness proof).
Given that most of the constraints in this problem are
homogeneous linear constraints involving only integer coefficients,
it is worth noting that these constraints can be handled using techniques from linear programming
(rather than integer programming).
More specifically, we can first use linear programming techniques to find
a \emph{rational} solution $\rvec{q}$ to these linear constraints,
and then exploit homogeneity by scaling $\rvec{q}$ to obtain the desired integral solution $\rvec{x}$.
From this perspective,
the quadrilateral constraints stand out as playing a key role in the hardness of Problem \ref{prbm:absNormConOpt},
since these are the only constraints that are non-linear.
It is therefore unsurprising that the quadrilateral constraints lie at the heart of our proof of $\mathrm{NP}$-hardness.

More broadly, our proof strategy is to give a reduction from a variant of satisfiability called \scap{monotone one-in-three satisfiability}. This problem is well-known, but we give the following detailed formulation anyway, so as to establish our notation.

A \textbf{clause} is a triple $c=\paren{u_1,u_2,u_3}$ of distinct Boolean variables, and we say that $c$ is \textbf{satisfied} if exactly one of its three variables is true. Given a set $C=\curly{c_1,\ldots,c_n}$ of such clauses, let $V$ be the set of all variables that appear in $C$. A \textbf{truth assignment} for $V$ is a map $t:V\to\curly{0,1}$, where we interpret $t(v)=1$ to mean ``$v$ is true'', and $t(v)=0$ to mean ``$t$ is false''. A clause is said to be \textbf{satisfied} under $t$ if exactly one of its variables is true under $t$. A collection $C$ of clauses is \textbf{satisfiable} if there is a truth assignment that satisfies every clause of $C$ in this way.

\begin{problem}[\scap{Monotone one-in-three satisfiability}]\label{prbm:m1in3SAT}
\hspace{2em}
\begin{description}[nosep]
\item[\texttt{INSTANCE:}]A collection $C=\curly{c_1,\ldots,c_n}$ of clauses (where each clause $c_k$ is a triple $\paren{u_{k,1},u_{k,2},u_{k,3}}$ of distinct variables from $V=\curly{v_1,\ldots,v_m}$).
\item[\texttt{QUESTION:}]Is $C$ satisfiable?
\end{description}
\end{problem}

Problem \ref{prbm:m1in3SAT} was proven $\mathrm{NP}$-complete by Schaefer in 1978, as an application of his general \emph{dichotomy theorem for satisfiability} \cite{Schaefer78}. Note that Schaefer simply called this problem ``one-in-three satisfiability''; we include the specifier ``monotone'' to emphasise that the clauses never contain any negated variables.

As mentioned above, we prove that Problem \ref{prbm:absNormConOpt} is $\mathrm{NP}$-hard by giving a polynomial reduction from Problem \ref{prbm:m1in3SAT}. To this end, suppose we are given any collection of clauses $C=\curly{c_1,\ldots,c_n}$, where each clause $c_k$ is a triple $\paren{u_{k,1},u_{k,2},u_{k,3}}$ of distinct Boolean variables. The main idea of our reduction is to represent each clause $c_k$ as an abstract tetrahedron $T_k=\paren{x_{k,i}}_{i=1}^7$, where for each $i=1,2,3$ we consider the variable $u_{k,i}$ to be ``true'' if and only if the quadrilateral coordinate $x_{k,i}$ is non-zero. For this to work, our construction needs to enforce the following two conditions.
\begin{enumerate}[(1)]
\item If $u_{k,i}$ and $u_{\ell,j}$ are two occurrences of the same variable, then the corresponding quadrilateral coordinates must be equal.
\item In each $T_k$, at least one of the three quadrilateral coordinates $x_{k,1},x_{k,2},x_{k,3}$ must be non-zero.
\end{enumerate}
Recall that each $T_k$ can only have at most one of its quadrilateral coordinates being non-zero, due to the quadrilateral constraints. Together with condition (2), this will force exactly one quadrilateral coordinate to be non-zero in each abstract tetrahedron $T_k$, which corresponds to the requirement that exactly one variable is true in each clause $c_k$.

To enforce condition (1), we introduce $n-1$ abstract tetrahedra
\[
S_k=\pparen{w_{k,i}}_{i=1}^7, \quad k\in\{1,\ldots,n-1\}.
\]
We then construct a particular collection $M_1$ of matching equations, such that if these equations are satisfied, then we must have $x_{k,i}=x_{\ell,j}$ for all $k,\ell\in\curly{1,\ldots,n}$ and $i,j\in\curly{1,2,3}$ such that $u_{k,i}$ and $u_{\ell,j}$ are two occurrences of the same variable. The details of this construction are captured in Lemma \ref{lem:repVars}.

To enforce condition (2), we introduce $n-1$ abstract tetrahedra
\[
U_k=\pparen{y_{k,i}}_{i=1}^7, \quad k\in\{1,\ldots,n-1\}.
\]
We then construct a particular collection $M_2$ of matching equations, such that if these equations are satisfied, then we must have
\[
x_{1,1}+x_{1,2}+x_{1,3} = x_{2,1}+x_{2,2}+x_{2,3} = \cdots = x_{n,1}+x_{n,2}+x_{n,3}.
\]
This constraint, combined with a carefully chosen function $\chi$, will be enough to enforce condition (2). The details of this construction are captured in Lemma \ref{lem:quadUnif}.

With these ideas in mind, we now present the proofs of Lemmas \ref{lem:repVars} and \ref{lem:quadUnif}.
We then finish this section by applying these two lemmas to prove that
\scap{abstract normal constraint optimisation} (Problem \ref{prbm:absNormConOpt})
is $\mathrm{NP}$-complete.

\begin{lemma}\label{lem:repVars}
Given:
\begin{itemize}
\item a collection $C=\curly{c_1,\ldots,c_n}$ of clauses (as in Problem \ref{prbm:m1in3SAT});
\item $n$ abstract tetrahedra $T_k=\paren{x_{k,i}}_{i=1}^7$, $k\in\{1,\ldots,n\}$; and
\item $n-1$ abstract tetrahedra $S_k=\paren{w_{k,i}}_{i=1}^7$, $k\in\{1,\ldots,n-1\}$;
\end{itemize}
we can construct a collection $M_1$ of matching equations such that:
\begin{itemize}
\item the equations in $M_1$ are satisfied if and only if we have:
	\begin{itemize}
	\item $w_{k,4}=w_{k,5}=w_{k,6}=w_{k,7}$ for all $k\in\curly{1,\ldots,n-1}$, and
	\item $x_{k,i}=x_{\ell,j}$ for all $k,\ell\in\curly{1,\ldots,n}$ and $i,j\in\curly{1,2,3}$ such that $u_{k,i}$ and $u_{\ell,j}$ are two occurrences of the same variable;
	\end{itemize}
\item for all $k\in\curly{1,\ldots,n}$ and $i\in\curly{1,2,3}$, the quadrilateral coordinate $x_{k,i}$ appears at most twice among the equations in $M_1$;
\item for all $k\in\curly{1,\ldots,n}$ and $i\in\curly{4,5,6,7}$, the triangle coordinate $x_{k,i}$ never appears among the equations in $M_1$;
\item for all $k\in\curly{1,\ldots,n-1}$ and $i\in\curly{1,2,3}$, the quadrilateral coordinate $w_{k,i}$ appears twice among the equations in $M_1$; and
\item for all $k\in\curly{1,\ldots,n-1}$ and $i\in\curly{4,5,6,7}$, the triangle coordinate $w_{k,i}$ appears at most three times among the equations in $M_1$.
\end{itemize}
Moreover, this can be done in $O(n^2)$ time.
\end{lemma}

\begin{proof}
For each $k\in\{1,\ldots,n-1\}$, we add the matching equations
\begin{align*}
w_{k,1}+w_{k,4} &= w_{k,1}+w_{k,5}; \tag{a} \\
w_{k,2}+w_{k,4} &= w_{k,2}+w_{k,7}; \tag{b} \\
w_{k,3}+w_{k,6} &= w_{k,3}+w_{k,7} \tag{c}
\end{align*}
to $M_1$. Observe that these equations reduce to
\[
w_{k,4} = w_{k,5} = w_{k,6} = w_{k,7}.
\]

Now, given any fixed variable $v\in V$, let $u_{k_1,i_1},u_{k_2,i_2},\ldots,u_{k_m,i_m}$ denote all the occurrences of $v$. We would like to force
\[
x_{k_j,i_j}=x_{k_{j+1},i_{j+1}}
\]
for each $j\in\{1,\ldots,m-1\}$, because this would imply
\[
x_{k_1,i_1} = x_{k_2,i_2} = \cdots = x_{k_m,i_m}.
\]
Since the triangle coordinates $w_{k_j,i_j+3}$ and $w_{k_j,i_j+4}$ must be equal, we can do this by adding the equation
\[
x_{k_j,i_j} + w_{k_j,i_j+3} = x_{k_{j+1},i_{j+1}} + w_{k_j,i_j+4}
\]
to $M_1$, for each $j\in\{1,\ldots,m\}$. By doing the same thing for each variable $v\in V$, we can force $x_{k,i}=x_{\ell,j}$ for all $k,\ell\in\curly{1,\ldots,n}$ and $i,j\in\curly{1,2,3}$ such that $u_{k,i}$ and $u_{\ell,j}$ are two occurrences of the same variable. We describe a polynomial-time procedure to add all these equations to $M_1$.

For each fixed $u_{k,i}$, where $k\in\curly{1,\ldots,n-1}$ and $i\in\curly{1,2,3}$, we perform a search that finishes by adding at most one equation to $M_1$. More specifically, for $\ell\in\{k+1,\ldots,n\}$ and $j\in\{1,2,3\}$, we sequentially check each $u_{\ell,j}$ to see whether $u_{k,i}$ and $u_{\ell,j}$ are two occurrences of the same variable. If so, we stop searching, and add the matching equation
\[
x_{k,i}+w_{k,i+3} = x_{\ell,j}+w_{k,i+4} \tag{d}
\]
to $M_1$. After checking $O(n)$ coordinates, we either:
\begin{itemize}
\item find such a $u_{\ell,j}$, and impose the corresponding equation; or
\item conclude that no such $u_{\ell,j}$ exists.
\end{itemize}
We perform this search once for each $u_{k,i}$, which requires $O(n^2)$ steps in total. When this whole procedure is finished, we get a collection $M_1$ of matching equations that are satisfied if and only if we have:
\begin{itemize}
\item $w_{k,4}=w_{k,5}=w_{k,6}=w_{k,7}$ for all $k\in\curly{1,\ldots,n-1}$; and
\item $x_{k,i}=x_{\ell,j}$ for all $k,\ell\in\curly{1,\ldots,n}$ and $i,j\in\curly{1,2,3}$ such that $u_{k,i}$ and $u_{\ell,j}$ are two occurrences of the same variable.
\end{itemize}
Moreover, given any fixed $k\in\{1,\ldots,n\}$, observe that:
\begin{itemize}
\item for each $i\in\{1,2,3\}$, the quadrilateral coordinate $x_{k,i}$ appears at most twice among equations of type (d), and nowhere else in $M_1$; and
\item for each $i\in\{4,5,6,7\}$, the triangle coordinate $x_{k,i}$ never appears among the equations in $M_1$.
\end{itemize}
To see that:
\begin{itemize}
\item for all $k\in\curly{1,\ldots,n-1}$ and $i\in\curly{1,2,3}$, the quadrilateral coordinate $w_{k,i}$ appears twice among the equations in $M_1$; and
\item for all $k\in\curly{1,\ldots,n-1}$ and $i\in\curly{4,5,6,7}$, the triangle coordinate $w_{k,i}$ appears at most three times among the equations in $M_1$;
\end{itemize}
consider any fixed $k\in\{1,\ldots,n-1\}$, and observe that:
\begin{itemize}
\item the quadrilateral coordinate $w_{k,1}$ appears twice in an equation of type (a), and nowhere else in $M_1$;
\item the quadrilateral coordinate $w_{k,2}$ appears twice in an equation of type (b), and nowhere else in $M_1$;
\item the quadrilateral coordinate $w_{k,3}$ appears twice in an equation of type (c), and nowhere else in $M_1$;
\item the triangle coordinate $w_{k,4}$ appears once in an equation of type (a), once in an equation of type (b), possibly once more in an equation of type (d) (since $i+3=4$ when $i=1$), and nowhere else in $M_1$;
\item the triangle coordinate $w_{k,5}$ appears once in an equation of type (a), possibly twice more in equations of type (d) (since $i+3=5$ when $i=2$, and $i+4=5$ when $i=1$), and nowhere else in $M_1$;
\item the triangle coordinate $w_{k,6}$ appears once in an equation of type (c), possibly twice more in equations of type (d) (since $i+3=6$ when $i=3$, and $i+4=6$ when $i=2$), and nowhere else in $M_1$; and
\item the triangle coordinate $w_{k,7}$ appears once in an equation of type (b), once in an equation of type (c), possibly once more in an equation of type (d) (since $i+4=7$ when $i=3$), and nowhere else in $M_1$.\qedhere
\end{itemize}
\end{proof}

\newcommand{\QuadUnifEqns}[1]{\begin{align*}
y_{k,1}&=y_{k,2}=y_{k,5}=y_{k,6}=0, \tag{i} \\
x_{k,4} &= x_{k,3}, \tag{ii} \\
x_{k+1,6} &= x_{k+1,3}, \tag{iii} \\
x_{k,5} &= x_{k,2}+x_{k,3}, \tag{iv} \\
x_{k+1,7} &= x_{k+1,2}+x_{k+1,3}, \tag{v} \\
x_{k,1}+x_{k,2}+x_{k,3} &= x_{k+1,1}+x_{k+1,2}+x_{k+1,3}#1 \tag{vi}
\end{align*}}

\begin{lemma}\label{lem:quadUnif}
Given:
\begin{itemize}
\item $n$ abstract tetrahedra $T_k=\paren{x_{k,i}}_{i=1}^7$, $k\in\{1,\ldots,n\}$; and
\item $n-1$ abstract tetrahedra $U_k=\paren{y_{k,i}}_{i=1}^7$, $k\in\{1,\ldots,n-1\}$;
\end{itemize}
we can construct a collection $M_2$ of matching equations such that:
\begin{itemize}
\item if we impose $y_{1,5}=0$, then the equations in $M_2$ are satisfied if and only if for all $k\in\{1,\ldots,n-1\}$ we have
\QuadUnifEqns{;}
\item for all $k\in\curly{1,\ldots,n}$ and $i\in\curly{1,\ldots,7}$, the coordinate $x_{k,i}$ appears at most twice among the equations in $M_2$;
\item for all $k\in\curly{1,\ldots,n-1}$ and $i\in\curly{1,2,3}$, the quadrilateral coordinate $y_{k,i}$ appears at most four times among the equations in $M_2$; and
\item for all $k\in\curly{1,\ldots,n-1}$ and $i\in\curly{4,5,6,7}$, the triangle coordinate $y_{k,i}$ appears at most three times among the equations in $M_2$.
\end{itemize}
Moreover, this can be done in $O(n)$ time.
\end{lemma}

\begin{proof}
For each $k\in\{1,\ldots,n-2\}$, we add the matching equation
\[
y_{k,2}+y_{k,6} = y_{k+1,1}+y_{k+1,5} \tag{a}
\]
to $M_2$. For each $k\in\{1,\ldots,n-1\}$, we also add the following matching equations to $M_2$:
\begin{align*}
y_{k,1}+y_{k,4} &= y_{k,2}+y_{k,4}; \tag{b} \\
y_{k,3}+y_{k,5} &= y_{k,3}+y_{k,6}; \tag{c} \\
y_{k,1}+x_{k,4} &= x_{k,3}+y_{k,5}; \tag{d} \\
y_{k,1}+x_{k+1,6} &= x_{k+1,3}+y_{k,6}; \tag{e} \\
y_{k,2}+x_{k,5} &= x_{k,2}+x_{k,4}; \tag{f} \\
y_{k,2}+x_{k+1,7} &= x_{k+1,2}+x_{k+1,6}; \tag{g} \\
x_{k,1}+x_{k,5} &= x_{k+1,1}+x_{k+1,7}. \tag{h}
\end{align*}

We claim that if we impose $y_{1,5}=0$, then the equations in $M_2$ are satisfied if and only if equations (i)--(vi) are satisfied for all $k\in\{1,\ldots,n-1\}$. To see this, first observe that the equations of type (b) reduce to $y_{k,1}=y_{k,2}=0$, since the quadrilateral coordinates $y_{k,1}$ and $y_{k,2}$ cannot be simultaneously non-zero, due to the quadrilateral constraint. This means that the equations of type (a) reduce to $y_{k,6}=y_{k+1,5}$. Together with the equations of type (c), which reduce to $y_{k,5}=y_{k,6}$, we therefore have
\[
y_{1,5}=y_{1,6}=y_{2,5}=y_{2,6}=\cdots=y_{n-1,5}=y_{n-1,6}.
\]
So, by imposing $y_{1,5}=0$, we get $y_{k,5}=y_{k,6}=0$ for all $k\in\{1,\ldots,n-1\}$. So, we have
\[
y_{k,1}=y_{k,2}=y_{k,5}=y_{k,6}=0
\]
for each $k\in\{1,\ldots,n-1\}$. As a result, the equations of type (d), (e), (f), (g) and (h) reduce to:
\begin{align*}
x_{k,4} &= x_{k,3}; \\
x_{k+1,6} &= x_{k+1,3}; \\
x_{k,5} &= x_{k,2}+x_{k,4} = x_{k,2}+x_{k,3}; \\
x_{k+1,7} &= x_{k+1,2}+x_{k+1,6} = x_{k+1,2}+x_{k+1,3}; \\
x_{k,1}+x_{k,2}+x_{k,3} &= x_{k,1}+x_{k,5} = x_{k+1,1}+x_{k+1,7} = x_{k+1,1}+x_{k+1,2}+x_{k+1,3};
\end{align*}
for all $k\in\{1,\ldots,n-1\}$. Thus, we recover equations (i)--(vi).

To see that for all $k\in\curly{1,\ldots,n}$ and $i\in\curly{1,\ldots,7}$, the coordinate $x_{k,i}$ appears at most twice among the equations in $M_2$, consider any fixed $k\in\{1,\ldots,n\}$ and observe that:
\begin{itemize}
\item the quadrilateral coordinate $x_{k,1}$ appears at most twice among equations of type (h), and nowhere else in $M_2$;
\item the quadrilateral coordinate $x_{k,2}$ appears at most once in an equation of type (f), possibly once more in an equation of type (g), and nowhere else in $M_2$;
\item the quadrilateral coordinate $x_{k,3}$ appears at most once in an equation of type (d), possibly once more in an equation of type (e), and nowhere else in $M_2$;
\item the triangle coordinate $x_{k,4}$ appears at most once in an equation of type (d), possibly once more in an equation of type (f), and nowhere else in $M_2$;
\item the triangle coordinate $x_{k,5}$ appears at most once in an equation of type (f), possibly once more in an equation of type (h), and nowhere else in $M_2$;
\item the triangle coordinate $x_{k,6}$ appears at most once in an equation of type (e), possibly once more in an equation of type (g), and nowhere else in $M_2$; and
\item the triangle coordinate $x_{k,7}$ appears at most once in an equation of type (g), possibly once more in an equation of type (h), and nowhere else in $M_2$.
\end{itemize}
To see that:
\begin{itemize}
\item for all $k\in\curly{1,\ldots,n-1}$ and $i\in\curly{1,2,3}$, the quadrilateral coordinate $y_{k,i}$ appears at most four times among the equations in $M_2$; and
\item for all $k\in\curly{1,\ldots,n-1}$ and $i\in\curly{4,5,6,7}$, the triangle coordinate $y_{k,i}$ appears at most three times among the equations in $M_2$;
\end{itemize}
consider any $k\in\{1,\ldots,n-1\}$, and observe that:
\begin{itemize}
\item the quadrilateral coordinate $y_{k,1}$ only appears in equations of type (a), (b), (d) and (e), and only at most once for each type;
\item the quadrilateral coordinate $y_{k,2}$ only appears in equations of type (a), (b), (f) and (g), and only at most once for each type;
\item the quadrilateral coordinate $y_{k,3}$ appears twice in an equation of type (c), and nowhere else in $M_2$;
\item the triangle coordinate $y_{k,4}$ appears twice in an equation of type (b), and nowhere else in $M_2$;
\item the triangle coordinate $y_{k,5}$ only appears in equations of type (a), (c) and (d), and only at most once for each type;
\item the triangle coordinate $y_{k,6}$ only appears in equations of type (a), (c) and (e), and only at most once for each type; and
\item the triangle coordinate $y_{k,7}$ never appears in $M_2$.
\end{itemize}
Finally, note that constructing all the equations in $M_2$ requires $O(n)$ steps.
\end{proof}

\newcommand{\thmAbstractNPhard}{\scap{Abstract normal constraint optimisation}
(Problem \ref{prbm:absNormConOpt}) is $\mathrm{NP}$-complete.}
\begin{theorem}\label{thm:abstractNPhard}
\thmAbstractNPhard
\end{theorem}

\begin{proof}
To show that Problem \ref{prbm:absNormConOpt} is in $\mathrm{NP}$,
we use essentially the same ideas as those used by Hass, Lagarias and Pippenger \cite{HassLagariasPippenger99}
to show that unknot recognition is in $\mathrm{NP}$.
Given a proposed solution $\rvec{x}$ to an instance $I$ of Problem \ref{prbm:absNormConOpt},
it is easy to check whether $t$ is set to $0$ in $\rvec{x}$.
However, to guarantee that we can decide in polynomial time
whether $\rvec{x}$ is $M$-admissible and whether $\chi(\rvec{x})>0$,
we need to bound the size of $\rvec{x}$.

To this end, consider the set $\mathcal{C}\subset\mathbb{R}^{7n}$ of all vectors satisfying both
the non-negativity condition and the equations from $M$.
Note that $\mathcal{C}$ is a polyhedral cone,
and is analogous to the standard solution cone from the theory of concrete normal coordinates.
We can then define a \emph{vertex solution}
(analogous to a vertex normal surface)
to be a vector $\rvec{x}\in\mathbb{R}^{7n}$ such that:
\begin{itemize}
\item $\rvec{x}$ lies on an extreme ray of the cone $\mathcal{C}$;
\item there is no $q\in(0,1)$ such that $q\rvec{x}$ is an integral point in $\mathbb{R}^{7n}$; and
\item $\rvec{x}$ satisfies the quadrilateral constraints.
\end{itemize}

We claim that if some solution $\rvec{x}$ exists for an instance $I$ of Problem \ref{prbm:absNormConOpt},
then the same instance $I$ must have a \emph{vertex} solution.
Indeed, since $\rvec{x}$ lies in the cone $\mathcal{C}$,
we can write $\rvec{x}=q_1\rvec{v}_1+\cdots+q_k\rvec{v}_k$,
where each vector $\rvec{v}_i$ lies on an extreme ray of $\mathcal{C}$,
and each coefficient $q_i$ is positive.
Because everything is non-negative,
the fact that $\rvec{x}$ is a solution to $I$ immediately implies that
each $\rvec{v}_i$ satisfies the quadrilateral constraints,
and hence that each $\rvec{v}_i$ is (up to scaling by a positive number) a vertex solution;
similar reasoning tells us that $t$ must be set to $0$ in each $\rvec{v}_i$.
Moreover, since the function $\chi$ is linear, we have
\[
0 < \chi(\rvec{x}) = q_1\chi(\rvec{v}_1) + \cdots + q_k\chi(\rvec{v}_k),
\]
which tells us that
there must be at least one vertex solution $\rvec{v}_i$ satisfying $\chi(\rvec{v}_i)>0$.

Thus, it suffices to bound the size of a \emph{vertex} solution $\rvec{x}$ for Problem \ref{prbm:absNormConOpt}.
We can do this using Lemma 6.1 from \cite{HassLagariasPippenger99}.
Although this lemma is stated in the setting of concrete matching equations,
its proof only uses algebraic properties that are captured by our abstract matching equations.
We can therefore apply this lemma in our setting,
which yields the following bound:
each of the $7n$ entries of $\rvec{x}$ can be written using at most $7n$ bits.
This bound is enough to guarantee that we can decide in polynomial time
whether $\rvec{x}$ is a solution to Problem \ref{prbm:absNormConOpt},
and hence completes the proof that Problem \ref{prbm:absNormConOpt} is in $\mathrm{NP}$.

To show that Problem \ref{prbm:absNormConOpt} is also $\mathrm{NP}$-hard,
we give a reduction from \scap{monotone one-in-three satisfiability}
(Problem \ref{prbm:m1in3SAT}).
Suppose we are given any collection of clauses $C=\curly{c_1,\ldots,c_n}$,
where each clause $c_k$ is
a triple $\paren{u_{k,1},u_{k,2},u_{k,3}}$ of distinct Boolean variables.
We construct a corresponding instance of \scap{abstract normal constraint optimisation}, as follows.
\begin{itemize}
\item Construct the following $p=3n-2$ abstract tetrahedra:
	\begin{itemize}
	\item $S_k=\paren{w_{k,i}}_{i=1}^7$, $k\in\{1,\ldots,n-1\}$;
	\item $T_k=\paren{x_{k,i}}_{i=1}^7$, $k\in\{1,\ldots,n\}$; and
	\item $U_k=\paren{y_{k,i}}_{i=1}^7$, $k\in\{1,\ldots,n-1\}$.
	\end{itemize}
\item Construct a collection $M=M_1\cup M_2$ of matching equations, where $M_1$ is the collection of matching equations given in Lemma \ref{lem:repVars}, and $M_2$ is the collection of matching equations given in Lemma \ref{lem:quadUnif}.
\item Choose the homogeneous linear function $\chi:\Integer^{7p}\to\Integer$ defined by
\[
\chi(\rvec{x}) = \sum_{k=1}^n\pparen{x_{k,1}+x_{k,2}+x_{k,3}}
\]
for all
\begin{align*}
\rvec{x} = \pparen{
& w_{1,1},\ldots,w_{1,7};\; \ldots;\; w_{n-1,1},\ldots,w_{n-1,7}; \\
& x_{1,1},\ldots,x_{1,7};\; \ldots;\; x_{n,1},\ldots,x_{n,7}; \\
& y_{1,1},\ldots,y_{1,7};\; \ldots;\; y_{n-1,1},\ldots,y_{n-1,7}} \in \Integer^{7p}.
\end{align*}
\item Fix $t=y_{1,5}$.
\end{itemize}
By Lemmas \ref{lem:repVars} and \ref{lem:quadUnif}, the entire construction requires $O(n^2)$ steps in total. To see that $M$ is actually a compatible collection of matching equations, recall from Lemmas \ref{lem:repVars} and \ref{lem:quadUnif} that:
\begin{itemize}
\item for each $k\in\{1,\ldots,n-1\}$ and $i\in\{1,2,3\}$, the quadrilateral coordinate $w_{k,i}$ appears twice in $M_1$, and never appears in $M_2$;
\item for each $k\in\{1,\ldots,n-1\}$ and $i\in\{4,5,6,7\}$, the triangle coordinate $w_{k,i}$ appears at most three times in $M_1$, and never appears in $M_2$;
\item for each $k\in\{1,\ldots,n\}$ and $i\in\{1,2,3\}$, the quadrilateral coordinate $x_{k,i}$ appears at most twice in $M_1$, and appears at most twice in $M_2$;
\item for each $k\in\{1,\ldots,n\}$ and $i\in\{4,5,6,7\}$, the triangle coordinate $x_{k,i}$ never appears in $M_1$, and appears at most twice in $M_2$;
\item for each $k\in\{1,\ldots,n-1\}$ and $i\in\{1,2,3\}$, the quadrilateral coordinate $y_{k,i}$ never appears in $M_1$, and appears at most four times in $M_2$; and
\item for each $k\in\{1,\ldots,n-1\}$ and $i\in\{4,5,6,7\}$, the triangle coordinate $y_{k,i}$ never appears in $M_1$, and appears at most three times in $M_2$.
\end{itemize}
So, it only remains to show that $C$ is satisfiable if and only if there exists an $M$-admissible point $\rvec{x}\in\Integer^{7p}$ such that $\chi(\rvec{x})>0$ and $t=0$.
\begin{itemize}
\item Suppose $C$ is satisfiable. Then we can fix some truth assignment such that exactly one variable is true in every clause in $C$. Consider the point
\begin{align*}
\rvec{x} = \pparen{
& w_{1,1},\ldots,w_{1,7};\; \ldots;\; w_{n-1,1},\ldots,w_{n-1,7}; \\
& x_{1,1},\ldots,x_{1,7};\; \ldots;\; x_{n,1},\ldots,x_{n,7}; \\
& y_{1,1},\ldots,y_{1,7};\; \ldots;\; y_{n-1,1},\ldots,y_{n-1,7}} \in \Integer^{7p}
\end{align*}
given by:
	\begin{itemize}
	\item $w_{k,i}=y_{k,i}=0$ for all $k\in\curly{1,\ldots,n-1}$ and $i\in\curly{1,\ldots,7}$ (in particular, $t=y_{1,5}=0$);
	\item $x_{k,i}=1$ for all $k\in\curly{1,\ldots,n}$ and $i\in\curly{1,2,3}$ such that $u_{k,i}$ is true;
	\item $x_{k,i}=0$ for all $k\in\curly{1,\ldots,n}$ and $i\in\curly{1,2,3}$ such that $u_{k,i}$ is false;
	\item $x_{k,4}=x_{k,6}=x_{k,3}$ for all $k\in\curly{1,\ldots,n}$; and
	\item $x_{k,5}=x_{k,7}=x_{k,2}+x_{k,3}$ for all $k\in\curly{1,\ldots,n}$.
	\end{itemize}
For each $k\in\{1,\ldots,n\}$, since exactly one of the variables $u_{k,1},u_{k,2},u_{k,3}$ is true, we must have exactly one of the quadrilateral coordinates $x_{k,1},x_{k,2},x_{k,3}$ equal to $1$, and the other two quadrilateral coordinates equal to $0$. In particular, this means that
\[
x_{k,1}+x_{k,2}+x_{k,3} = 1
\]
for all $k\in\curly{1,\ldots,n}$.

With this in mind, we claim that $\rvec{x}$ is $M$-admissible. To see this, first note that every coordinate of $\rvec{x}$ is clearly non-negative. Additionally, observe that for each $k\in\{1,\ldots,n-1\}$, the quadrilateral coordinates in $S_k=\paren{w_{k,i}}_{i=1}^7$ and $U_k=\paren{y_{k,i}}_{i=1}^7$ are all zero, so we can immediately see that the quadrilateral constraints are satisfied in $S_k$ and $U_k$. For $T_k=\paren{x_{k,i}}_{i=1}^7$, $k\in\{1,\ldots,n\}$, recall that exactly one of the quadrilateral coordinates $x_{k,1},x_{k,2},x_{k,3}$ is equal to $1$, and the other two quadrilateral coordinates are equal to $0$. Thus, the quadrilateral constraints are also satisfied in each $T_k$.

It remains to show that $\rvec{x}$ satisfies every equation in $M=M_1\cup M_2$. To this end, observe that:
	\begin{itemize}
	\item for all $k\in\curly{1,\ldots,n-1}$, we have $w_{k,4}=w_{k,5}=w_{k,6}=w_{k,7}=0$; and
	\item given any  $k,\ell\in\curly{1,\ldots,n}$ and $i,j\in\curly{1,2,3}$ such that $u_{k,i}$ and $u_{\ell,j}$ are two occurrences of the same variable, $u_{k,i}$ is true if and only if $u_{\ell,j}$ is true, which means that $x_{k,i}=x_{\ell,j}$.
	\end{itemize}
So, by Lemma \ref{lem:repVars}, $\rvec{x}$ satisfies the equations in $M_1$. Moreover, by construction, for all $k\in\curly{1,\ldots,n-1}$ we recover equations (i)--(vi) from Lemma \ref{lem:quadUnif}:
\QuadUnifEqns{.}
Thus, $\rvec{x}$ also satisfies the equations in $M_2$. Altogether, we see that $\rvec{x}$ is indeed $M$-admissible.

Finally, observe that
\[
\chi(\rvec{x}) = \sum_{k=1}^n\pparen{x_{k,1}+x_{k,2}+x_{k,3}} = \sum_{k=1}^n1 = n > 0.
\]
Thus, we have found an $M$-admissible point $\rvec{x}\in\Integer^{7p}$ such that $\chi(\rvec{x})>0$ and $t=0$.
\item Suppose
\begin{align*}
\rvec{x} = \pparen{
& w_{1,1},\ldots,w_{1,7};\; \ldots;\; w_{n-1,1},\ldots,w_{n-1,7}; \\
& x_{1,1},\ldots,x_{1,7};\; \ldots;\; x_{n,1},\ldots,x_{n,7}; \\
& y_{1,1},\ldots,y_{1,7};\; \ldots;\; y_{n-1,1},\ldots,y_{n-1,7}} \in \Integer^{7p}
\end{align*}
is an $M$-admissible point such that $\chi(\rvec{x})>0$ and $t=0$. For each $k\in\{1,\ldots,n\}$ and $i\in\{1,2,3\}$, take the variable $u_{k,i}$ to be true if and only if the quadrilateral coordinate $x_{k,i}$ is non-zero. To see that this gives a valid truth assignment, recall from Lemma \ref{lem:repVars} that the matching equations in $M_1\subset M$ ensure that for any $k,\ell\in\curly{1,\ldots,n}$ and $i,j\in\curly{1,2,3}$, if $u_{k,i}$ and $u_{\ell,j}$ are two occurrences of the same variable, then $x_{k,i}=x_{\ell,j}$.

By the quadrilateral constraints, for each fixed $k=1,\ldots,n$, at most one of the quadrilateral coordinates $x_{k,1},x_{k,2},x_{k,3}$ is non-zero. Thus, at most one of the three variables $u_{k,1},u_{k,2},u_{k,3}$ is true in each clause $c_k\in C$. We claim that no clause has all three variables false, and hence that exactly one variable must be true in every clause. To see this, suppose instead that for some $\ell\in\curly{1,\ldots,n}$, the variables $u_{\ell,1},u_{\ell,2},u_{\ell,3}$ are all false. This means that $x_{\ell,1}=x_{\ell,2}=x_{\ell,3}=0$.

Since we have forced $y_{1,5}=t=0$, and since the matching equations in $M_2\subset M$ are satisfied, recall from Lemma \ref{lem:quadUnif} that equation (vi) is satisfied for all $k\in\{1,\ldots,n-1\}$; that is, we have
\[
x_{k,1}+x_{k,2}+x_{k,3} = x_{k+1,1}+x_{k+1,2}+x_{k+1,3}
\]
for all $k\in\{1,\ldots,n-1\}$. Thus, we see that
\[
x_{k,1}+x_{k,2}+x_{k,3} = x_{\ell,1}+x_{\ell,2}+x_{\ell,3} = 0
\]
for all $k\in\{1,\ldots,n\}$. But this means that
\[
\chi(\rvec{x}) = \sum_{k=1}^n\pparen{x_{k,1}+x_{k,2}+x_{k,3}} = 0,
\]
contradicting the initial assumption that $\chi(\rvec{x})>0$. So, we must have a truth assignment such that exactly one variable is true in every clause in $C$. In other words, $C$ must be satisfiable.
\end{itemize}
Altogether, we conclude that our construction gives a polynomial reduction from \scap{monotone one-in-three satisfiability} (Problem \ref{prbm:m1in3SAT}) to \scap{abstract normal constraint optimisation} (Problem \ref{prbm:absNormConOpt}), and hence that Problem \ref{prbm:absNormConOpt} is $\mathrm{NP}$-hard.
\end{proof}

\section{Detecting splitting surfaces and connected spanning central surfaces}\label{sec:splitting}
In section \ref{sec:intro}, we mentioned that the problem of finding a non-trivial vertex normal sphere or disc can be solved using vertex normal surface enumeration. This exemplifies the central role that vertex normal surfaces play in a very wide range of algorithms based on normal surface theory. Thus, we have a strong incentive to investigate the computational complexity of problems that can be solved by enumerating vertex normal surfaces. In this section, we consider two closely-related problems of this type. We show that one of these problems can be solved in polynomial time, while the other is $\mathrm{NP}$-complete.

Before we formulate these two problems, we make a brief diversion to:
define the type of normal surface that lies at the heart of both of these problems,
and then discuss the significance of such surfaces.
A \textbf{spanning central surface} in a triangulation $\mathcal{T}$ is a normal surface that meets each tetrahedron of $\mathcal{T}$ in precisely one elementary disc.
We focus on spanning central surfaces that are \emph{connected},
as cutting the ambient triangulation $\mathcal{T}$ along such a surface can
provide a useful decomposition of the underlying $3$-manifold of $\mathcal{T}$.

More specifically, cutting a $3$-manifold triangulation $\mathcal{T}$ along a connected spanning central surface always decomposes $\mathcal{T}$ into one or two pieces, where each piece is homeomorphic to a regular neighbourhood of some $2$-dimensional complex. To illustrate why this happens, we first examine how a spanning central surface $S$ cuts through a single tetrahedron $\Delta$. If $S$ meets $\Delta$ in a triangle, then this triangle separates a vertex $v$ of $\Delta$ from its opposite face $f$, as illustrated in Figure \ref{subfig:cutTetTri}. Cutting along $S$ therefore divides $\Delta$ into the following two regions:
\begin{itemize}
\item a tetrahedron attached to $v$, which we can view as a regular neighbourhood of $v$ in $\Delta$; and
\item a truncated tetrahedron attached to $f$, which we can view as a regular neighbourhood of $f$ in $\Delta$.
\end{itemize}
On the other hand, if $S$ meets $\Delta$ in a quadrilateral, then this quadrilateral separates two opposite edges $e_1$ and $e_2$ of $\Delta$, as illustrated in Figure \ref{subfig:cutTetQuad}. Cutting along $S$ therefore divides $\Delta$ into the following two regions:
\begin{itemize}
\item a wedge attached to $e_1$, which we can view as a regular neighbourhood of $e_1$ in $\Delta$; and
\item a wedge attached to $e_2$, which we can view as a regular neighbourhood of $e_2$ n $\Delta$.
\end{itemize}

\begin{figure}[htbp]
\centering
	\begin{subfigure}[t]{0.45\textwidth}
	\centering
		\begin{tikzpicture}[scale=0.4]
		
		
		\fill[cyan] (0,0) -- (6,1) -- (9,7) -- cycle;		
		
		\draw[thick, line cap=round, dashed] (0,0) -- (9,7);	
		
		\fill[pink] (2.4,5.4) -- (4.8,5.8) -- (6,8.2) -- cycle;
		
		\draw[thick, line cap=round] (0,0) -- (2.4,5.4);
		\draw[thick, line cap=round] (6,8.2) -- (9,7);
		\draw[thick, line cap=round] (9,7) -- (6,1);
		\draw[thick, line cap=round] (6,1) -- (0,0);
		\draw[thick, line cap=round] (4.8,5.8) -- (6,1);
		
		\draw[thick, line cap=round] (2.4,5.4) -- (4.8,5.8);
		\draw[thick, line cap=round] (4.8,5.8) -- (6,8.2);
		\draw[thick, line cap=round] (6,8.2) -- (2.4,5.4);
		
		
		\begin{scope}[shift={(0,2)}]
		
		\fill[pink] (2.4,5.4) -- (4.8,5.8) -- (6,8.2) -- cycle;
		
		\draw[thick, line cap=round, dashed] (6,8.2) -- (2.4,5.4);
		
		\draw[thick, line cap=round] (2.4,5.4) -- (4,9);
		\draw[thick, line cap=round] (4,9) -- (6,8.2);
		\draw[thick, line cap=round] (4,9) -- (4.8,5.8);
		
		\draw[thick, line cap=round] (2.4,5.4) -- (4.8,5.8);
		\draw[thick, line cap=round] (4.8,5.8) -- (6,8.2);
		
		\fill[cyan] (4,9) circle (0.3);
		
		\end{scope}
		
		\end{tikzpicture}
	\subcaption{The pink triangle separates the blue vertex from the blue face.}
	\label{subfig:cutTetTri}
	\end{subfigure}
	\hspace{12pt}
	\begin{subfigure}[t]{0.45\textwidth}
	\centering
		\begin{tikzpicture}[scale=0.4]
		
		
		\draw[thick, line cap=round, dashed] (0,0) -- (4.5,3.5);
		
		\fill[pink] (2,4.5) -- (4.5,3.5) -- (7.5,4) -- (5,5) -- cycle;
		
		\draw[thick, line cap=round, dashed] (2,4.5) -- (4.5,3.5);
		\draw[thick, line cap=round, dashed] (4.5,3.5) -- (7.5,4);
		
		\draw[thick, line cap=round] (0,0) -- (2,4.5);
		\draw[thick, line cap=round] (7.5,4) -- (6,1);
		\draw[thick, line cap=round] (5,5) -- (6,1);
		\draw[ultra thick, line cap=round, cyan] (6,1) -- (0,0);
		
		\draw[thick, line cap=round] (7.5,4) -- (5,5);
		\draw[thick, line cap=round] (5,5) -- (2,4.5);
		
		\begin{scope}[shift={(0,2)}]
		
		\draw[thick, line cap=round, dashed] (4.5,3.5) -- (9,7);
		
		\fill[pink] (2,4.5) -- (4.5,3.5) -- (7.5,4) -- (5,5) -- cycle;
		
		\draw[thick, line cap=round] (2,4.5) -- (4,9);
		\draw[thick, line cap=round] (9,7) -- (7.5,4);
		\draw[thick, line cap=round] (4,9) -- (5,5);
		\draw[ultra thick, line cap=round, cyan] (4,9) -- (9,7);
		
		\draw[thick, line cap=round] (2,4.5) -- (4.5,3.5);
		\draw[thick, line cap=round] (4.5,3.5) -- (7.5,4);
		\draw[thick, line cap=round] (7.5,4) -- (5,5);
		\draw[thick, line cap=round] (5,5) -- (2,4.5);
		
		\end{scope}
		
		\end{tikzpicture}
	\subcaption{The pink quadrilateral separates the blue pair of opposite edges.}
	\label{subfig:cutTetQuad}
	\end{subfigure}
\caption{An elementary disc cuts a tetrahedron into two regions.}
\label{fig:cutTet}
\end{figure}
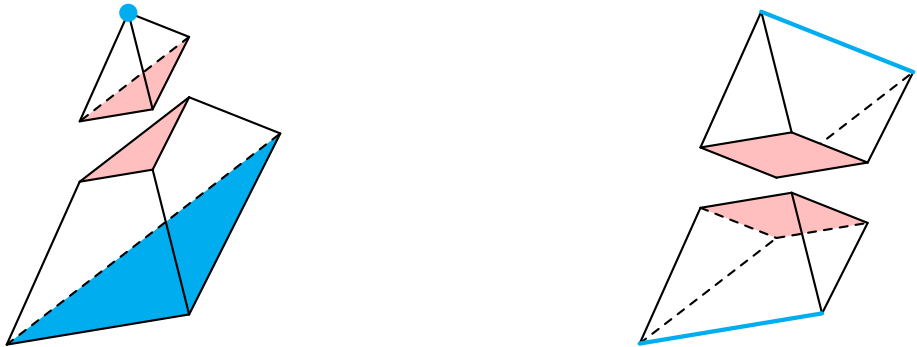

With this in mind, let $\mathcal{T}$ be a triangulation of a closed $3$-manifold, and suppose we cut $\mathcal{T}$ along a connected spanning central surface $S$. If $S$ is two-sided, then $\mathcal{T}$ decomposes into two bounded manifolds; otherwise, if $S$ is one-sided, then $\mathcal{T}$ ``unfolds'' into a single bounded manifold. Due to the way in which $S$ cuts through each tetrahedron of $\mathcal{T}$, the resulting bounded manifolds are essentially regular neighbourhoods of $2$-dimensional complexes built from triangles, edges and vertices. The structure of these complexes can often provide insight into the structure of the original $3$-manifold. In particular, it is worth mentioning the following two special cases.
\begin{itemize}
\item First, suppose $S$ is a spanning central sphere that bounds a $3$-ball in $\mathcal{T}$.
By the preceding discussion, deleting the interior of this $3$-ball from $\mathcal{T}$
leaves behind a regular neighbourhood of a $2$-dimensional complex.
This complex is therefore a \textbf{spine} of the underlying $3$-manifold of $\mathcal{T}$.
\item Second, suppose $S$ is a \textbf{splitting surface} in $\mathcal{T}$,
meaning that $S$ is a spanning central surface that is entirely composed of quadrilaterals.
Splitting surfaces were originally motivated by ideas discussed by Rubinstein in \cite{Rubinstein78},
with the first detailed discussion of these surfaces appearing in \cite{Burton03}.
Since a quadrilateral meets all four triangular faces of its ambient tetrahedron,
note that $S$ is automatically connected (assuming that $\mathcal{T}$ is connected) \cite{Burton03}.
By the preceding discussion, $S$ separates a pair of opposite edges in each tetrahedron of $\mathcal{T}$
(recall Figure \ref{subfig:cutTetQuad}).
Thus, if $S$ is two-sided, then cutting along $S$ decomposes $\mathcal{T}$ into two pieces,
where each piece is essentially a regular neighbourhood of a \emph{graph} built from the aforementioned edges.
In other words, $\mathcal{T}$ decomposes into a pair of handlebodies,
which means that $S$ realises a \textbf{Heegaard splitting} of the underlying $3$-manifold of $\mathcal{T}$.
By similar reasoning, if $S$ is one-sided, then cutting along $S$ ``unfolds'' $\mathcal{T}$ into a single handlebody,
in which case $S$ realises a \textbf{one-sided Heegaard splitting} of the underlying $3$-manifold of $\mathcal{T}$.
\end{itemize}

Given the significance of splitting surfaces, and more generally the significance of connected spanning central surfaces, it would be useful to understand the computational complexity of the following two problems.

\begin{problem}[\scap{Splitting surface}]\label{prbm:splittingSurf}
\hspace{2em}
\begin{description}[nosep]
\item[\texttt{INSTANCE:}]A triangulation $\mathcal{T}$.
\item[\texttt{QUESTION:}] Does $\mathcal{T}$ contain a splitting surface?
\end{description}
\end{problem}

\begin{problem}[\scap{Connected spanning central surface}]\label{prbm:spanningCentral}
\hspace{2em}
\begin{description}[nosep]
\item[\texttt{INSTANCE:}]A triangulation $\mathcal{T}$.
\item[\texttt{QUESTION:}]Does $\mathcal{T}$ contain a connected spanning central surface?
\end{description}
\end{problem}

As suggested earlier, our interest in these problems is also motivated by the fact that they can both be solved by enumerating vertex normal surfaces. For Problem \ref{prbm:splittingSurf}, this fact is a consequence of the following result: a splitting surface is always a vertex normal surface \cite{Burton03}. It turns out that this result extends to \emph{all} connected spanning central surfaces, which is why Problem \ref{prbm:spanningCentral} can be solved by enumerating vertex normal surfaces.

\begin{proposition}\label{prop:connCentralImpliesVertex}
Let $\mathcal{T}$ be a triangulation. If $S$ is a connected spanning central surface in $\mathcal{T}$, then $S$ is a vertex normal surface.
\end{proposition}

\begin{proof}
Let $S$ be a connected spanning central surface in $\mathcal{T}$.
Suppose, for the sake of contradiction, that $S$ is not a vertex normal surface.
Then there must exist two normal surfaces $M$ and $N$ in $\mathcal{T}$,
neither of which are multiples of $S$, such that $kS = M+N$ for some positive integer $k$.

Since $S$ is a spanning central surface, it meets each tetrahedron of $\mathcal{T}$ in exactly one elementary disc. For each tetrahedron $\Delta$ in $\mathcal{T}$, let:
\begin{itemize}
\item $\varepsilon_\Delta$ denote the elementary disc type in which $S$ meets $\Delta$; and
\item $m_\Delta$ denote the number of times $M$ meets $\Delta$ in a disc of type $\varepsilon_\Delta$.
\end{itemize}
With this in mind, consider any internal face $f$ in $\mathcal{T}$ that meets $S$, and let $\Delta$ and $\Delta'$ be the two (not necessarily distinct) tetrahedra that are glued together along $f$. Since the spanning central surface $S$ passes through $f$, observe that the disc types $\varepsilon_\Delta$ and $\varepsilon_{\Delta'}$ must both give rise to the same normal arc type in $f$. Thus, the matching equations in $f$ force $m_\Delta=m_{\Delta'}$.

Since $S$ is a connected surface, these equalities propagate so that we have $m_\Delta=m_{\Delta'}$ for any two (not necessarily adjacent) tetrahedra $\Delta$ and $\Delta'$ that meet $S$. But this is impossible, since we initially assumed that $M$ was not a multiple of $S$. So, we conclude that $S$ must be a vertex normal surface.
\end{proof}

With all this in mind, we now turn to the two main complexity results of this section. First, we show that Problem \ref{prbm:splittingSurf} has a simple polynomial-time solution. Independently, this first result is not particularly enlightening. Our interest in this result is mostly due to its stark contrast to the second main result: Problem \ref{prbm:spanningCentral} is $\mathrm{NP}$-complete. In essence, by allowing the surfaces of interest to have non-zero triangle coordinates, we have turned a computationally easy problem into a computationally hard problem. Since Problems \ref{prbm:splittingSurf} and \ref{prbm:spanningCentral} are so closely related, we can imagine that this pair of problems ``straddles'' the threshold between ``easy'' and ``hard'', with one problem lying on each side.

Our polynomial-time solution for Problem \ref{prbm:splittingSurf} is essentially just a series of three breadth-first searches.

\newcommand{\thmSplittingPoly}{Problem \ref{prbm:splittingSurf} has a polynomial-time algorithm.
That is, given an $n$-tetrahedron triangulation $\mathcal{T}$,
we can decide whether $\mathcal{T}$ has a splitting surface in time bounded by a polynomial in $n$.}
\begin{theorem}\label{thm:splittingPoly}
\thmSplittingPoly
\end{theorem}

\begin{proof}
First, consider any two (not necessarily distinct) tetrahedra $\Delta$ and $\Delta'$
that are glued together along a triangular face $f$,
and suppose we have a single quadrilateral $q$ in $\Delta$.
Since quadrilaterals of different types in $\Delta'$ never share any normal arc types,
the quadrilateral $q$ can only match up across $f$ with one choice of quadrilateral in $\Delta'$.

With this in mind, fix an initial tetrahedron $\Delta_0$ in $\mathcal{T}$.
A splitting surface in $\mathcal{T}$ must meet $\Delta_0$ in exactly one of the three possible quadrilaterals.
By propagating the above reasoning through all pairs of adjacent tetrahedra,
we see that each initial choice of quadrilateral $q_0$ in $\Delta_0$
forces at most one choice of quadrilateral in \emph{every} other tetrahedron of $\mathcal{T}$.
Thus, to determine whether $q_0$ forms part of a splitting surface,
we simply use a breadth-first search to sequentially visit each tetrahedron $\Delta$ of $\mathcal{T}$,
checking at each step whether it is possible to insert a quadrilateral in $\Delta$
without violating any matching equations.
If we encounter some $\Delta$ in which no such quadrilateral can be inserted,
then we conclude that there is no splitting surface that meets $\Delta_0$ in the quadrilateral $q_0$.
On the other hand, if the search finishes and we find that we can insert a quadrilateral at every step,
then we will have constructed a splitting surface in $\mathcal{T}$.
Thus, by performing this breadth-first search once for each of the three possible choices of quadrilateral in $\Delta_0$,
we can determine in polynomial time whether $\mathcal{T}$ contains a splitting surface.
\end{proof}

The remainder of this section is devoted to proving that detecting connected spanning central surfaces is $\mathrm{NP}$-complete. In fact, we prove a slightly stronger result: the problem remains $\mathrm{NP}$-complete even if we restrict the input to be a triangulation of an orientable $3$-manifold (see Theorem \ref{thm:spanningCentralNPhard} below). Our proof strategy is to find a reduction from the graph-theoretic computational problem \scap{Hamiltonian cycle}. We restrict our attention to graphs that are $3$-regular, since detecting Hamiltonian cycles remains $\mathrm{NP}$-complete under this condition \cite{GareyJohnson79}. To avoid confusion with the vertices and edges of triangulations, we will refer to the vertices of graphs as \emph{nodes} and the edges of graphs as \emph{arcs}.

\begin{problem}[\scap{Hamiltonian cycle}]\label{prbm:HamCycle}
\hspace{2em}
\begin{description}[nosep]
\item[\texttt{INSTANCE:}]A $3$-regular graph $G$.
\item[\texttt{QUESTION:}]Does $G$ contain a Hamiltonian cycle?
\end{description}
\end{problem}

In essence, given any $3$-regular graph $G$, our goal is to build an orientable $3$-manifold triangulation $\mathcal{T}_G$, such that $\mathcal{T}_G$ contains a connected spanning central surface if and only if $G$ contains a Hamiltonian cycle. The key to this construction is the \textbf{node gadget}, a small triangulation which we use to represent the nodes in $G$. More specifically, each node in $G$ will be assigned a corresponding node gadget, and the arcs in $G$ will determine how we glue all the node gadgets together to form the triangulation $\mathcal{T}_G$.

Before we state and prove Theorem \ref{thm:spanningCentralNPhard}, we spend some time working through the construction of the node gadget. The first step is to construct a ``triangular solid torus'', which acts as a sort of precursor to the node gadget.

\begin{construction}[Triangular solid torus]\label{cons:triTorus}
To build the triangular solid torus, start with three tetrahedra:
\begin{itemize}
\item $\Delta_0$, with vertices labelled $A,B,C,D$;
\item $\Delta_1$, with vertices labelled $E,F,G,H$; and
\item $\Delta_2$, with vertices labelled $I,J,K,L$.
\end{itemize}
The idea is to form a solid torus by stacking these tetrahedra in a cycle, with $\Delta_1$ on top of $\Delta_0$, $\Delta_2$ on top of $\Delta_1$, and $\Delta_0$ on top of $\Delta_2$. More precisely, we glue the tetrahedra together using the following face identifications.
\begin{enumerate}[(1)]
\item $ABD \longleftrightarrow GFH$
\item $EGH \longleftrightarrow IKJ$
\item $BCD \longleftrightarrow IKL$
\end{enumerate}
As a result of these face identifications, observe that the triangular solid torus has three vertices, all of which are on the boundary.
\end{construction}

This construction is illustrated in Figure \ref{fig:triTorus}. In particular, note the arrows in Figure \ref{subfig:triTorus}, which indicate that the ``top'' (face $IKL$) has been identified with the ``bottom'' (face $BCD$). These arrows should be assumed to be present in all subsequent illustrations of the triangular solid torus; omitting the arrows simply helps to declutter some of the figures that appear later on.

\begin{figure}[htbp]
\centering
	\begin{subfigure}[t]{0.63\textwidth}
	\centering
		\begin{tikzpicture}[scale=0.3]
		
		\begin{scope}[shift={(-11,0)}]
		\draw[very thick, dashed, line cap=round] (-5,-6) node[left] {$C$} -- (5,-6) node[right] {$D$};
		\draw[very thick, line cap=round] (-5,2) node[above] {$A$} -- (5,-6);
		\draw[very thick, line cap=round] (-5,-6) -- (0,-8) node[below] {$B$};
		\draw[very thick, line cap=round] (0,-8) -- (5,-6);
		\draw[very thick, line cap=round] (-5,-6) -- (-5,2);
		\draw[very thick, line cap=round] (-5,2) -- (0,-8);
		\end{scope}
		
		\draw[very thick, dashed, line cap=round] (-5,2) -- (5,-6);
		\draw[very thick, line cap=round] (0,-8) node[below] {$F$} -- (5,-6) node[right] {$H$};
		\draw[very thick, line cap=round] (-5,2) node[left] {$G$} -- (0,-8);
		\draw[very thick, line cap=round] (0,-8) -- (0,0) node[above] {$E$};
		\draw[very thick, line cap=round] (0,0) -- (5,-6);
		\draw[very thick, line cap=round] (0,0) -- (-5,2);
		
		\begin{scope}[shift={(10,0)}]
		\draw[very thick, line cap=round] (-5,2) node[left] {$K$} -- (5,-6) node[below] {$J$};
		\draw[very thick, line cap=round] (0,0) node[above] {$I$} -- (5,-6);
		\draw[very thick, line cap=round] (5,-6) -- (5,2) node[right] {$L$};
		\draw[very thick, line cap=round] (0,0) -- (-5,2);
		\draw[very thick, line cap=round] (-5,2) -- (5,2);
		\draw[very thick, line cap=round] (5,2) -- (0,0);
		\end{scope}
		
		\end{tikzpicture}
	\subcaption{Tetrahedron vertex labels.}
	\end{subfigure}
	\hspace{24pt}
	\begin{subfigure}[t]{0.28\textwidth}
	\centering
		\begin{tikzpicture}[scale=0.4]
		
		\draw[very thick, dashed, line cap=round, midarrow={0.55}{\whitearrow\whitearrow}] (5,-6) -- (-5,-6);
		\draw[very thick, dashed, line cap=round] (-5,2) -- (5,-6);
		
		\draw[very thick, line cap=round, midarrow={0.55}{\whitearrow}] (-5,-6) -- (0,-8);
		\draw[very thick, line cap=round, midarrow={0.55}{\blackarrow}] (0,-8) -- (5,-6);
		
		\draw[very thick, line cap=round] (-5,-6) -- (-5,2);
		\draw[very thick, line cap=round] (-5,2) -- (0,-8);
		\draw[very thick, line cap=round] (0,-8) -- (0,0);
		\draw[very thick, line cap=round] (0,0) -- (5,-6);
		\draw[very thick, line cap=round] (5,-6) -- (5,2);
		
		\draw[very thick, line cap=round, midarrow={0.55}{\whitearrow}] (-5,2) -- (0,0);
		\draw[very thick, line cap=round, midarrow={0.55}{\whitearrow\whitearrow}] (5,2) -- (-5,2);
		\draw[very thick, line cap=round, midarrow={0.55}{\blackarrow}] (0,0) -- (5,2);
		
		\end{tikzpicture}
	\subcaption{The triangular solid torus.}
	\label{subfig:triTorus}
	\end{subfigure}
\caption{Construction of the triangular solid torus.}
\label{fig:triTorus}
\end{figure}
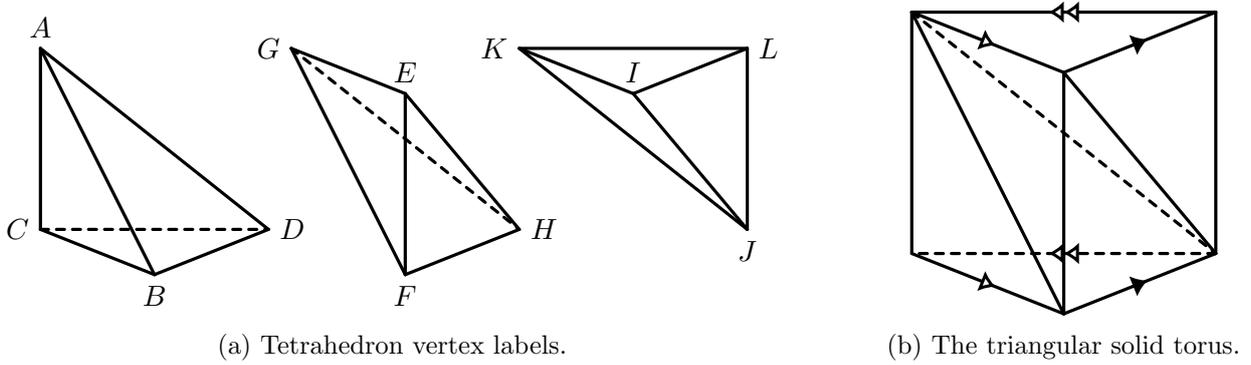

Because of its three-way symmetry,
the triangular solid torus seems particularly suitable for simulating nodes of degree $3$.
To pin down this symmetry more precisely,
we classify the edges of the triangular solid torus according to their \textbf{degree};
the degree of an edge $e$ in a triangulation $\mathcal{T}$ is the number of times
$e$ appears as an edge of some tetrahedron in $\mathcal{T}$.
As illustrated in Figure \ref{fig:triTorusDegrees}, the triangular solid torus has:
\begin{itemize}
\item three edges of degree $1$, which we will call \textbf{axis edges};
\item three edges of degree $2$, which we will call \textbf{minor edges}; and
\item three edges of degree $3$, which we will call \textbf{major edges}.
\end{itemize}
Note that these nine edges are all on the boundary of the triangular solid torus.

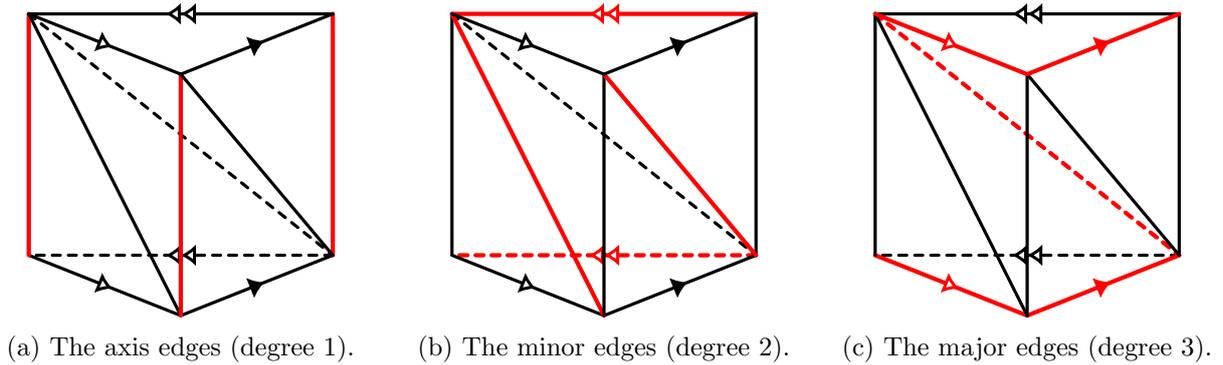
\begin{figure}[htbp]
\centering
	\begin{subfigure}[t]{0.32\textwidth}
	\centering
		\begin{tikzpicture}[scale=0.4]
		
		\draw[very thick, dashed, line cap=round, midarrow={0.55}{\whitearrow\whitearrow}] (5,-6) -- (-5,-6);
		\draw[very thick, dashed, line cap=round] (-5,2) -- (5,-6);
		
		\draw[very thick, line cap=round, midarrow={0.55}{\whitearrow}] (-5,-6) -- (0,-8);
		\draw[very thick, line cap=round, midarrow={0.55}{\blackarrow}] (0,-8) -- (5,-6);
		\draw[ultra thick, line cap=round, red] (-5,-6) -- (-5,2);
		\draw[ultra thick, line cap=round, red] (5,-6) -- (5,2);
		\draw[very thick, line cap=round, midarrow={0.55}{\whitearrow\whitearrow}] (5,2) -- (-5,2);
		
		\draw[very thick, line cap=round] (-5,2) -- (0,-8);
		\draw[ultra thick, line cap=round, red] (0,-8) -- (0,0);
		\draw[very thick, line cap=round] (0,0) -- (5,-6);
		\draw[very thick, line cap=round, midarrow={0.55}{\whitearrow}] (-5,2) -- (0,0);
		\draw[very thick, line cap=round, midarrow={0.55}{\blackarrow}] (0,0) -- (5,2);
		
		\end{tikzpicture}
	\subcaption{The axis edges (degree $1$).}
	\end{subfigure}
	\begin{subfigure}[t]{0.32\textwidth}
	\centering
		\begin{tikzpicture}[scale=0.4]
		
		\draw[ultra thick, dashed, line cap=round, red, midarrow={0.55}{\whitearrow\whitearrow}] (5,-6) -- (-5,-6);
		\draw[very thick, dashed, line cap=round] (-5,2) -- (5,-6);
		
		\draw[very thick, line cap=round, midarrow={0.55}{\whitearrow}] (-5,-6) -- (0,-8);
		\draw[very thick, line cap=round, midarrow={0.55}{\blackarrow}] (0,-8) -- (5,-6);
		\draw[very thick, line cap=round] (-5,-6) -- (-5,2);
		\draw[very thick, line cap=round] (5,-6) -- (5,2);
		\draw[ultra thick, line cap=round, red, midarrow={0.55}{\whitearrow\whitearrow}] (5,2) -- (-5,2);
		
		\draw[ultra thick, line cap=round, red] (-5,2) -- (0,-8);
		\draw[very thick, line cap=round] (0,-8) -- (0,0);
		\draw[ultra thick, line cap=round, red] (0,0) -- (5,-6);
		\draw[very thick, line cap=round, midarrow={0.55}{\whitearrow}] (-5,2) -- (0,0);
		\draw[very thick, line cap=round, midarrow={0.55}{\blackarrow}] (0,0) -- (5,2);
		
		\end{tikzpicture}
	\subcaption{The minor edges (degree $2$).}
	\end{subfigure}
	\begin{subfigure}[t]{0.32\textwidth}
	\centering
		\begin{tikzpicture}[scale=0.4]
		
		\draw[very thick, dashed, line cap=round, midarrow={0.55}{\whitearrow\whitearrow}] (5,-6) -- (-5,-6);
		\draw[ultra thick, dashed, line cap=round, red] (-5,2) -- (5,-6);
		
		\draw[ultra thick, line cap=round, red, midarrow={0.55}{\whitearrow}] (-5,-6) -- (0,-8);
		\draw[ultra thick, line cap=round, red, midarrow={0.55}{\blackarrow}] (0,-8) -- (5,-6);
		\draw[very thick, line cap=round] (-5,-6) -- (-5,2);
		\draw[very thick, line cap=round] (5,-6) -- (5,2);
		\draw[very thick, line cap=round, midarrow={0.55}{\whitearrow\whitearrow}] (5,2) -- (-5,2);
		
		\draw[very thick, line cap=round] (-5,2) -- (0,-8);
		\draw[very thick, line cap=round] (0,-8) -- (0,0);
		\draw[very thick, line cap=round] (0,0) -- (5,-6);
		\draw[ultra thick, line cap=round, red, midarrow={0.55}{\whitearrow}] (-5,2) -- (0,0);
		\draw[ultra thick, line cap=round, red, midarrow={0.55}{\blackarrow}] (0,0) -- (5,2);
		
		\end{tikzpicture}
	\subcaption{The major edges (degree $3$).}
	\end{subfigure}
\caption{The edges of the triangular solid torus all have degree $1$, $2$ or $3$.}
\label{fig:triTorusDegrees}
\end{figure}

Observe that for each pair of axis edges $e$ and $e'$, there is a single minor edge that joins an endpoint of $e$ to an endpoint of $e'$; similarly, there is a single major edge that joins an endpoint of $e$ to an endpoint of $e'$. These four edges bound a pair of triangular faces that glue together to form an annulus. The two triangular faces always come from two of the three tetrahedra in the triangular torus; we will say that the annulus ``excludes'' a tetrahedron $\Delta$ if it does \emph{not} include a triangular face from $\Delta$. As illustrated in Figure \ref{fig:triTorusAnnuli}, the boundary of the triangular solid torus consists of three annuli, each of which excludes a different tetrahedron. We number these annuli $0,1,2$ so that for each $i\in\curly{0,1,2}$, annulus $i$ excludes the tetrahedron $\Delta_i$. When discussing a particular annulus numbered $i$, we will often refer to the other two annuli as annulus $i+1$ and annulus $i-1$ (reducing modulo $3$ if necessary).

\begin{figure}[htbp]
\centering
	\begin{subfigure}[t]{0.3\textwidth}
	\centering
		\begin{tikzpicture}[scale=0.4]
		
		\fill[cyan] (5,2) -- (0,0) -- (0,-8) -- (5,-6);
		
		\draw[very thick, dashed, line cap=round, midarrow={0.55}{\whitearrow\whitearrow}] (5,-6) -- (-5,-6);
		\draw[very thick, dashed, line cap=round] (-5,2) -- (5,-6);
		
		\draw[very thick, line cap=round, midarrow={0.55}{\whitearrow}] (-5,-6) -- (0,-8);
		\draw[very thick, line cap=round, midarrow={0.55}{\blackarrow}, blue] (0,-8) -- (5,-6);
		\draw[very thick, line cap=round] (-5,-6) -- (-5,2);
		\draw[very thick, line cap=round, blue] (5,-6) -- (5,2);
		\draw[very thick, line cap=round, midarrow={0.55}{\whitearrow\whitearrow}] (5,2) -- (-5,2);
		
		\draw[very thick, line cap=round] (-5,2) -- (0,-8);
		\draw[very thick, line cap=round, blue] (0,-8) -- (0,0);
		\draw[very thick, line cap=round, blue] (0,0) -- (5,-6);
		\draw[very thick, line cap=round, midarrow={0.55}{\whitearrow}] (-5,2) -- (0,0);
		\draw[very thick, line cap=round, midarrow={0.55}{\blackarrow}, blue] (0,0) -- (5,2);
		
		\node[right, blue] at (0,-4) {$a$};
		\node[left, blue] at (5,-2) {$b$};
		
		\end{tikzpicture}
	\subcaption{Annulus $0$ excludes $\Delta_0$.}
	\end{subfigure}
	\hspace{6pt}
	\begin{subfigure}[t]{0.3\textwidth}
	\centering
		\begin{tikzpicture}[scale=0.4]
		
		\fill[cyan] (-5,2) -- (5,2) -- (5,-6) -- (-5,-6);
		
		\draw[very thick, dashed, line cap=round, midarrow={0.55}{\whitearrow\whitearrow}, blue] (5,-6) -- (-5,-6);
		\draw[very thick, dashed, line cap=round, blue] (-5,2) -- (5,-6);
		
		\draw[very thick, line cap=round, midarrow={0.55}{\whitearrow}] (-5,-6) -- (0,-8);
		\draw[very thick, line cap=round, midarrow={0.55}{\blackarrow}] (0,-8) -- (5,-6);
		\draw[very thick, line cap=round, blue] (-5,-6) -- (-5,2);
		\draw[very thick, line cap=round, blue] (5,-6) -- (5,2);
		\draw[very thick, line cap=round, midarrow={0.55}{\whitearrow\whitearrow}, blue] (5,2) -- (-5,2);
		
		\draw[very thick, line cap=round] (-5,2) -- (0,-8);
		\draw[very thick, line cap=round] (0,-8) -- (0,0);
		\draw[very thick, line cap=round] (0,0) -- (5,-6);
		\draw[very thick, line cap=round, midarrow={0.55}{\whitearrow}] (-5,2) -- (0,0);
		\draw[very thick, line cap=round, midarrow={0.55}{\blackarrow}] (0,0) -- (5,2);
		
		\node[left, blue] at (5,-2) {$a$};
		\node[right, blue] at (-5,-2) {$b$};
		
		\end{tikzpicture}
	\subcaption{Annulus $1$ excludes $\Delta_1$.}
	\end{subfigure}
	\hspace{6pt}
	\begin{subfigure}[t]{0.3\textwidth}
	\centering
		\begin{tikzpicture}[scale=0.4]
		
		\fill[cyan] (-5,2) -- (0,0) -- (0,-8) -- (-5,-6);
		
		\draw[very thick, dashed, line cap=round, midarrow={0.55}{\whitearrow\whitearrow}] (5,-6) -- (-5,-6);
		\draw[very thick, dashed, line cap=round] (-5,2) -- (5,-6);
		
		\draw[very thick, line cap=round, midarrow={0.55}{\whitearrow}, blue] (-5,-6) -- (0,-8);
		\draw[very thick, line cap=round, midarrow={0.55}{\blackarrow}] (0,-8) -- (5,-6);
		\draw[very thick, line cap=round, blue] (-5,-6) -- (-5,2);
		\draw[very thick, line cap=round] (5,-6) -- (5,2);
		\draw[very thick, line cap=round, midarrow={0.55}{\whitearrow\whitearrow}] (5,2) -- (-5,2);
		
		\draw[very thick, line cap=round, blue] (-5,2) -- (0,-8);
		\draw[very thick, line cap=round, blue] (0,-8) -- (0,0);
		\draw[very thick, line cap=round] (0,0) -- (5,-6);
		\draw[very thick, line cap=round, midarrow={0.55}{\whitearrow}, blue] (-5,2) -- (0,0);
		\draw[very thick, line cap=round, midarrow={0.55}{\blackarrow}] (0,0) -- (5,2);
		
		\node[right, blue] at (-5,-2) {$a$};
		\node[left, blue] at (0,-4) {$b$};
		
		\end{tikzpicture}
	\subcaption{Annulus $2$ excludes $\Delta_2$.}
	\end{subfigure}
\caption{The three annuli on the boundary of the triangular solid torus.}
\label{fig:triTorusAnnuli}
\end{figure}
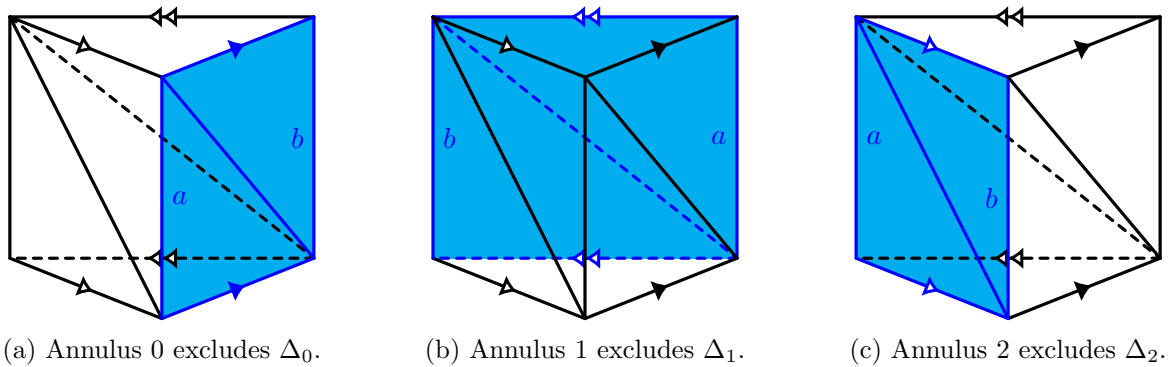

We can number the axis edges and the vertices of the triangular solid torus in a similar way. Since each axis edge meets a different tetrahedron, we can number them $0,1,2$ so that axis edge $i$ always meets the tetrahedron $\Delta_i$. Moreover, each vertex only meets one of the three axis edges; we number the vertices $0,1,2$ so that vertex $i$ always meets axis edge $i$. With this numbering, observe that for each $i\in\{0,1,2\}$, annulus $i$ meets axis edges $i+1$ and $i-1$.

Recall that our goal is to take any $3$-regular graph $G$ and build an orientable $3$-manifold triangulation $\mathcal{T}_G$, such that $\mathcal{T}_G$ contains a connected spanning central surface if and only if $G$ contains a Hamiltonian cycle. Since $G$ is $3$-regular, a Hamiltonian cycle in $G$ can pass through each node in precisely three possible ways; specifically, for each node $v$ in $G$, a Hamiltonian cycle in $G$ must contain precisely two of the three arcs incident with $v$. Our idea is to simulate this using spanning central surfaces that are incident with precisely two of the three annuli on the boundary of the triangular solid torus. This role is played perfectly by the three spanning central surfaces illustrated in Figure \ref{fig:triTorusTubes}. Note that each of these surfaces is topologically a tube. (A tube is topologically equivalent to an annulus, but since we have already reserved the word ``annulus'', we will exclusively use the word ``tube'' to describe these normal surfaces.)

\begin{figure}[htbp]
\centering
	\begin{subfigure}[t]{0.32\textwidth}
	\centering
		\begin{tikzpicture}[scale=0.4]
		
		\draw[thick, dashed, line cap=round, gray] (5,-6) -- (-5,-6);
		\draw[thick, dashed, line cap=round, gray] (-5,2) -- (5,-6);
		
		\draw[thick, line cap=round, gray] (-5,-6) -- (0,-8);
		\draw[thick, line cap=round, gray] (0,-8) -- (5,-6);
		\draw[thick, line cap=round, gray] (-5,-6) -- (-5,2);
		\draw[thick, line cap=round, gray] (5,-6) -- (5,2);
		\draw[thick, line cap=round, gray] (5,2) -- (-5,2);
		
		\fill[pink] (-3.75,1.5) -- (-2.5,2) -- (-2.5,-6) -- (-3.75,-6.5) -- cycle;
		\draw[very thick, line cap=round] (-3.75,1.5) -- (-2.5,2);
		\draw[very thick, dashed, line cap=round] (-2.5,2) -- (-2.5,-6);
		\draw[very thick, dashed, line cap=round] (-2.5,-6) -- (-3.75,-6.5);
		\draw[very thick, line cap=round] (-3.75,-6.5) -- (-3.75,1.5);
		\draw[very thick, dashed, line cap=round] (-3.75,1.5) -- (-2.5,0);
		\draw[very thick, dashed, line cap=round] (-2.5,0) -- (-3.75,-0.5);
		
		\draw[thick, line cap=round, gray] (-5,2) -- (0,-8);
		\draw[thick, line cap=round, gray] (0,-8) -- (0,0);
		\draw[thick, line cap=round, gray] (0,0) -- (5,-6);
		\draw[thick, line cap=round, gray] (-5,2) -- (0,0);
		\draw[thick, line cap=round, gray] (0,0) -- (5,2);
		
		\end{tikzpicture}
	\end{subfigure}
	\begin{subfigure}[t]{0.32\textwidth}
	\centering
		\begin{tikzpicture}[scale=0.4]
		
		\draw[thick, dashed, line cap=round, gray] (5,-6) -- (-5,-6);
		\draw[thick, dashed, line cap=round, gray] (-5,2) -- (5,-6);
		
		\draw[thick, line cap=round, gray] (-5,-6) -- (0,-8);
		\draw[thick, line cap=round, gray] (0,-8) -- (5,-6);
		\draw[thick, line cap=round, gray] (-5,-6) -- (-5,2);
		\draw[thick, line cap=round, gray] (5,-6) -- (5,2);
		\draw[thick, line cap=round, gray] (5,2) -- (-5,2);
		
		\fill[pink] (-1.25,0.5) -- (1.25,0.5) -- (1.25,-7.5) -- (-1.25,-7.5) -- cycle;
		\draw[very thick, line cap=round] (-1.25,0.5) -- (1.25,0.5);
		\draw[very thick, line cap=round] (1.25,0.5) -- (1.25,-7.5);
		\draw[very thick, dashed, line cap=round] (1.25,-7.5) -- (-1.25,-7.5);
		\draw[very thick, line cap=round] (-1.25,-7.5) -- (-1.25,0.5);
		\draw[very thick, dashed, line cap=round] (-1.25,0.5) -- (1.25,-1.5);
		\draw[very thick, dashed, line cap=round] (1.25,-7.5) -- (-1.25,-5.5);
		
		\draw[thick, line cap=round, gray] (-5,2) -- (0,-8);
		\draw[thick, line cap=round, gray] (0,-8) -- (0,0);
		\draw[thick, line cap=round, gray] (0,0) -- (5,-6);
		\draw[thick, line cap=round, gray] (-5,2) -- (0,0);
		\draw[thick, line cap=round, gray] (0,0) -- (5,2);
		
		\end{tikzpicture}
	\end{subfigure}
	\begin{subfigure}[t]{0.32\textwidth}
	\centering
		\begin{tikzpicture}[scale=0.4]
		
		\draw[thick, dashed, line cap=round, gray] (5,-6) -- (-5,-6);
		\draw[thick, dashed, line cap=round, gray] (-5,2) -- (5,-6);
		
		\draw[thick, line cap=round, gray] (-5,-6) -- (0,-8);
		\draw[thick, line cap=round, gray] (0,-8) -- (5,-6);
		\draw[thick, line cap=round, gray] (-5,-6) -- (-5,2);
		\draw[thick, line cap=round, gray] (5,-6) -- (5,2);
		\draw[thick, line cap=round, gray] (5,2) -- (-5,2);
		
		\fill[pink] (3.75,1.5) -- (2.5,2) -- (2.5,-6) -- (3.75,-6.5) -- cycle;
		\draw[very thick, line cap=round] (3.75,1.5) -- (2.5,2);
		\draw[very thick, dashed, line cap=round] (2.5,2) -- (2.5,-6);
		\draw[very thick, dashed, line cap=round] (2.5,-6) -- (3.75,-6.5);
		\draw[very thick, line cap=round] (3.75,-6.5) -- (3.75,1.5);
		\draw[very thick, dashed, line cap=round] (3.75,-6.5) -- (2.5,-4);
		\draw[very thick, dashed, line cap=round] (2.5,-4) -- (3.75,-4.5);
		
		\draw[thick, line cap=round, gray] (-5,2) -- (0,-8);
		\draw[thick, line cap=round, gray] (0,-8) -- (0,0);
		\draw[thick, line cap=round, gray] (0,0) -- (5,-6);
		\draw[thick, line cap=round, gray] (-5,2) -- (0,0);
		\draw[thick, line cap=round, gray] (0,0) -- (5,2);
		
		\end{tikzpicture}
	\end{subfigure}
\caption{Three spanning central tubes in the triangular solid torus.}
\label{fig:triTorusTubes}
\end{figure}
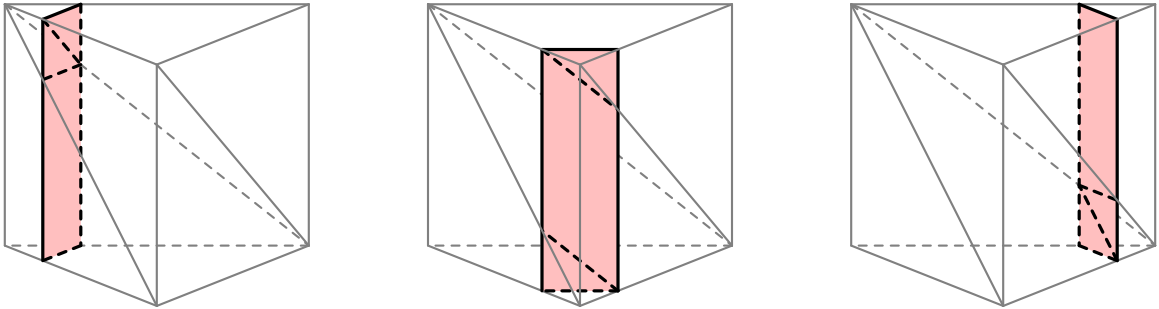

Unfortunately, the triangular solid torus also contains a number of other, unwanted spanning central surfaces. Some of these are illustrated in Figure \ref{fig:unwantedSpanning}. To circumvent this issue, we construct a ``triangular pillow'', which allows us to modify the triangular solid torus in a way that eliminates all but one of the unwanted surfaces.

\begin{figure}[htbp]
\centering
	\begin{subfigure}[t]{0.24\textwidth}
	\centering
		\begin{tikzpicture}[scale=0.4]
		
		\draw[thick, dashed, line cap=round, gray] (5,-6) -- (-5,-6);
		\draw[thick, dashed, line cap=round, gray] (-5,2) -- (5,-6);
		
		\draw[thick, line cap=round, gray] (-5,-6) -- (0,-8);
		\draw[thick, line cap=round, gray] (0,-8) -- (5,-6);
		\draw[thick, line cap=round, gray] (-5,-6) -- (-5,2);
		\draw[thick, line cap=round, gray] (5,-6) -- (5,2);
		\draw[thick, line cap=round, gray] (5,2) -- (-5,2);
		
		\fill[pink] (-5,-2) -- (0,-4) -- (5,-2) -- cycle;
		\draw[very thick, line cap=round] (-5,-2) -- (0,-4);
		\draw[very thick, dashed, line cap=round] (5,-2) -- (-5,-2);
		\draw[very thick, line cap=round] (0,-4) -- (5,-2);
		\draw[very thick, dashed, line cap=round] (-2.5,-3) -- (0,-2);
		\draw[very thick, dashed, line cap=round] (2.5,-3) -- (0,-2);
		
		\draw[thick, line cap=round, gray] (-5,2) -- (0,-8);
		\draw[thick, line cap=round, gray] (0,-8) -- (0,0);
		\draw[thick, line cap=round, gray] (0,0) -- (5,-6);
		\draw[thick, line cap=round, gray] (-5,2) -- (0,0);
		\draw[thick, line cap=round, gray] (0,0) -- (5,2);
		
		\end{tikzpicture}
	\end{subfigure}
	\begin{subfigure}[t]{0.24\textwidth}
	\centering
		\begin{tikzpicture}[scale=0.4]
		
		\draw[thick, dashed, line cap=round, gray] (5,-6) -- (-5,-6);
		\draw[thick, dashed, line cap=round, gray] (-5,2) -- (5,-6);
		
		\draw[thick, line cap=round, gray] (-5,-6) -- (0,-8);
		\draw[thick, line cap=round, gray] (0,-8) -- (5,-6);
		\draw[thick, line cap=round, gray] (-5,-6) -- (-5,2);
		\draw[thick, line cap=round, gray] (5,-6) -- (5,2);
		\draw[thick, line cap=round, gray] (5,2) -- (-5,2);
		
		\fill[pink] (-5,-2) -- (-2.5,-7) -- (0,-6) -- cycle;
		\fill[pink] (0,-4) -- (5,-2) -- (0,2) -- (-2.5,1) -- cycle;
		\draw[very thick, line cap=round] (-5,-2) -- (-2.5,-7);
		\draw[very thick, dashed, line cap=round] (-2.5,-7) -- (0,-6);
		\draw[very thick, dashed, line cap=round] (0,-6) -- (-5,-2);
		\draw[very thick, line cap=round] (0,-4) -- (5,-2);
		\draw[very thick, dashed, line cap=round] (5,-2) -- (0,2);
		\draw[very thick, line cap=round] (0,2) -- (-2.5,1);
		\draw[very thick, line cap=round] (-2.5,1) -- (0,-4);
		\draw[very thick, dashed, line cap=round] (-2.5,1) -- (2.5,-3);
		
		\draw[thick, line cap=round, gray] (-5,2) -- (0,-8);
		\draw[thick, line cap=round, gray] (0,-8) -- (0,0);
		\draw[thick, line cap=round, gray] (0,0) -- (5,-6);
		\draw[thick, line cap=round, gray] (-5,2) -- (0,0);
		\draw[thick, line cap=round, gray] (0,0) -- (5,2);
		
		\end{tikzpicture}
	\end{subfigure}
	\begin{subfigure}[t]{0.24\textwidth}
	\centering
		\begin{tikzpicture}[scale=0.4]
		
		\draw[thick, dashed, line cap=round, gray] (5,-6) -- (-5,-6);
		\draw[thick, dashed, line cap=round, gray] (-5,2) -- (5,-6);
		
		\draw[thick, line cap=round, gray] (-5,-6) -- (0,-8);
		\draw[thick, line cap=round, gray] (0,-8) -- (5,-6);
		\draw[thick, line cap=round, gray] (-5,-6) -- (-5,2);
		\draw[thick, line cap=round, gray] (5,-6) -- (5,2);
		\draw[thick, line cap=round, gray] (5,2) -- (-5,2);
		
		\fill[pink] (0,2) -- (2.5,1) -- (5,-2) -- cycle;
		\fill[pink] (-5,-2) -- (0,-4) -- (2.5,-7) -- (0,-6) -- cycle;
		\draw[very thick, line cap=round] (0,2) -- (2.5,1);
		\draw[very thick, line cap=round] (2.5,1) -- (5,-2);
		\draw[very thick, dashed, line cap=round] (5,-2) -- (0,2);
		\draw[very thick, line cap=round] (-5,-2) -- (0,-4);
		\draw[very thick, line cap=round] (0,-4) -- (2.5,-7);
		\draw[very thick, dashed, line cap=round] (2.5,-7) -- (0,-6);
		\draw[very thick, dashed, line cap=round] (0,-6) -- (-5,-2);
		\draw[very thick, dashed, line cap=round] (2.5,-7) -- (-2.5,-3);
		
		\draw[thick, line cap=round, gray] (-5,2) -- (0,-8);
		\draw[thick, line cap=round, gray] (0,-8) -- (0,0);
		\draw[thick, line cap=round, gray] (0,0) -- (5,-6);
		\draw[thick, line cap=round, gray] (-5,2) -- (0,0);
		\draw[thick, line cap=round, gray] (0,0) -- (5,2);
		
		\end{tikzpicture}
	\end{subfigure}
	\begin{subfigure}[t]{0.24\textwidth}
	\centering
		\begin{tikzpicture}[scale=0.4]
		
		\draw[thick, dashed, line cap=round, gray] (5,-6) -- (-5,-6);
		\draw[thick, dashed, line cap=round, gray] (-5,2) -- (5,-6);
		
		\draw[thick, line cap=round, gray] (-5,-6) -- (0,-8);
		\draw[thick, line cap=round, gray] (0,-8) -- (5,-6);
		\draw[thick, line cap=round, gray] (-5,-6) -- (-5,2);
		\draw[thick, line cap=round, gray] (5,-6) -- (5,2);
		\draw[thick, line cap=round, gray] (5,2) -- (-5,2);
		
		\fill[pink] (2.5,-7) -- (-2.5,-7) -- (-5,-2) -- (0,-2) -- cycle;
		\draw[very thick, dashed, line cap=round] (2.5,-7) -- (-2.5,-7);
		\draw[very thick, line cap=round] (-2.5,-7) -- (-5,-2);
		\draw[very thick, dashed, line cap=round] (-5,-2) -- (0,-2);
		\fill[pink] (0,-2) -- (2.5,-7) -- (0,-4) -- (-2.5,1) -- cycle;
		\draw[very thick, dashed, line cap=round] (0,-2) -- (2.5,-7);
		\draw[very thick, line cap=round] (2.5,-7) -- (0,-4);
		\draw[very thick, line cap=round] (0,-4) -- (-2.5,1);
		\fill[pink] (0,-2) -- (-2.5,1) -- (2.5,1) -- (5,-2) -- cycle;
		\draw[very thick, dashed, line cap=round] (0,-2) -- (-2.5,1);
		\draw[very thick, line cap=round] (-2.5,1) -- (2.5,1);
		\draw[very thick, line cap=round] (2.5,1) -- (5,-2);
		\draw[very thick, dashed, line cap=round] (5,-2) -- (0,-2);
		
		\draw[thick, line cap=round, gray] (-5,2) -- (0,-8);
		\draw[thick, line cap=round, gray] (0,-8) -- (0,0);
		\draw[thick, line cap=round, gray] (0,0) -- (5,-6);
		\draw[thick, line cap=round, gray] (-5,2) -- (0,0);
		\draw[thick, line cap=round, gray] (0,0) -- (5,2);
		
		\end{tikzpicture}
	\end{subfigure}
\caption{Some unwanted spanning central surfaces in the triangular solid torus.}
\label{fig:unwantedSpanning}
\end{figure}
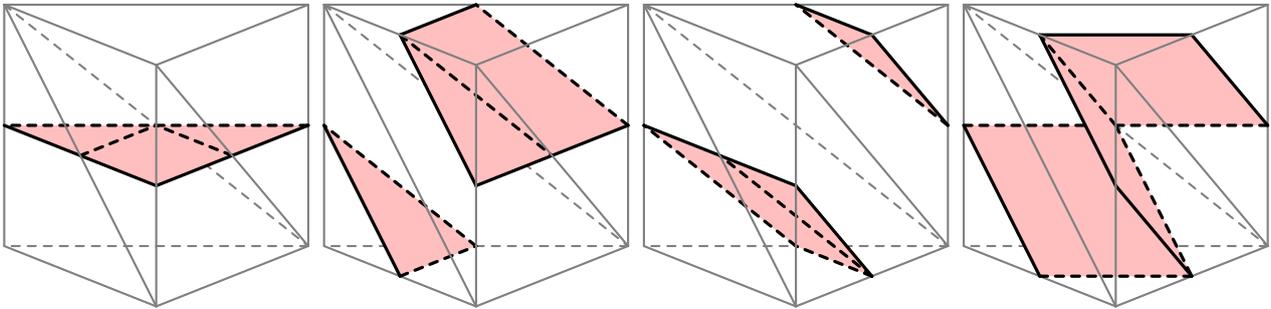

\begin{construction}[Triangular pillow]\label{cons:triPillow}
To build the triangular pillow, start with two tetrahedra:
\begin{itemize}
\item $\Delta_0$, with vertices labelled $A,B,C,D$; and
\item $\Delta_1$, with vertices labelled $E,F,G,H$.
\end{itemize}
We glue these tetrahedra together using the following face identifications.
\begin{enumerate}[(1)]
\item $ABD \longleftrightarrow EFG$
\item $ACD \longleftrightarrow EHG$
\item $BCD \longleftrightarrow FHG$
\end{enumerate}
As a result of these face identifications, observe that the triangular pillow has one internal vertex and three boundary vertices.
\end{construction}

This construction is illustrated in Figure \ref{fig:consTriPillow}. Note that the boundary faces are $ABC$ and $EFH$, with boundary edges:
\begin{itemize}
\item $AB\sim EF$, which we label edge $a$;
\item $BC\sim FH$, which we label edge $b$; and
\item $AC\sim EH$, which we label edge $c$.
\end{itemize}

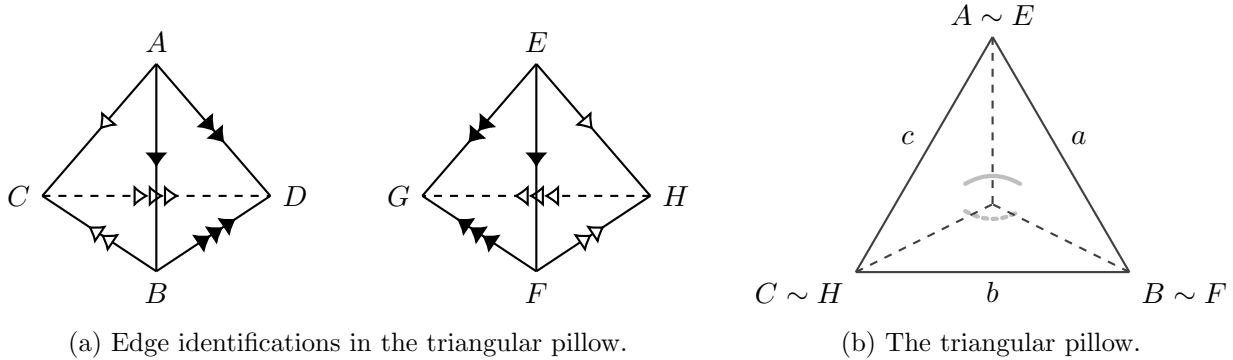
\begin{figure}[htbp]
\centering
	\begin{subfigure}[t]{0.59\textwidth}
	\centering
		\begin{tikzpicture}[scale=0.5]
		
		\node[inner sep=0pt, label=above:{$A$}] at (0,5.5) (A) {};
		\node[inner sep=0pt, label=below:{$B$}] at (0,0) (B) {};
		\node[inner sep=0pt, label=left:{$C$}] at (-3,2) (C) {};
		\node[inner sep=0pt, label=right:{$D$}] at (3,2) (D) {};
		
		\draw[thick, dashed, midarrow={0.6}{\whitearrow\whitearrow\whitearrow}] (C.center) -- (D.center);
		\draw[thick, midarrow={0.6}{\whitearrow\whitearrow}] (B.center) -- (C.center);
		\draw[thick, midarrow={0.5}{\whitearrow}] (A.center) -- (C.center);
		\draw[thick, midarrow={0.6}{\blackarrow\blackarrow}] (A.center) -- (D.center);
		\draw[thick, midarrow={0.7}{\blackarrow\blackarrow\blackarrow}] (B.center) -- (D.center);
		\draw[thick, midarrow={0.5}{\blackarrow}] (A.center) -- (B.center);

		\begin{scope}[shift={(10,0)}]
		\node[inner sep=0pt, label=above:{$E$}] at (0,5.5) (E) {};
		\node[inner sep=0pt, label=below:{$F$}] at (0,0) (F) {};
		\node[inner sep=0pt, label=left:{$G$}] at (-3,2) (G) {};
		\node[inner sep=0pt, label=right:{$H$}] at (3,2) (H) {};
		
		\draw[thick, dashed, midarrow={0.6}{\whitearrow\whitearrow\whitearrow}] (H.center) -- (G.center);
		\draw[thick, midarrow={0.7}{\blackarrow\blackarrow\blackarrow}] (F.center) -- (G.center);
		\draw[thick, midarrow={0.6}{\blackarrow\blackarrow}] (E.center) -- (G.center);
		\draw[thick, midarrow={0.5}{\whitearrow}] (E.center) -- (H.center);
		\draw[thick, midarrow={0.6}{\whitearrow\whitearrow}] (F.center) -- (H.center);
		\draw[thick, midarrow={0.5}{\blackarrow}] (E.center) -- (F.center);
		\end{scope}
		
		\end{tikzpicture}
	\subcaption{Edge identifications in the triangular pillow.}
	\label{fig:triPillowEdges}
	\end{subfigure}
	\begin{subfigure}[t]{0.39\textwidth}
	\centering
		\begin{tikzpicture}[scale=1.8]
		
		\draw[ultra thick, lightgray, line cap=round]
			(-0.2,0.65) to[out=30, in=150] (0.2,0.65);
		\draw[ultra thick, lightgray, line cap=round, dotted]
			(-0.2,0.45) to[out=330, in=210] (0.2,0.45);
		
		\draw[thick, darkgray, dashed] (1,0) -- (0,0.5);
		\draw[thick, darkgray, dashed] (-1,0) -- (0,0.5);
		\draw[thick, darkgray, dashed] (0,1.732) -- (0,0.5);
		
		\draw[thick, darkgray] (1,0) -- (-1,0);
		\draw[thick, darkgray] (1,0) -- (0,1.732);
		\draw[thick, darkgray] (-1,0) -- (0,1.732);
		
		\node[above right] at (0.5,0.866) {$a$};
		\node[below] at (0,0) {$b$};
		\node[above left] at (-0.5,0.866) {$c$};
		
		\node[label={90:$A\sim E$}, inner sep=0pt] at (0,1.732) {};
		\node[label={330:$B\sim F$}, inner sep=0pt] at (1,0) {};
		\node[label={210:$C\sim H$}, inner sep=0pt] at (-1,0) {};
		
		\end{tikzpicture}
	\subcaption{The triangular pillow.}
	\label{fig:triPillow}
	\end{subfigure}
\caption{Construction of the triangular pillow.}
\label{fig:consTriPillow}
\end{figure}

Before we explain how we use the triangular pillow to modify the triangular solid torus, it will be useful to first understand all the spanning central surfaces in the triangular pillow. By definition, such surfaces consist of one elementary disc in each of the tetrahedra $\Delta_0$ and $\Delta_1$. Observe that each of the seven choices of elementary disc in $\Delta_0$ can only match up with one of the seven choices of elementary disc in $\Delta_1$. Thus, the triangular pillow contains seven spanning central surfaces, all of which are connected: three discs that form the links of the boundary vertices, a sphere that forms the link of the internal vertex, and three discs built by attaching a quadrilateral in $\Delta_0$ to a quadrilateral in $\Delta_1$.

Each spanning central disc $D$ in the triangular pillow meets two of the three boundary edges; we will say that $D$ is ``parallel'' to the boundary edge it does \emph{not} meet. For each boundary edge $e$ in the triangular pillow, two of the six spanning central discs are ``parallel'' to $e$; one of these discs forms a vertex link, while the other is built from quadrilaterals. Since we really only care about how these discs meet the boundary of the triangular pillow, we can categorise them into three ``types'', as illustrated in Figure \ref{fig:triPillowDiscs}.

\begin{figure}[htbp]
\centering
	\begin{subfigure}[t]{0.3\textwidth}
	\centering
		\begin{tikzpicture}[scale=2]
		
		\fill[pink] (-0.5,0) to[out=90, in=330] (-0.75,0.433) to[out=255, in=165] cycle;
		
		\draw[very thick, line cap=round, dashed] (-0.5,0) to[out=90, in=330] (-0.75,0.433);
		
		\draw[ultra thick, lightgray, line cap=round]
			(-0.2,0.65) to[out=30, in=150] (0.2,0.65);
		\draw[ultra thick, lightgray, line cap=round, dotted]
			(-0.2,0.45) to[out=330, in=210] (0.2,0.45);
		
		\draw[thick, darkgray] (1,0) -- (-1,0);
		\draw[thick, darkgray] (1,0) -- (0,1.732);
		\draw[thick, darkgray] (-1,0) -- (0,1.732);
		
		\node[above right] at (0.5,0.866) {$a$};
		\node[below] at (0,0) {$b$};
		\node[above left] at (-0.5,0.866) {$c$};
		
		\draw[very thick, line cap=round] (-0.75,0.433) to[out=255, in=165] (-0.5,0);
		
		\end{tikzpicture}
	\subcaption{Disc parallel to edge $a$.}
	\end{subfigure}
	\hspace{12pt}
	\begin{subfigure}[t]{0.3\textwidth}
	\centering
		\begin{tikzpicture}[scale=2]
		
		\fill[pink] (0.25,1.299) to[out=210, in=330] (-0.25,1.299) to[out=45, in=135] cycle;
		
		\draw[very thick, line cap=round, dashed] (0.25,1.299) to[out=210, in=330] (-0.25,1.299);
		
		\draw[ultra thick, lightgray, line cap=round]
			(-0.2,0.65) to[out=30, in=150] (0.2,0.65);
		\draw[ultra thick, lightgray, line cap=round, dotted]
			(-0.2,0.45) to[out=330, in=210] (0.2,0.45);
		
		\draw[thick, darkgray] (1,0) -- (-1,0);
		\draw[thick, darkgray] (1,0) -- (0,1.732);
		\draw[thick, darkgray] (-1,0) -- (0,1.732);
		
		\node[above right] at (0.5,0.866) {$a$};
		\node[below] at (0,0) {$b$};
		\node[above left] at (-0.5,0.866) {$c$};
		
		\draw[very thick, line cap=round] (-0.25,1.299) to[out=45, in=135] (0.25,1.299);
		
		\end{tikzpicture}
	\subcaption{Disc parallel to edge $b$.}
	\end{subfigure}
	\hspace{12pt}
	\begin{subfigure}[t]{0.3\textwidth}
	\centering
		\begin{tikzpicture}[scale=2]
		
		\fill[pink] (0.5,0) to[out=90, in=210] (0.75,0.433) to[out=285, in=15] cycle;
		
		\draw[very thick, line cap=round, dashed] (0.5,0) to[out=90, in=210] (0.75,0.433);
		
		\draw[ultra thick, lightgray, line cap=round]
			(-0.2,0.65) to[out=30, in=150] (0.2,0.65);
		\draw[ultra thick, lightgray, line cap=round, dotted]
			(-0.2,0.45) to[out=330, in=210] (0.2,0.45);
		
		\draw[thick, darkgray] (1,0) -- (-1,0);
		\draw[thick, darkgray] (1,0) -- (0,1.732);
		\draw[thick, darkgray] (-1,0) -- (0,1.732);
		
		\node[above right] at (0.5,0.866) {$a$};
		\node[below] at (0,0) {$b$};
		\node[above left] at (-0.5,0.866) {$c$};
		
		\draw[very thick, line cap=round] (0.75,0.433) to[out=285, in=15] (0.5,0);
		
		\end{tikzpicture}
	\subcaption{Disc parallel to edge $c$.}
	\end{subfigure}
\caption{The three ``types'' of spanning central disc in the triangular pillow.}
\label{fig:triPillowDiscs}
\end{figure}
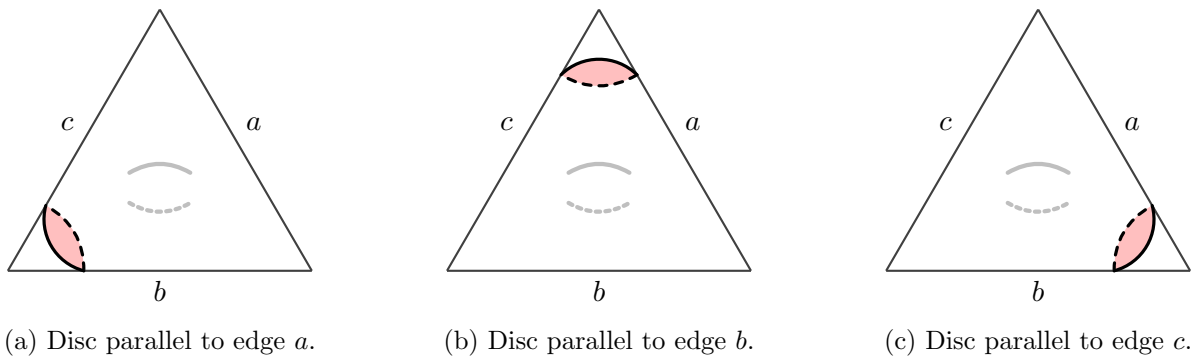

Our idea is to use the spanning central discs in the triangular pillow to mimic the normal arcs in a triangular face. To see how this works, imagine a triangulation $\mathcal{T}$ in which two tetrahedron faces $IJK$ and $LMN$ have been identified, and let $f$ denote the resulting internal face in $\mathcal{T}$. We can ``insert'' the triangular pillow from Construction \ref{cons:triPillow} by replacing the face identification $IJK\longleftrightarrow LMN$ with:
\begin{itemize}
\item $IJK\longleftrightarrow ABC$; and
\item $EFH\longleftrightarrow LMN$.
\end{itemize}
Having done this, we can imagine that the face $f$ has been ``inflated'' to become a pillow. Moreover, we can think of the discs passing through the pillow as thickened versions of the normal arcs in $f$.

\begin{observation}\label{obs:twoDiscs}
Since the triangular pillow has two discs ``parallel'' to each boundary edge, each normal arc in $f$ has two corresponding discs.
\end{observation}

Let $\mathcal{T}'$ denote the triangulation that results from inserting the pillow into $\mathcal{T}$ in this way. A connected spanning central surface $S$ in $\mathcal{T}'$ must pass through both tetrahedra of the inserted pillow. Moreover, since $S$ is connected, $S$ cannot intersect the inserted pillow in the internal vertex-linking sphere. In other words, $S$ must pass through the inserted pillow in one of the six possible discs. The key takeaway is that instead of thinking of $S$ as a connected spanning central surface in $\mathcal{T}'$, we can essentially think of $S$ as a connected spanning central surface in $\mathcal{T}$ that passes through the face $f$. To put it another way, by inserting the pillow along the face $f$, we have eliminated all the connected spanning central surfaces in $\mathcal{T}$ that do \emph{not} pass through $f$. This idea is the key inspiration for the node gadget.

\begin{construction}[Node gadget]\label{cons:nodeGadget}
To build the node gadget, insert a copy of the triangular pillow (Construction \ref{cons:triPillow}) between each of the three pairs of identified faces in the triangular solid torus (Construction \ref{cons:triTorus}).
\end{construction}

Observe that the node gadget has three boundary vertices and nine boundary edges, all of which are inherited from the vertices and edges of the triangular solid torus. The nine boundary edges inherit the classifications (axis, minor, or major) from the edges of the triangular solid torus. Thus, the node gadget also inherits the three boundary annuli from the triangular solid torus, as well as all the labellings that we introduced. (For reference, see Construction \ref{cons:triTorus}, and the paragraphs following that construction.) In addition, note that the node gadget has three internal vertices and nine internal edges, which come from the three inserted pillows.

In effect, we can associate the connected spanning central surfaces in the node gadget with the
connected spanning central surfaces that pass through all three internal faces of the triangular solid torus.
In fact, we claim that every connected spanning central surface in the node gadget corresponds to
one of the four ``types'' of surface shown in Figure \ref{fig:nodeSurfs}.
To see why, recall from Observation \ref{obs:twoDiscs} that there are
two choices of disc in each of the three inserted pillows in the node gadget,
which means that each of the four surface types includes eight different surfaces;
thus, the node gadget must contain \emph{at least} 32 connected spanning central surfaces.
Using \texttt{Regina} \cite{Burton13Regina,Regina} to enumerate all vertex normal surfaces,
we find that the node gadget contains \emph{precisely} 32 connected spanning central surfaces.
Thus, we see that the four surface types shown in Figure \ref{fig:nodeSurfs}
capture all possible connected spanning central surfaces in the node gadget.

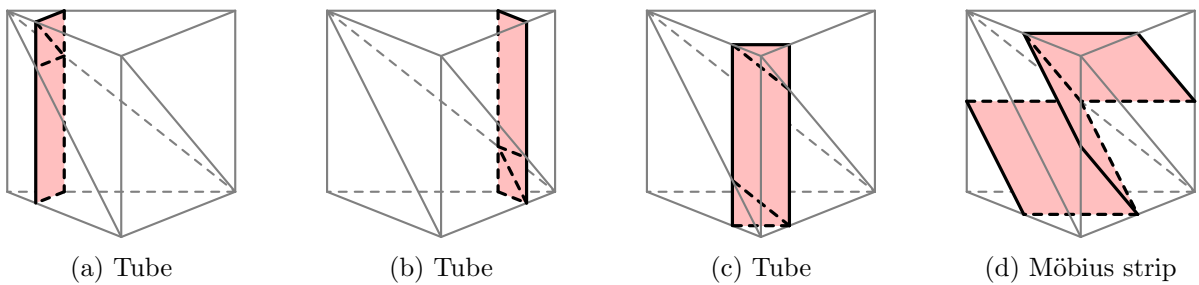
\begin{figure}[htbp]
\centering
	\begin{subfigure}[t]{0.24\textwidth}
	\centering
		\begin{tikzpicture}[scale=0.3]
		
		\draw[thick, dashed, line cap=round, gray] (5,-6) -- (-5,-6);
		\draw[thick, dashed, line cap=round, gray] (-5,2) -- (5,-6);
		
		\draw[thick, line cap=round, gray] (-5,-6) -- (0,-8);
		\draw[thick, line cap=round, gray] (0,-8) -- (5,-6);
		\draw[thick, line cap=round, gray] (-5,-6) -- (-5,2);
		\draw[thick, line cap=round, gray] (5,-6) -- (5,2);
		\draw[thick, line cap=round, gray] (5,2) -- (-5,2);
		
		\fill[pink] (-3.75,1.5) -- (-2.5,2) -- (-2.5,-6) -- (-3.75,-6.5) -- cycle;
		\draw[very thick, line cap=round] (-3.75,1.5) -- (-2.5,2);
		\draw[very thick, dashed, line cap=round] (-2.5,2) -- (-2.5,-6);
		\draw[very thick, dashed, line cap=round] (-2.5,-6) -- (-3.75,-6.5);
		\draw[very thick, line cap=round] (-3.75,-6.5) -- (-3.75,1.5);
		\draw[very thick, dashed, line cap=round] (-3.75,1.5) -- (-2.5,0);
		\draw[very thick, dashed, line cap=round] (-2.5,0) -- (-3.75,-0.5);
		
		\draw[thick, line cap=round, gray] (-5,2) -- (0,-8);
		\draw[thick, line cap=round, gray] (0,-8) -- (0,0);
		\draw[thick, line cap=round, gray] (0,0) -- (5,-6);
		\draw[thick, line cap=round, gray] (-5,2) -- (0,0);
		\draw[thick, line cap=round, gray] (0,0) -- (5,2);
		
		\end{tikzpicture}
	\subcaption{Tube}
	\end{subfigure}
	\begin{subfigure}[t]{0.24\textwidth}
	\centering
		\begin{tikzpicture}[scale=0.3]
		
		\draw[thick, dashed, line cap=round, gray] (5,-6) -- (-5,-6);
		\draw[thick, dashed, line cap=round, gray] (-5,2) -- (5,-6);
		
		\draw[thick, line cap=round, gray] (-5,-6) -- (0,-8);
		\draw[thick, line cap=round, gray] (0,-8) -- (5,-6);
		\draw[thick, line cap=round, gray] (-5,-6) -- (-5,2);
		\draw[thick, line cap=round, gray] (5,-6) -- (5,2);
		\draw[thick, line cap=round, gray] (5,2) -- (-5,2);
		
		\fill[pink] (3.75,1.5) -- (2.5,2) -- (2.5,-6) -- (3.75,-6.5) -- cycle;
		\draw[very thick, line cap=round] (3.75,1.5) -- (2.5,2);
		\draw[very thick, dashed, line cap=round] (2.5,2) -- (2.5,-6);
		\draw[very thick, dashed, line cap=round] (2.5,-6) -- (3.75,-6.5);
		\draw[very thick, line cap=round] (3.75,-6.5) -- (3.75,1.5);
		\draw[very thick, dashed, line cap=round] (3.75,-6.5) -- (2.5,-4);
		\draw[very thick, dashed, line cap=round] (2.5,-4) -- (3.75,-4.5);
		
		\draw[thick, line cap=round, gray] (-5,2) -- (0,-8);
		\draw[thick, line cap=round, gray] (0,-8) -- (0,0);
		\draw[thick, line cap=round, gray] (0,0) -- (5,-6);
		\draw[thick, line cap=round, gray] (-5,2) -- (0,0);
		\draw[thick, line cap=round, gray] (0,0) -- (5,2);
		
		\end{tikzpicture}
	\subcaption{Tube}
	\end{subfigure}
	\begin{subfigure}[t]{0.24\textwidth}
	\centering
		\begin{tikzpicture}[scale=0.3]
		
		\draw[thick, dashed, line cap=round, gray] (5,-6) -- (-5,-6);
		\draw[thick, dashed, line cap=round, gray] (-5,2) -- (5,-6);
		
		\draw[thick, line cap=round, gray] (-5,-6) -- (0,-8);
		\draw[thick, line cap=round, gray] (0,-8) -- (5,-6);
		\draw[thick, line cap=round, gray] (-5,-6) -- (-5,2);
		\draw[thick, line cap=round, gray] (5,-6) -- (5,2);
		\draw[thick, line cap=round, gray] (5,2) -- (-5,2);
		
		\fill[pink] (-1.25,0.5) -- (1.25,0.5) -- (1.25,-7.5) -- (-1.25,-7.5) -- cycle;
		\draw[very thick, line cap=round] (-1.25,0.5) -- (1.25,0.5);
		\draw[very thick, line cap=round] (1.25,0.5) -- (1.25,-7.5);
		\draw[very thick, dashed, line cap=round] (1.25,-7.5) -- (-1.25,-7.5);
		\draw[very thick, line cap=round] (-1.25,-7.5) -- (-1.25,0.5);
		\draw[very thick, dashed, line cap=round] (-1.25,0.5) -- (1.25,-1.5);
		\draw[very thick, dashed, line cap=round] (1.25,-7.5) -- (-1.25,-5.5);
		
		\draw[thick, line cap=round, gray] (-5,2) -- (0,-8);
		\draw[thick, line cap=round, gray] (0,-8) -- (0,0);
		\draw[thick, line cap=round, gray] (0,0) -- (5,-6);
		\draw[thick, line cap=round, gray] (-5,2) -- (0,0);
		\draw[thick, line cap=round, gray] (0,0) -- (5,2);
		
		\end{tikzpicture}
	\subcaption{Tube}
	\end{subfigure}
	\begin{subfigure}[t]{0.24\textwidth}
	\centering
		\begin{tikzpicture}[scale=0.3]
		
		\draw[thick, dashed, line cap=round, gray] (5,-6) -- (-5,-6);
		\draw[thick, dashed, line cap=round, gray] (-5,2) -- (5,-6);
		
		\draw[thick, line cap=round, gray] (-5,-6) -- (0,-8);
		\draw[thick, line cap=round, gray] (0,-8) -- (5,-6);
		\draw[thick, line cap=round, gray] (-5,-6) -- (-5,2);
		\draw[thick, line cap=round, gray] (5,-6) -- (5,2);
		\draw[thick, line cap=round, gray] (5,2) -- (-5,2);
		
		\fill[pink] (2.5,-7) -- (-2.5,-7) -- (-5,-2) -- (0,-2) -- cycle;
		\draw[very thick, dashed, line cap=round] (2.5,-7) -- (-2.5,-7);
		\draw[very thick, line cap=round] (-2.5,-7) -- (-5,-2);
		\draw[very thick, dashed, line cap=round] (-5,-2) -- (0,-2);
		\fill[pink] (0,-2) -- (2.5,-7) -- (0,-4) -- (-2.5,1) -- cycle;
		\draw[very thick, dashed, line cap=round] (0,-2) -- (2.5,-7);
		\draw[very thick, line cap=round] (2.5,-7) -- (0,-4);
		\draw[very thick, line cap=round] (0,-4) -- (-2.5,1);
		\fill[pink] (0,-2) -- (-2.5,1) -- (2.5,1) -- (5,-2) -- cycle;
		\draw[very thick, dashed, line cap=round] (0,-2) -- (-2.5,1);
		\draw[very thick, line cap=round] (-2.5,1) -- (2.5,1);
		\draw[very thick, line cap=round] (2.5,1) -- (5,-2);
		\draw[very thick, dashed, line cap=round] (5,-2) -- (0,-2);
		
		\draw[thick, line cap=round, gray] (-5,2) -- (0,-8);
		\draw[thick, line cap=round, gray] (0,-8) -- (0,0);
		\draw[thick, line cap=round, gray] (0,0) -- (5,-6);
		\draw[thick, line cap=round, gray] (-5,2) -- (0,0);
		\draw[thick, line cap=round, gray] (0,0) -- (5,2);
		
		\end{tikzpicture}
	\subcaption{M\"{o}bius strip}
	\end{subfigure}
\caption{There are four ``types'' of connected spanning central surface in the node gadget: three types of tube, and one type of M\"{o}bius strip.}
\label{fig:nodeSurfs}
\end{figure}

As suggested earlier, our node gadget plays a prominent role in our proof of the following theorem.

\newcommand{\thmSpanningCentralNPhard}{\scap{Connected spanning central surface} (Problem \ref{prbm:spanningCentral}) is $\mathrm{NP}$-complete, and remains $\mathrm{NP}$-complete even if we restrict the input to be an orientable $3$-manifold triangulation.}
\begin{theorem}\label{thm:spanningCentralNPhard}
\thmSpanningCentralNPhard
\end{theorem}

\begin{proof}
We first show that Problem \ref{prbm:spanningCentral} is in $\mathrm{NP}$. Recall that a connected spanning central surface in an $n$-tetrahedron triangulation $\mathcal{T}$ consists of a single choice of elementary disc in each tetrahedron. Such a choice of elementary discs therefore forms a linear-sized certificate. We claim that such a certificate can be verified in polynomial time. Since $\mathcal{T}$ yields at most $6n$ matching equations, it is straightforward to check that the certificate defines a valid normal surface. This is already enough to verify that we have a spanning central surface. To check connectedness, fix one of the $n$ elementary discs $d$, and use a breadth-first search to visit all the elementary discs that are connected to $d$. The normal surface is connected if and only if this search manages to visit all $n$ elementary discs. Clearly, this entire verification process can be done in polynomial time. Thus, Problem \ref{prbm:spanningCentral} is in $\mathrm{NP}$.

To show that Problem \ref{prbm:spanningCentral} is $\mathrm{NP}$-complete, we give a polynomial reduction from \scap{Hamiltonian cycle} (Problem \ref{prbm:HamCycle}). Let $G$ be an arbitrary $3$-regular graph. Our goal is to build a corresponding triangulation $\mathcal{T}_G$, such that $\mathcal{T}_G$ contains a connected spanning central surface if and only if $G$ contains a Hamiltonian cycle. In short, we do this by assigning to each node of $G$ a copy of the node gadget, and then using the arcs of $G$ to determine how we glue together all our copies of the node gadget.

Letting $n$ denote the number of nodes in $G$, we label the nodes in $G$ by $u_0,\ldots,u_{n-1}$. We assign to each node $u_k$ a corresponding copy $\mathcal{N}_k$ of the node gadget. To help us describe how we glue together these copies of the node gadget, we introduce the following notation, for $k\in\{0,\ldots,n-1\}$ and $i\in\{0,1,2\}$.
\begin{itemize}
\item Let $e_{k,i}$ denote axis edge $i$ in $\mathcal{N}_k$.
\item Let $v_{k,i}$ denote boundary vertex $i$ in $\mathcal{N}_k$.
\item Let $A_{k,i}$ denote boundary annulus $i$ in $\mathcal{N}_k$. This annulus is made up of two triangles; let $T_{k,i}^+$ denote the triangle in $A_{k,i}$ that has $e_{k,i+1}$ as one of its edges, and let $T_{k,i}^-$ denote the triangle in $A_{k,i}$ that has $e_{k,i-1}$ as one of its edges.
\end{itemize}
As an additional aid for describing our gluing scheme, we assign a ``direction'' to each axis edge $e_{k,i}$. To do this, imagine an ant walking on the outside of $\mathcal{N}_k$. We choose the direction so that if the ant walks along $e_{k,i}$ in the assigned direction, then it will have triangle $T_{k,i+1}^-$ on its left and triangle $T_{k,i-1}^+$ on its right. This is illustrated in Figure \ref{fig:axisDirections}, where axis edges are drawn using red lines, and the chosen directions are indicated using triple arrowheads.

\begin{figure}[htbp]
\centering
	\begin{subfigure}[t]{0.5\textwidth}
	\centering
		\begin{tikzpicture}[scale=0.8]
		
		\fill[cyan] (-4,2) -- (0,4) -- (0,0) -- cycle;
		\fill[lime] (4,2) -- (0,4) -- (0,0) -- cycle;
		
		\draw[ultra thick, red, midarrow={0.65}{\dirarrows}] (0,0) -- (0,4);
		\draw[very thick] (-4,2) -- (0,0);
		\draw[very thick] (-4,2) -- (0,4);
		\draw[very thick] (4,2) -- (0,0);
		\draw[very thick] (4,2) -- (0,4);
		
		\fill[black] (0,0) circle (0.15);
		\fill[black] (0,4) circle (0.15);
		\fill[black] (-4,2) circle (0.15);
		\fill[black] (4,2) circle (0.15);
		
		\node[above left, anchor=south, rotate=26.35, inner sep=0.5pt] at (-2,3) {major edge};
		\node[below left, anchor=north, rotate=-26.35, inner sep=0.5pt] at (-2,1) {minor edge};
		\node[above right, anchor=south, rotate=-26.35, inner sep=0.5pt] at (2,3) {minor edge};
		\node[below right, anchor=north, rotate=26.35, inner sep=0.5pt] at (2,1) {major edge};
		
		\node[above] at (0,4) {$v_{k,i}$};
		\node[below] at (0,0) {$v_{k,i}$};
		\node[text=blue] at (-1.6,2) {\Large $T_{k,i+1}^-$};
		\node[text=olive] at (1.6,2) {\Large $T_{k,i-1}^+$};
		
		\end{tikzpicture}
	\subcaption{An ant walking along $e_{k,i}$ in the assigned direction will have $T_{k,i+1}^-$ to its left and $T_{k,i-1}^+$ to its right.}
	\end{subfigure}
	\hspace{12pt}
	\begin{subfigure}[t]{0.4\textwidth}
	\centering
		\begin{tikzpicture}[scale=0.4]
		
		\draw[very thick, dashed, line cap=round, midarrow={0.55}{\whitearrow\whitearrow}] (5,-6) -- (-5,-6);
		\draw[very thick, dashed, line cap=round] (-5,2) -- (5,-6);
		
		\draw[very thick, line cap=round, midarrow={0.55}{\whitearrow}] (-5,-6) -- (0,-8);
		\draw[very thick, line cap=round, midarrow={0.55}{\blackarrow}] (0,-8) -- (5,-6);
		
		\draw[very thick, line cap=round, red, midarrow={0.65}{\dirarrows}] (-5,-6) -- (-5,2);
		\draw[very thick, line cap=round] (-5,2) -- (0,-8);
		\draw[very thick, line cap=round, red, midarrow={0.65}{\dirarrows}] (0,-8) -- (0,0);
		\draw[very thick, line cap=round] (0,0) -- (5,-6);
		\draw[very thick, line cap=round, red, midarrow={0.65}{\dirarrows}] (5,-6) -- (5,2);
		
		\draw[very thick, line cap=round, midarrow={0.55}{\whitearrow}] (-5,2) -- (0,0);
		\draw[very thick, line cap=round, midarrow={0.55}{\whitearrow\whitearrow}] (5,2) -- (-5,2);
		\draw[very thick, line cap=round, midarrow={0.55}{\blackarrow}] (0,0) -- (5,2);
		
		\end{tikzpicture}
	\subcaption{The assigned directions for the three axis edges of the node gadget.}
	\end{subfigure}
\caption{We assign a ``direction'' to each axis edge of the node gadget.}
\label{fig:axisDirections}
\end{figure}
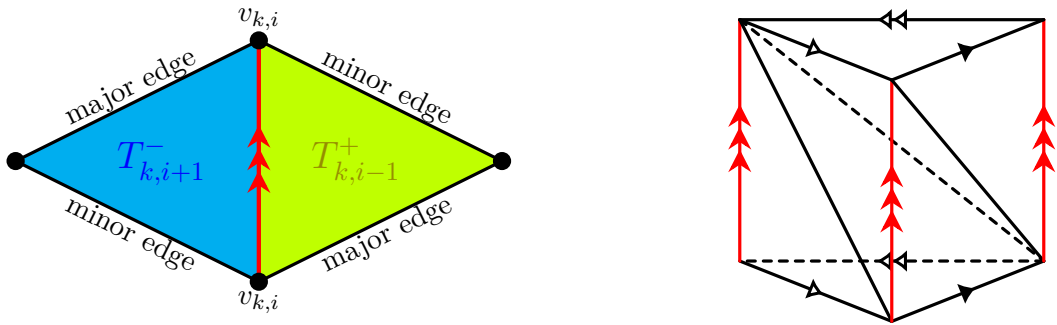

With this in mind, we build $\mathcal{T}_G$ by gluing together pairs of annuli in the following way. For each arc $\{u_k,u_\ell\}$ in $G$, choose an unglued annulus $A_{k,i}$ from $\mathcal{N}_k$ and an unglued annulus $A_{\ell,j}$ from $\mathcal{N}_\ell$.  We glue these two annuli together by:
\begin{itemize}
\item identifying the triangles $T_{k,i}^+$ and $T_{\ell,j}^-$ in such a way that the axis edges $e_{k,i+1}$ and $e_{\ell,j-1}$ get identified with matching directions;
\item identifying the triangles $T_{k,i}^-$ and $T_{\ell,j}^+$ in such a way that the axis edges $e_{k,i-1}$ and $e_{\ell,j+1}$ get identified with matching directions.
\end{itemize}
Observe that as a result of this gluing, the major edge in $A_{k,i}$ is identified with the \emph{minor} edge in $A_{\ell,j}$ to form a new internal edge; similarly, the minor edge in $A_{k,i}$ is identified with the \emph{major} edge in $A_{\ell,j}$ to form a new internal edge. Since $G$ is $3$-regular, it has a total of $\sfrac{3n}{2}$ arcs, so our construction of $\mathcal{T}_G$ can be done in $O(n)$ time.

We need to check that $\mathcal{T}_G$ is actually a ``valid'' triangulation, in the sense that it does not contain any invalid edges. It is clear that none of the internal edges of (the copies of) the node gadget are invalid. It is also clear that the edges of $\mathcal{T}_G$ formed by identifying major and minor edges will never be invalid. Finally, since we always identify axis edges with matching directions, it is impossible for any axis edge to end up being identified with itself in reverse. Thus, $\mathcal{T}_G$ is indeed a valid triangulation.

To show that this construction gives a polynomial reduction from Problem \ref{prbm:HamCycle} to Problem \ref{prbm:spanningCentral}, we prove that there exists a Hamiltonian cycle in $G$ if and only if there exists a connected spanning central surface in $\mathcal{T}_G$.
\begin{itemize}
\item Suppose $G$ contains a Hamiltonian cycle $H$ with node sequence
\[
\pparen{u_{k_0},u_{k_1},\ldots,u_{k_{n-1}},u_{k_0}}.
\]
Thus, for each $m\in\{0,\ldots,n-1\}$, there must exist indices $i_m,j_m\in\{0,1,2\}$ such that the node gadgets $\mathcal{N}_{k_m}$ and $\mathcal{N}_{k_{m+1}}$ (where $m+1$ is calculated modulo $n$) are glued together along the annuli $A_{k_m,i_m}$ and $A_{k_{m+1},j_m}$. Let $T_m$ denote the (type of) spanning central tube in $\mathcal{N}_{k_m}$ that meets the annuli $A_{k_m,i_m}$ and $A_{k_m,j_{m-1}}$.

The tubes $T_0,\ldots,T_{n-1}$ together meet each tetrahedron of $\mathcal{T}_G$ in precisely one elementary disc. So, if we can show that these tubes join together to form a connected normal surface in $\mathcal{T}_G$, then we will in fact have shown that these tubes form a connected spanning central surface. With this in mind, fix any $m\in\{0,\ldots,n-1\}$, and observe that $T_m$ meets the annulus $A_{k_m,i_m}$ in a curve that is ``parallel'' to the axis edges, in the following sense: it never meets the axis edges, and it meets the major and minor edges once each. Similarly, $T_{m+1}$ meets $A_{k_{m+1},j_m}$ in a curve that is ``parallel'' to the axis edges. Since the gluing of $A_{k_m,i_m}$ and $A_{k_{m+1},j_m}$ always identifies axis edges with other axis edges, we see that $T_m$ and $T_{m+1}$ are able to ``match up'' to form a piece of a normal surface. Thus, each tube in the (cyclic) sequence $T_0,\ldots,T_{n-1}$ matches up with the next, which results in a connected spanning central surface in $\mathcal{T}_G$.

\item Conversely, suppose $\mathcal{T}_G$ contains a connected spanning central surface $S$. Such a surface $S$ must pass through each copy of the node gadget in either:
	\begin{itemize}
	\item one of the three possible types of spanning central tube; or
	\item the one possible type of spanning central M\"{o}bius strip.
	\end{itemize}
We claim that $S$ actually cannot meet any of the node gadgets in the one type of M\"{o}bius strip. To see why, first observe that the M\"{o}bius strip meets the axis edges and major edges of the node gadget once each, but never meets any of the minor edges (see Figure \ref{fig:nodeMobius}). With this in mind, suppose for the sake of contradiction that $S$ intersects some $\mathcal{N}_k$, $k\in\{0,\ldots,n-1\}$, in the M\"{o}bius strip. Recall that the annulus $A_{k,0}$ must be glued to some other annulus $A_{\ell,i}$ in such a way that:
	\begin{itemize}
	\item axis edges are identified with other axis edges; and
	\item major edges are identified with \emph{minor} edges.
	\end{itemize}
Since the M\"{o}bius strip in $\mathcal{N}_k$ meets the axis edges and the major edges of $A_{k,0}$, we see that $S$ must intersect $\mathcal{N}_\ell$ in a surface that meets the axis edges and the minor edge of $A_{\ell,i}$. This is impossible, since the three types of tube in $\mathcal{N}_\ell$ never meet any axis edges, and the one type of M\"{o}bius strip in $\mathcal{N}_\ell$ never meets any minor edges.

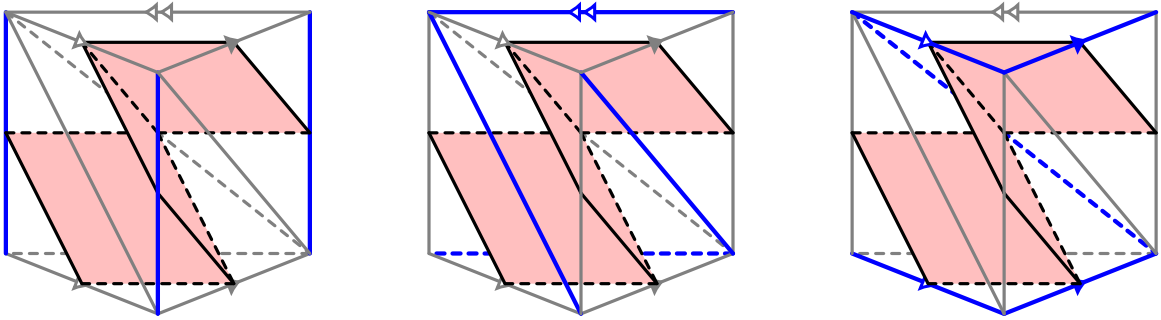
\begin{figure}[htbp]
\centering
	\begin{subfigure}[t]{0.3\textwidth}
	\centering
		\begin{tikzpicture}[scale=0.4]
		
		\draw[very thick, dashed, line cap=round, midarrow={0.55}{\whitearrow\whitearrow}, gray] (5,-6) -- (-5,-6);
		\draw[very thick, dashed, line cap=round, gray] (-5,2) -- (5,-6);
		
		\draw[very thick, line cap=round, midarrow={0.55}{\whitearrow}, gray] (-5,-6) -- (0,-8);
		\draw[very thick, line cap=round, midarrow={0.55}{\blackarrow}, gray] (0,-8) -- (5,-6);
		\draw[ultra thick, line cap=round, blue] (-5,-6) -- (-5,2);
		\draw[ultra thick, line cap=round, blue] (5,-6) -- (5,2);
		\draw[very thick, line cap=round, midarrow={0.55}{\whitearrow\whitearrow}, gray] (5,2) -- (-5,2);
		
		\fill[pink] (2.5,-7) -- (-2.5,-7) -- (-5,-2) -- (0,-2) -- cycle;
		\draw[very thick, dashed, line cap=round] (2.5,-7) -- (-2.5,-7);
		\draw[very thick, line cap=round] (-2.5,-7) -- (-5,-2);
		\draw[very thick, dashed, line cap=round] (-5,-2) -- (0,-2);
		\fill[pink] (0,-2) -- (2.5,-7) -- (0,-4) -- (-2.5,1) -- cycle;
		\draw[very thick, dashed, line cap=round] (0,-2) -- (2.5,-7);
		\draw[very thick, line cap=round] (2.5,-7) -- (0,-4);
		\draw[very thick, line cap=round] (0,-4) -- (-2.5,1);
		\fill[pink] (0,-2) -- (-2.5,1) -- (2.5,1) -- (5,-2) -- cycle;
		\draw[very thick, dashed, line cap=round] (0,-2) -- (-2.5,1);
		\draw[very thick, line cap=round] (-2.5,1) -- (2.5,1);
		\draw[very thick, line cap=round] (2.5,1) -- (5,-2);
		\draw[very thick, dashed, line cap=round] (5,-2) -- (0,-2);
		
		\draw[very thick, line cap=round, gray] (-5,2) -- (0,-8);
		\draw[ultra thick, line cap=round, blue] (0,-8) -- (0,0);
		\draw[very thick, line cap=round, gray] (0,0) -- (5,-6);
		\draw[very thick, line cap=round, midarrow={0.55}{\whitearrow}, gray] (-5,2) -- (0,0);
		\draw[very thick, line cap=round, midarrow={0.55}{\blackarrow}, gray] (0,0) -- (5,2);
		
		\end{tikzpicture}
	\subcaption{The M\"{o}bius strip meets every axis edge.}
	\end{subfigure}
	\hspace{6pt}
	\begin{subfigure}[t]{0.3\textwidth}
	\centering
		\begin{tikzpicture}[scale=0.4]
		
		\draw[ultra thick, dashed, line cap=round, blue, midarrow={0.55}{\whitearrow\whitearrow}] (5,-6) -- (-5,-6);
		\draw[very thick, dashed, line cap=round, gray] (-5,2) -- (5,-6);
		
		\draw[very thick, line cap=round, midarrow={0.55}{\whitearrow}, gray] (-5,-6) -- (0,-8);
		\draw[very thick, line cap=round, midarrow={0.55}{\blackarrow}, gray] (0,-8) -- (5,-6);
		\draw[very thick, line cap=round, gray] (-5,-6) -- (-5,2);
		\draw[very thick, line cap=round, gray] (5,-6) -- (5,2);
		\draw[ultra thick, line cap=round, blue, midarrow={0.55}{\whitearrow\whitearrow}] (5,2) -- (-5,2);
		
		\fill[pink] (2.5,-7) -- (-2.5,-7) -- (-5,-2) -- (0,-2) -- cycle;
		\draw[very thick, dashed, line cap=round] (2.5,-7) -- (-2.5,-7);
		\draw[very thick, line cap=round] (-2.5,-7) -- (-5,-2);
		\draw[very thick, dashed, line cap=round] (-5,-2) -- (0,-2);
		\fill[pink] (0,-2) -- (2.5,-7) -- (0,-4) -- (-2.5,1) -- cycle;
		\draw[very thick, dashed, line cap=round] (0,-2) -- (2.5,-7);
		\draw[very thick, line cap=round] (2.5,-7) -- (0,-4);
		\draw[very thick, line cap=round] (0,-4) -- (-2.5,1);
		\fill[pink] (0,-2) -- (-2.5,1) -- (2.5,1) -- (5,-2) -- cycle;
		\draw[very thick, dashed, line cap=round] (0,-2) -- (-2.5,1);
		\draw[very thick, line cap=round] (-2.5,1) -- (2.5,1);
		\draw[very thick, line cap=round] (2.5,1) -- (5,-2);
		\draw[very thick, dashed, line cap=round] (5,-2) -- (0,-2);
		
		\draw[ultra thick, line cap=round, blue] (-5,2) -- (0,-8);
		\draw[very thick, line cap=round, gray] (0,-8) -- (0,0);
		\draw[ultra thick, line cap=round, blue] (0,0) -- (5,-6);
		\draw[very thick, line cap=round, midarrow={0.55}{\whitearrow}, gray] (-5,2) -- (0,0);
		\draw[very thick, line cap=round, midarrow={0.55}{\blackarrow}, gray] (0,0) -- (5,2);
		
		\end{tikzpicture}
	\subcaption{The M\"{o}bius strip never meets any of the minor edges.}
	\end{subfigure}
	\hspace{6pt}
	\begin{subfigure}[t]{0.3\textwidth}
	\centering
		\begin{tikzpicture}[scale=0.4]
		
		\draw[very thick, dashed, line cap=round, midarrow={0.55}{\whitearrow\whitearrow}, gray] (5,-6) -- (-5,-6);
		\draw[ultra thick, dashed, line cap=round, blue] (-5,2) -- (5,-6);
		
		\draw[ultra thick, line cap=round, blue, midarrow={0.55}{\whitearrow}] (-5,-6) -- (0,-8);
		\draw[ultra thick, line cap=round, blue, midarrow={0.55}{\blackarrow}] (0,-8) -- (5,-6);
		\draw[very thick, line cap=round, gray] (-5,-6) -- (-5,2);
		\draw[very thick, line cap=round, gray] (5,-6) -- (5,2);
		\draw[very thick, line cap=round, midarrow={0.55}{\whitearrow\whitearrow}, gray] (5,2) -- (-5,2);
		
		\fill[pink] (2.5,-7) -- (-2.5,-7) -- (-5,-2) -- (0,-2) -- cycle;
		\draw[very thick, dashed, line cap=round] (2.5,-7) -- (-2.5,-7);
		\draw[very thick, line cap=round] (-2.5,-7) -- (-5,-2);
		\draw[very thick, dashed, line cap=round] (-5,-2) -- (0,-2);
		\fill[pink] (0,-2) -- (2.5,-7) -- (0,-4) -- (-2.5,1) -- cycle;
		\draw[very thick, dashed, line cap=round] (0,-2) -- (2.5,-7);
		\draw[very thick, line cap=round] (2.5,-7) -- (0,-4);
		\draw[very thick, line cap=round] (0,-4) -- (-2.5,1);
		\fill[pink] (0,-2) -- (-2.5,1) -- (2.5,1) -- (5,-2) -- cycle;
		\draw[very thick, dashed, line cap=round] (0,-2) -- (-2.5,1);
		\draw[very thick, line cap=round] (-2.5,1) -- (2.5,1);
		\draw[very thick, line cap=round] (2.5,1) -- (5,-2);
		\draw[very thick, dashed, line cap=round] (5,-2) -- (0,-2);
		
		\draw[very thick, line cap=round, gray] (-5,2) -- (0,-8);
		\draw[very thick, line cap=round, gray] (0,-8) -- (0,0);
		\draw[very thick, line cap=round, gray] (0,0) -- (5,-6);
		\draw[ultra thick, line cap=round, blue, midarrow={0.55}{\whitearrow}] (-5,2) -- (0,0);
		\draw[ultra thick, line cap=round, blue, midarrow={0.55}{\blackarrow}] (0,0) -- (5,2);
		
		\end{tikzpicture}
	\subcaption{The M\"{o}bius strip meets every major edge.}
	\end{subfigure}
\caption{The unwanted M\"{o}bius strip meets every axis edge and every major edge, but never meets any minor edges.}
\label{fig:nodeMobius}
\end{figure}

The upshot is that $S$ must be built entirely out of tubes. So, for each $k\in\{0,\ldots,n-1\}$, let $T_k$ denote the tube in which $S$ meets $\mathcal{N}_k$. Because each tube $T_k$ meets two of the three boundary annuli of $\mathcal{N}_k$, we can think of $T_k$ as a path through $\mathcal{N}_k$ that joins these two annuli. Since the tubes $T_0,\ldots,T_{n-1}$ must all ``match up'' to form the connected spanning central surface $S$, we can therefore think of $S$ as a path that visits each node gadget in $\mathcal{T}_G$ exactly once, before returning to the beginning. This corresponds to a Hamiltonian cycle in $G$.
\end{itemize}
Altogether, we have given a polynomial reduction from \scap{Hamiltonian cycle} (Problem \ref{prbm:HamCycle}) to \scap{connected spanning central surface} (Problem \ref{prbm:spanningCentral}), which shows that Problem \ref{prbm:spanningCentral} is $\mathrm{NP}$-complete.

We finish this proof by showing that $\mathcal{T}_G$ always represents an orientable $3$-manifold, which implies that Problem \ref{prbm:spanningCentral} remains $\mathrm{NP}$-complete even if we restrict the input to be an orientable $3$-manifold triangulation. Orientability of $\mathcal{T}_G$ follows immediately from the following observation: in our construction of $\mathcal{T}_G$, we always identified triangular faces with opposite orientations. Thus it suffices to show that $\mathcal{T}_G$ is a $3$-manifold triangulation, and we can do this by showing that the link of every vertex in $\mathcal{T}_G$ is a sphere.

Since every vertex in $\mathcal{T}_G$ is formed by identifying some number of boundary vertices from the node gadgets $\mathcal{N}_0,\ldots,\mathcal{N}_{n-1}$, we start by taking a closer look at the surfaces that form the links of the boundary vertices of the node gadget. These surfaces, all of which are discs, are illustrated in Figure \ref{fig:nodeVertexLinks}. For each $k\in\{0,\ldots,n-1\}$ and $i\in\{0,1,2\}$, let $D_{k,i}$ denote the disc that forms the link of vertex $v_{k,i}$. As shown in Figure \ref{fig:discBdryPieces}, the disc $D_{k,i}$ meets two annuli in $\mathcal{N}_k$, thus dividing its boundary into two pieces: a curve $c_{k,i}^+=\partial D_{k,i}\cap A_{k,i+1}$, and a curve $c_{k,i}^-=\partial D_{k,i}\cap A_{k,i-1}$. We assign directions to these curves so that if an ant on the outside of the node gadget walks in the assigned direction, then the vertex $v_{k,i}$ always remains on the ant's right-hand side. These assigned directions are indicated using triple arrowheads in Figure \ref{fig:discBdryPieces}.

\begin{figure}[htbp]
\centering
	\begin{subfigure}[t]{0.32\textwidth}
	\centering
		\begin{tikzpicture}[scale=0.4]
		
		\draw[thick, dashed, line cap=round, gray] (5,-6) -- (-5,-6);
		\draw[thick, dashed, line cap=round, gray] (-5,2) -- (5,-6);
		
		\draw[thick, line cap=round, gray] (-5,-6) -- (0,-8);
		\draw[thick, line cap=round, gray] (0,-8) -- (5,-6);
		\draw[thick, line cap=round, gray] (-5,-6) -- (-5,2);
		\draw[thick, line cap=round, gray] (5,-6) -- (5,2);
		\draw[thick, line cap=round, gray] (5,2) -- (-5,2);
		
		\fill[pink] (-3,-6.8) -- (-5,-2.8) -- (-1,-6) -- cycle;
		\fill[pink] (-3,1.2) -- (-5,-1.2) -- (-1,2) -- cycle;
		\draw[very thick, line cap=round] (-3,-6.8) -- (-5,-2.8);
		\draw[very thick, dashed, line cap=round] (-1,-6) -- (-5,-2.8) ;
		\draw[very thick, dashed, line cap=round] (-1,-6) -- (-3,-6.8);
		\draw[very thick, line cap=round] (-3,1.2) -- (-5,-1.2);
		\draw[very thick, dashed, line cap=round] (-1,2) -- (-3,0.4);
		\draw[very thick, dashed, line cap=round] (-5,-1.2) -- (-3,0.4);
		\draw[very thick, line cap=round] (-1,2) -- (-3,1.2);
		\draw[very thick, dotted, line cap=round] (-3,0.4) -- (-4,0);
		\draw[very thick, dotted, line cap=round] (-3,0.4) -- (-3,1.2);
		
		\draw[thick, line cap=round, gray] (-5,2) -- (0,-8);
		\draw[thick, line cap=round, gray] (0,-8) -- (0,0);
		\draw[thick, line cap=round, gray] (0,0) -- (5,-6);
		\draw[thick, line cap=round, gray] (-5,2) -- (0,0);
		\draw[thick, line cap=round, gray] (0,0) -- (5,2);
		
		\end{tikzpicture}
	\subcaption{The link of boundary vertex $0$.}
	\end{subfigure}
	\begin{subfigure}[t]{0.32\textwidth}
	\centering
		\begin{tikzpicture}[scale=0.4]
		
		\draw[thick, dashed, line cap=round, gray] (5,-6) -- (-5,-6);
		\draw[thick, dashed, line cap=round, gray] (-5,2) -- (5,-6);
		
		\draw[thick, line cap=round, gray] (-5,-6) -- (0,-8);
		\draw[thick, line cap=round, gray] (0,-8) -- (5,-6);
		\draw[thick, line cap=round, gray] (-5,-6) -- (-5,2);
		\draw[thick, line cap=round, gray] (5,-6) -- (5,2);
		\draw[thick, line cap=round, gray] (5,2) -- (-5,2);
		
		\fill[pink] (-2,-7.2) -- (2,-7.2) -- (0,-4.8) -- cycle;
		\fill[pink] (-2,0.8) -- (2,0.8) -- (0,-3.2) -- cycle;
		\draw[very thick, dashed, line cap=round] (-2,-7.2) -- (2,-7.2);
		\draw[very thick, line cap=round] (2,-7.2) -- (0,-4.8);
		\draw[very thick, line cap=round] (0,-4.8) -- (-2,-7.2);
		\draw[very thick, dotted, line cap=round] (2,-7.2) -- (-1,-6);
		\draw[very thick, line cap=round] (-2,0.8) -- (2,0.8);
		\draw[very thick, line cap=round] (2,0.8) -- (0,-3.2);
		\draw[very thick, line cap=round] (0,-3.2) -- (-2,0.8);
		\draw[very thick, dotted, line cap=round] (-2,0.8) -- (1,-1.2);
		
		\draw[thick, line cap=round, gray] (-5,2) -- (0,-8);
		\draw[thick, line cap=round, gray] (0,-8) -- (0,0);
		\draw[thick, line cap=round, gray] (0,0) -- (5,-6);
		\draw[thick, line cap=round, gray] (-5,2) -- (0,0);
		\draw[thick, line cap=round, gray] (0,0) -- (5,2);
		
		\end{tikzpicture}
	\subcaption{The link of boundary vertex $1$.}
	\end{subfigure}
	\begin{subfigure}[t]{0.32\textwidth}
	\centering
		\begin{tikzpicture}[scale=0.4]
		
		\draw[thick, dashed, line cap=round, gray] (5,-6) -- (-5,-6);
		\draw[thick, dashed, line cap=round, gray] (-5,2) -- (5,-6);
		
		\draw[thick, line cap=round, gray] (-5,-6) -- (0,-8);
		\draw[thick, line cap=round, gray] (0,-8) -- (5,-6);
		\draw[thick, line cap=round, gray] (-5,-6) -- (-5,2);
		\draw[thick, line cap=round, gray] (5,-6) -- (5,2);
		\draw[thick, line cap=round, gray] (5,2) -- (-5,2);
		
		\fill[pink] (3,-6.8) -- (5,-2.8) -- (1,-6) -- cycle;
		\fill[pink] (3,1.2) -- (5,-1.2) -- (1,2) -- cycle;
		\draw[very thick, line cap=round] (3,-6.8) -- (5,-2.8);
		\draw[very thick, dashed, line cap=round] (5,-2.8) -- (1,-6);
		\draw[very thick, dashed, line cap=round] (1,-6) -- (3,-6.8);
		\draw[very thick, dotted, line cap=round] (3,-4.4) -- (4,-4.8);
		\draw[very thick, dotted, line cap=round] (3,-4.4) -- (3,-6.82);
		\draw[very thick, line cap=round] (3,1.2) -- (5,-1.2);
		\draw[very thick, dashed, line cap=round] (1,2) -- (3,0.4);
		\draw[very thick, dashed, line cap=round] (5,-1.2) -- (3,0.4);
		\draw[very thick, line cap=round] (1,2) -- (3,1.2);
		
		\draw[thick, line cap=round, gray] (-5,2) -- (0,-8);
		\draw[thick, line cap=round, gray] (0,-8) -- (0,0);
		\draw[thick, line cap=round, gray] (0,0) -- (5,-6);
		\draw[thick, line cap=round, gray] (-5,2) -- (0,0);
		\draw[thick, line cap=round, gray] (0,0) -- (5,2);
		
		\end{tikzpicture}
	\subcaption{The link of boundary vertex $2$.}
	\end{subfigure}
\caption{The three vertex-linking discs in the node gadget.}
\label{fig:nodeVertexLinks}
\end{figure}
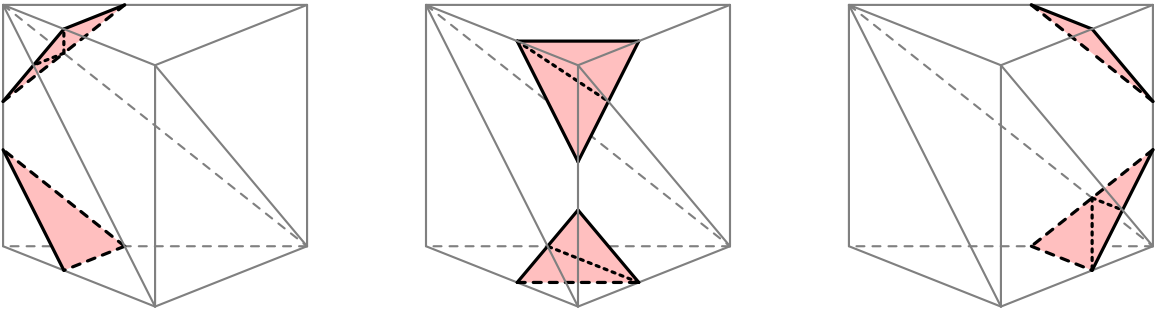

\begin{figure}[htbp]
\centering
	\begin{tikzpicture}[scale=0.45]
	
	\draw[thick, dashed, line cap=round, gray] (5,-6) -- (-5,-6);
	\draw[thick, dashed, line cap=round, gray] (-5,2) -- (5,-6);
	
	\draw[thick, line cap=round, gray] (-5,-6) -- (0,-8);
	\draw[thick, line cap=round, gray] (0,-8) -- (5,-6);
	\draw[thick, line cap=round, gray] (-5,-6) -- (-5,2);
	\draw[thick, line cap=round, gray] (5,-6) -- (5,2);
	\draw[thick, line cap=round, gray] (5,2) -- (-5,2);
	
	\fill[pink] (-2,-7.2) -- (2,-7.2) -- (0,-4.8) -- cycle;
	\fill[pink] (-2,0.8) -- (2,0.8) -- (0,-3.2) -- cycle;
	\draw[very thick, dashed, line cap=round] (-2,-7.2) -- (2,-7.2);
	\draw[ultra thick, line cap=round, blue, midarrow={0.8}{\dirarrows}] (0,-4.8) -- (2,-7.2);
	\draw[ultra thick, line cap=round, red, midarrow={0.8}{\dirarrows}] (-2,-7.2) -- (0,-4.8);
	\draw[very thick, dotted, line cap=round] (2,-7.2) -- (-1,-6);
	\draw[very thick, line cap=round] (-2,0.8) -- (2,0.8);
	\draw[ultra thick, line cap=round, blue, midarrow={0.7}{\dirarrows}] (2,0.8) -- (0,-3.2);
	\draw[ultra thick, line cap=round, red, midarrow={0.7}{\dirarrows}] (0,-3.2) -- (-2,0.8);
	\draw[very thick, dotted, line cap=round] (-2,0.8) -- (1,-1.2);
	
	\draw[thick, line cap=round, gray] (-5,2) -- (0,-8);
	\draw[thick, line cap=round, gray] (0,-8) -- (0,0);
	\draw[thick, line cap=round, gray] (0,0) -- (5,-6);
	\draw[thick, line cap=round, gray] (-5,2) -- (0,0);
	\draw[thick, line cap=round, gray] (0,0) -- (5,2);
	
	\end{tikzpicture}
\caption{The boundary of the disc $D_{k,i}$ is divided into the curves $c_{k,i}^+$ (red) and $c_{k,i}^-$ (blue). An ant walking along the boundary of $D_{k,i}$ in the assigned direction will always have vertex $v_{k,i}$ to its right.}
\label{fig:discBdryPieces}
\end{figure}
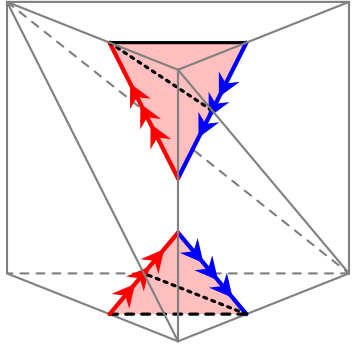

With all this in mind, consider any particular vertex $v$ in $\mathcal{T}_G$. Fix $k_0\in\{0,\ldots,n-1\}$ and $i_0\in\{0,1,2\}$ such that $v_{k_0,i_0}$ is one of the vertices that has been identified to form $v$. By construction of $\mathcal{T}_G$, the annulus $A_{k_0,i_0-1}$ is identified with some other annulus $A_{k_1,i_1+1}$, which causes the vertices $v_{k_0,i_0}$ and $v_{k_1,i_1}$ to be identified. Observe that as a result of this identification, the vertex-linking discs $D_{k_0,i_0}$ and $D_{k_1,i_1}$ get glued together in such a way that the curves $c_{k_0,i_0}^-$ and $c_{k_1,i_1}^+$ are identified with \emph{opposite} directions.

Continuing inductively, for each $m\in\{0,1,\ldots\}$, the annulus $A_{k_m,i_m-1}$ is identified with some other annulus $A_{k_{m+1},i_{m+1}-1}$, which causes the vertices $v_{k_m,i_m}$ and $v_{k_{m+1},i_{m+1}}$ to be identified. This identification subsequently causes the vertex-linking discs $D_{k_m,i_m}$ and $D_{k_{m+1},i_{m+1}}$ to be glued together in such a way that the curves $c_{k_m,i_m}^-$ and $c_{k_{m+1},i_{m+1}}^+$ are identified with opposing directions. Altogether, we get a sequence $D_{k_0,i_0},D_{k_1,i_1},\ldots$ of discs that are glued together in the manner shown in Figure \ref{fig:gluingVertLinks}.

\begin{figure}[htbp]
\centering
	\begin{tikzpicture}[scale=0.4]
	
	
	\draw[thick, <->] (3,4) -- (5,4);
	\draw[thick, <->] (11,4) -- (13,4);
	\node at (20,4) {\LARGE $\cdots$};
	\draw[ultra thick, -{Latex}] (22,4) -- (26.5,4);
	\node at (36,4) {\LARGE $\cdots$};
	
	
	
	\fill[pink] (0,0) to[out=135, in=270] (-2,3) -- (-2,5) to[out=90, in=225] (0,8) to[out=-45, in=90] (2,5) -- (2,3) to[out=270, in=45] (0,0);
	\node[text=red] at (0,4) {\Large $D_{k_0,i_0}$};
	
	\draw[ultra thick, line cap=round, midarrow={0.59}{\dirarrows}] (0,0) to[out=135, in=270] (-2,3) -- (-2,5) to[out=90, in=225] (0,8);
	\draw[ultra thick, line cap=round, midarrow={0.59}{\dirarrows}] (0,8) to[out=-45, in=90] (2,5) -- (2,3) to[out=270, in=45] (0,0);
	
	
	\begin{scope}[shift={(8,0)}]
	
	\fill[cyan] (0,0) to[out=135, in=270] (-2,3) -- (-2,5) to[out=90, in=225] (0,8) to[out=-45, in=90] (2,5) -- (2,3) to[out=270, in=45] (0,0);
	\node[text=blue] at (0,4) {\Large $D_{k_1,i_1}$};
	
	\draw[ultra thick, line cap=round, midarrow={0.59}{\dirarrows}] (0,0) to[out=135, in=270] (-2,3) -- (-2,5) to[out=90, in=225] (0,8);
	\draw[ultra thick, line cap=round, midarrow={0.59}{\dirarrows}] (0,8) to[out=-45, in=90] (2,5) -- (2,3) to[out=270, in=45] (0,0);
	
	\end{scope}
	
	\begin{scope}[shift={(16,0)}]
	
	\fill[lime] (0,0) to[out=135, in=270] (-2,3) -- (-2,5) to[out=90, in=225] (0,8) to[out=-45, in=90] (2,5) -- (2,3) to[out=270, in=45] (0,0);
	\node[text=olive] at (0,4) {\Large $D_{k_2,i_2}$};
	
	\draw[ultra thick, line cap=round, midarrow={0.59}{\dirarrows}] (0,0) to[out=135, in=270] (-2,3) -- (-2,5) to[out=90, in=225] (0,8);
	\draw[ultra thick, line cap=round, midarrow={0.59}{\dirarrows}] (0,8) to[out=-45, in=90] (2,5) -- (2,3) to[out=270, in=45] (0,0);
	
	\end{scope}
	
	\begin{scope}[shift={(31,0)}]
	
	\fill[pink] (0,0) to[out=150, in=270] (-3,3) -- (-3,5) to[out=90, in=210] (0,8) to[out=240, in=90] (-1,5) -- (-1,3) to[out=270, in=120] cycle;
	\fill[cyan] (0,0) to[out=120, in=270] (-1,3) -- (-1,5) to[out=90, in=240] (0,8) to[out=-60, in=90] (1,5) -- (1,3) to[out=270, in=60] cycle;
	\fill[lime] (0,0) to[out=60, in=270] (1,3) -- (1,5) to[out=90, in=-60] (0,8) to[out=-30, in=90] (3,5) -- (3,3) to[out=270, in=30] (0,0);
	
	\draw[ultra thick, line cap=round, midarrow={0.59}{\dirarrows}] (0,0) to[out=150, in=270] (-3,3) -- (-3,5) to[out=90, in=210] (0,8);
	\draw[ultra thick, line cap=round, midarrow={0.59}{\dirarrows}] (0,8) to[out=-30, in=90] (3,5) -- (3,3) to[out=270, in=30] (0,0);
	\draw[ultra thick, line cap=round] (0,0) to[out=120, in=270] (-1,3) -- (-1,5) to[out=90, in=240] (0,8);
	\draw[ultra thick, line cap=round] (0,0) to[out=60, in=270] (1,3) -- (1,5) to[out=90, in=-60] (0,8);
	
	\end{scope}
	
	\end{tikzpicture}
\caption{Each disc in the sequence $D_{k_0,i_0},D_{k_1,i_1},\ldots$ is glued to the next, in such a way that for all $m\in\{0,1,\ldots\}$, the curves $c_{k_m,i_m}^-$ and $c_{k_{m+1},i_{m+1}}^+$ are identified with opposing directions.}
\label{fig:gluingVertLinks}
\end{figure}
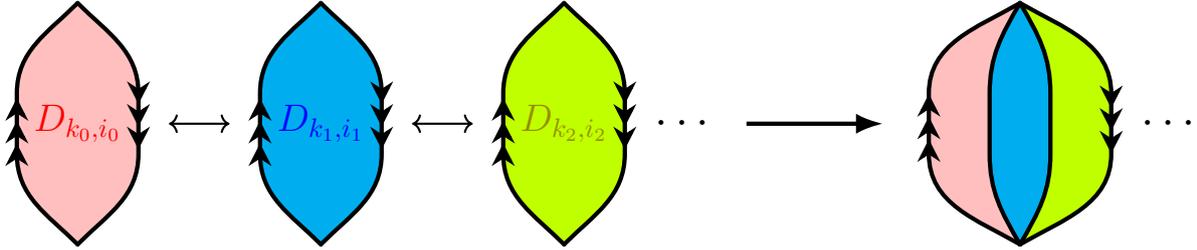

Eventually, there must be some $M$ such that $A_{k_M,i_M-1}$ is identified with $A_{k_0,i_0+1}$. The link of $v$ is therefore precisely the surface formed by gluing together all the discs $D_{k_0,i_0},\ldots,D_{k_M,i_M}$. Since the curves $c_{k_M,i_M}^-$ and $c_{k_0,i_0}^+$ are identified with opposing directions, observe that this vertex-linking surface is a sphere.

The same argument applies to every vertex of $\mathcal{T}_G$, so we conclude that $\mathcal{T}_G$ is a $3$-manifold triangulation. As we mentioned earlier, this shows that Problem \ref{prbm:spanningCentral} remains $\mathrm{NP}$-complete even if we restrict the input to be an orientable $3$-manifold triangulation.
\end{proof}

\section{Discussion}\label{sec:discussion}
As discussed in section \ref{sec:intro}, our underlying motivation for studying Problems \ref{prbm:absNormConOpt}, \ref{prbm:splittingSurf} and \ref{prbm:spanningCentral} was to gain insight into the computational complexity of the problem of finding a non-trivial normal sphere or disc. In particular, our proofs of $\mathrm{NP}$-hardness for Problems \ref{prbm:absNormConOpt} and \ref{prbm:spanningCentral} illustrate two possible approaches for proving that finding a non-trivial normal sphere or disc is $\mathrm{NP}$-hard; indeed, our two $\mathrm{NP}$-hardness proofs are by-products of our investigations into each of these approaches. Here, we briefly discuss some of the obstacles that we would need to overcome to make either of these approaches successful.

First, recall that Problem \ref{prbm:absNormConOpt} is an abstraction of the concrete problem of finding a non-trivial normal sphere or disc. As discussed in section \ref{sec:absNormConOpt}, our proof that the abstract problem is $\mathrm{NP}$-hard relies heavily on constructions that violate some of the geometric and topological restrictions that are inherent to the concrete problem. One potential strategy for proving that the concrete problem is also $\mathrm{NP}$-hard would therefore be to come up with new constructions that do not violate these geometric and topological restrictions. However, with these much stricter requirements, we can reasonably expect that any such constructions would need to be much more intricate than what we used in section \ref{sec:absNormConOpt}.

Alternatively, we could try to prove that finding a non-trivial normal sphere or disc is $\mathrm{NP}$-hard by coming up with a ``gadget proof'', similar to how we proved that finding a connected spanning central surface is $\mathrm{NP}$-hard. However, spanning central surfaces are much easier to work with in this context because they are ``determined locally'', in the following sense: a normal surface intersecting several gadgets is a spanning central surface (within the whole triangulation) if and only if the intersection with each gadget is itself a spanning central surface (within the gadget). This allows us to construct a triangulation with relatively little concern about what it will end up looking like globally, since we only ever need to analyse what happens locally. In contrast, spheres and discs are \emph{not} determined locally. We know that if a sphere or disc intersects multiple gadgets, then the intersection with each gadget must itself be a surface with genus $0$. But the converse does not hold: it is very easy to get a surface of positive genus from piecing together a collection of genus $0$ surfaces. For this reason, a ``gadget proof'' for the $\mathrm{NP}$-hardness of finding a non-trivial normal sphere or disc would probably need to be significantly more complicated than the one we used in section \ref{sec:splitting}.

\bibliography{HardnessNormalRefs}

\end{document}